\documentclass[a4paper]{article}
\usepackage{amsmath}
\usepackage[english]{babel}
\usepackage{latexsym}
\usepackage{amssymb}
\usepackage{amscd}
\usepackage{amsgen,amstext,amsbsy,amsopn}
\usepackage{math rsfs}

\usepackage{amsthm,epsfig,graphicx,graphics}
\usepackage[latin1]{inputenc}
\usepackage{xspace}
\usepackage{amsxtra}
\usepackage{bm}

\usepackage{color}

\newcommand{\version}{\today}
\numberwithin{equation}{section}
\newcommand{\bdm}{\begin{displaymath}}
\newcommand{\edm}{\end{displaymath}}
\newcommand{\bdn}{\begin{eqnarray}}
\newcommand{\edn}{\end{eqnarray}}
\newcommand{\bay}{\begin{array}{c}}
\newcommand{\eay}{\end{array}}
\newcommand{\ben}{\begin{enumerate}}
\newcommand{\een}{\end{enumerate}}
\newcommand{\beq}{\begin{equation}}
\newcommand{\eeq}{\end{equation}}
\newcommand{\bml}[1]{\begin{multline} #1 \end{multline}}
\newcommand{\bmln}[1]{\begin{multline*} #1 \end{multline*}}

\newcommand{\disp}{\displaystyle}
\newcommand{\tx}{\textstyle}

\newcommand{\lf}{\left}
\newcommand{\ri}{\right}
\newcommand{\bra}[1]{\lf\langle #1\ri|}
\newcommand{\ket}[1]{\lf|#1 \ri\rangle}
\newcommand{\braket}[2]{\lf\langle #1|#2 \ri\rangle}

\newcommand{\xv}{\mathbf{x}}

\newcommand{\rv}{\mathbf{r}}

\newcommand{\sv}{\mathbf{s}}

\newcommand{\aavo}{\mathbf{A}_{\omega}}
\newcommand{\aavoo}{\mathbf{A}_{\Omega}}
\newcommand{\taavoo}{\tilde{\mathbf{A}}}
\newcommand{\diff}{\mathrm{d}}
\newcommand{\eps}{\varepsilon}

\newcommand{\ete}{\eta_{\varepsilon}}
\newcommand{\delte}{\delta_{\eps}}

\newcommand{\omegac}{\omega_{\mathrm{c}}}
\newcommand{\Omegac}{\Omega_{\mathrm{c}}}

\newcommand{\Ao}{\A_{\omega}}
\newcommand{\Aoo}{\A_{\Omega}}

\newcommand{\Ofirst}{\Omega_{\mathrm{c_1}}}
\newcommand{\Osec}{\Omega_{\mathrm{c_2}}}
\newcommand{\Othird}{\Omega_{\mathrm{c_3}}}

\newcommand{\ba}{\mathcal{B}}

\newcommand{\fsup}{f_{\mathrm{sup}}}

\newcommand{\fsub}{f_{\mathrm{sub}}}

\newcommand{\gpf}{\mathcal{E}^{\mathrm{GP}}}
\newcommand{\gpe}{E^{\mathrm{GP}}}
\newcommand{\gpm}{\psi^{\mathrm{GP}}}
\newcommand{\chemGP}{\mu^{\mathrm{GP}}}
\newcommand{\gpd}{\rho^{\mathrm{GP}}}
\newcommand{\tgpd}{\tilde\rho^{\mathrm{GP}}}
\newcommand{\gpdom}{\mathscr{D}^{\mathrm{GP}}}

\newcommand{\ggpf}{\mathcal{E}^{\mathrm{GP}}_{\mathrm{phys}}}

\newcommand{\gpfo}{\gpf_{\omega}}
\newcommand{\gpeo}{E^{\mathrm{GP}}_{\omega}}
\newcommand{\gpmo}{\psi^{\mathrm{GP}}_{\omega}}

\newcommand{\chemo}{\mu^{\mathrm{GP}}_{\omega}}

\newcommand{\gpfoo}{\gpf_{\Omega}}
\newcommand{\gpeoo}{E^{\mathrm{GP}}_{\Omega}}
\newcommand{\gpmoo}{\psi^{\mathrm{GP}}_{\Omega}}

\newcommand{\chemoo}{\mu^{\mathrm{GP}}_{\Omega}}

\newcommand{\hgpdom}{\hat{\mathscr{D}}^{\mathrm{GP}}}

\newcommand{\hgpfo}{\hat{\mathcal{E}}^{\mathrm{GP}}_{\omega}}
\newcommand{\hgpeo}{\hat{E}^{\mathrm{GP}}_{\omega}}
\newcommand{\hchemo}{\hat{\mu}_{\omega}^{\mathrm{GP}}}

\newcommand{\hgpfoo}{\hat{\mathcal{E}}_{\Omega}^{\mathrm{GP}}}
\newcommand{\hgpeoo}{\hat{E}_{\Omega}^{\mathrm{GP}}}
\newcommand{\hchemoo}{\hat{\mu}_{\Omega}^{\mathrm{GP}}}

\newcommand{\gvf}{\mathcal{E}^{\mathrm{gv}}}
\newcommand{\gve}{E^{\mathrm{gv}}}

\newcommand{\gvm}{g_{\mathrm{gv}}}
\newcommand{\gvchem}{\mu^{\mathrm{gv}}}
\newcommand{\gvdom}{\mathscr{D}^{\mathrm{gv}}}

\newcommand{\tgvm}{\tilde{g}_{\mathrm{gv}}}

\newcommand{\gvfe}{\mathcal{E}^{\mathrm{gv}}_{\ete}}
\newcommand{\gvee}{E^{\mathrm{gv}}_{\ete}}

\newcommand{\gvme}{g_{\ete}}
\newcommand{\gvcheme}{\mu^{\mathrm{gv}}_{\ete}}
\newcommand{\gvdome}{\mathscr{D}^{\mathrm{gv}}_{\ete}}
\newcommand{\gvmaxe}{x_{\mathrm{max}}}

\newcommand{\tgvme}{\tilde{g}_{\ete}}

\newcommand{\gvmet}{\tilde{g}_{\ete}}

\newcommand{\gveet}{\tilde{E}^{\mathrm{gv}}_{\ete}}
\newcommand{\gvfet}{\tilde{\mathcal{E}}^{\mathrm{gv}}_{\ete}}

\newcommand{\Ee}{\mathcal{E}_{\ete}}
\newcommand{\Fe}{\mathcal{F}_{\ete}}
\newcommand{\Es}{\mathcal{E}_{\sqrt{\ete}}}

\newcommand{\gsym}{g_{\mathrm{sym}}}
\newcommand{\esym}{E^{\mathrm{sym}}}
\newcommand{\esymf}{\mathcal{E}^{\mathrm{sym}}}
\newcommand{\domsym}{\mathscr{D}_{\mathrm{sym}}}
\newcommand{\chemsym}{\mu^{\mathrm{sym}}}

\newcommand{\gsymt}{\tilde{g}_{\mathrm{sym}}}
\newcommand{\esymt}{\tilde{E}^{\mathrm{sym}}}
\newcommand{\esymft}{\tilde{\mathcal{E}}^{\mathrm{sym}}}

\newcommand{\hosc}{H_{\mathrm{osc}}}
\newcommand{\gosc}{g_{\mathrm{osc}}}

\newcommand{\curl}{{\rm curl}}

\newcommand{\set}{\mathcal{S}}

\newcommand{\tff}{\mathcal{E}^{\mathrm{TF}}}

\newcommand{\tfe}{E^{\mathrm{TF}}}
\newcommand{\tfm}{\rho^{\mathrm{TF}}}
\newcommand{\tfdom}{\mathscr{D}^{\mathrm{TF}}}

\newcommand{\tffo}{\mathcal{E}_{\omega}^{\mathrm{TF}}}
\newcommand{\tfeo}{E_{\omega}^{\mathrm{TF}}}
\newcommand{\tfmo}{\rho_{\omega}^{\mathrm{TF}}}
\newcommand{\tfchemo}{\mu^{\mathrm{TF}}_{\omega}}
\newcommand{\tfsuppo}{\mathrm{supp}\lf(\tfmo\ri)}

\newcommand{\tffoo}{\mathcal{E}_{\Omega}^{\mathrm{TF}}}
\newcommand{\tfeoo}{E_{\Omega}^{\mathrm{TF}}}
\newcommand{\tfmoo}{\rho_{\Omega}^{\mathrm{TF}}}
\newcommand{\tfchemoo}{\mu^{\mathrm{TF}}_{\Omega}}
\newcommand{\tfsuppoo}{\mathrm{supp}\lf(\tfmoo\ri)}

\newcommand{\xh}{x_{\mathrm{hole}}}

\newcommand{\rmax}{R_{\mathrm{m}}}
\newcommand{\reps}{R_{\eps}}

\newcommand{\xin}{x_{\mathrm{in}}}
\newcommand{\xout}{x_{\mathrm{out}}}
\newcommand{\xinout}{x_{\mathrm{in}/\mathrm{out}}}

\newcommand{\lamsym}{\lambda^{\rm  sym}}
\newcommand{\hsym}{H^{\rm sym}}

\newcommand{\ab}{\mathcal{A}_{\mathrm{bulk}}}

\newcommand{\anne}{\mathcal{A}_{\ete}}
\newcommand{\annet}{\tilde{\mathcal{A}}_{\ete}}
\newcommand{\anns}{\mathcal{A}_{\sqrt{\ete}}}

\newcommand{\tfsupp}{\mathrm{supp}\lf(\tfm\ri)}

\newcommand{\trial}{\psi_{\mathrm{trial}}}

\newcommand{\latt}{\mathcal{L}}

\newcommand{\cell}{\mathcal{Q}}
\newcommand{\celli}{\mathcal{Q}^i_{}}

\newcommand{\tcelli}{\tilde{\mathcal{Q}}_i}

\newcommand{\magnp}{\mathbf{A}_{\mathrm{rot}}}
\newcommand{\rmagnp}{\tilde{\mathbf{B}}}
\newcommand{\ftrial}{f_{\mathrm{trial}}}

\newcommand{\trialu}{u_{\mathrm{trial}}}
\newcommand{\half}{\hbox{$\frac12$}}
\newcommand{\salf}{\hbox{$\frac{1}{s}$}}

\newcommand{\hex}{h_{\mathrm{ex}}}

\newcommand{\Orot}{\Omega_{\mathrm{rot}}}
\newcommand{\Orotv}{{\bm \Omega}_{\mathrm{rot}}}
\newcommand{\Oeff}{\Omega_{\mathrm{eff}}}

\newcommand{\Oosc}{\Omega_{\mathrm{osc}}}

\newcommand{\into}{\lfloor \Omega \rfloor}

\newcommand{\Z}{\mathbb{Z}}
\newcommand{\R}{\mathbb{R}}
\newcommand{\N}{\mathbb{N}}

\newcommand{\D}{\mathcal{D}}
\newcommand{\F}{\mathcal{F}}

\newcommand{\A}{\mathcal{A}}
\newcommand{\B}{\mathcal{B}}

\newcommand{\Q}{\mathcal{Q}}
\newcommand{\OO}{\mathcal{O}}

\newcommand{\al}{\alpha}

\newcommand{\ep}{\varepsilon}

\newcommand{\Om}{\Omega}
\newcommand{\om}{\omega}

\newcommand{\dd}{\partial}

\newcommand{\muv}{\mu}
\newcommand{\nablap}{\nabla^{\perp}}

\newcommand{\At}{\tilde{\A}}
\newcommand{\Ac}{\A_{c}}

\newtheorem{teo}{Theorem}[section]
\newtheorem{lem}{Lemma}[section]
\newtheorem{pro}{Proposition}[section]
\newtheorem{cor}{Corollary}[section]

\newcounter{remark}[section]
\newenvironment{rem}{\stepcounter{remark} \vspace{0,1cm} \noindent \textit{Remark \thesection.\theremark}\,}{\vspace{0,2cm}}

\begin{document}

\markboth{\scriptsize{Critical Speeds for a BEC -- CPRY -- \version}}{\scriptsize{Critical Speeds in the GP Theory -- CPRY  -- \version}}

\title{Critical Rotational Speeds for Superfluids in  Homogeneous Traps}
\author{M. Correggi${}^{a}$, F. Pinsker${}^{b}$, N. Rougerie${}^{c}$, J. Yngvason${}^{d,e}$	\\
	\mbox{}	\\
	\normalsize\it ${}^{a}$ Dipartimento di Matematica, Universit\`{a} degli Studi Roma Tre,	\\
	\normalsize\it L.go S. Leonardo Murialdo 1, 00146, Rome, Italy.	\\
	\normalsize\it ${}^{b}$ DAMTP, University of Cambridge, Wilbertforce Road, Cambridge CB3 0WA, United Kingdom.\\
	\normalsize\it ${}^{c}$ Universit\'e Grenoble 1 and CNRS, LPMMC, UMR 5493, BP 166, 38042 Grenoble, France.\\
	\normalsize\it ${}^{d}$ Fakult\"at f\"ur Physik, Universit{\"a}t Wien, Boltzmanngasse 5, 1090 Vienna, Austria.	\\
	\normalsize\it ${}^{e}$ Erwin Schr{\"o}dinger Institute for Mathematical Physics, Boltzmanngasse 9, 1090 Vienna, Austria.}	
\date{\version}

\maketitle
\centerline{\it Dedicated to Elliott H. Lieb on the occasion of his 80th birthday}
\begin{abstract} 
	We present an asymptotic analysis of the effects of rapid rotation on the ground state properties of a superfluid confined in a two-dimensional trap. The trapping potential is assumed to be radial and homogeneous of degree larger than two in addition to a quadratic term. Three critical rotational velocities are identified, marking respectively the first appearance of vortices, the creation of a `hole' of low density within a vortex lattice, and the emergence of a giant vortex state free of vortices in the bulk. These phenomena have previously been established rigorously for a `flat' trap with fixed boundary but the `soft' traps considered in the present paper exhibit some significant differences, in particular  the giant vortex regime, that necessitate a new approach. These differences concern both the shape of the bulk profile and the size of vortices relative to the width of the annulus where the bulk of the superfluid resides. Close to the giant vortex transition the profile is of Thomas-Fermi type in `flat' traps, whereas it is gaussian for soft traps, and the `last' vortices to survive in the bulk before the giant vortex transition are small relative to the width of the annulus in the former case but of comparable size in the latter.

	\vspace{0,2cm}

	\vspace{0,2cm}
	
\end{abstract}

\tableofcontents

\section{Introduction and Main Results}

Since their first experimental realization in 1995, atomic Bose-Einstein condensates (BECs) have become a  subject of tremendous interest, 
stimulating intense theoretical activity. 
 The BE-condensed alkali vapors that are nowadays produced in many laboratories offer a spectacular level of tunability 
for several experimental parameters, making them a favorite testing ground for many intriguing quantum phenomena. 
Among these is  \emph{superfluidity}, i.e., the occurrence of frictionless flow, previously observed in 
liquid helium.  Atomic BECs offer a valuable alternative to the latter system, since 
experiments with dilute gases of ultracold atoms allow to test theories of superfluidity in much more detail.

From a theoretical point of view an appealing aspect of the field is its sound mathematical foundation. The most commonly used 
model for the description of BECs, the so-called Gross-Pitaevskii (GP) theory, is now supported both by
extensive comparisons with experiments \cite{A,Fe1} and a solid mathematical basis \cite{LSSY}. Indeed, substantial advances on the connection between GP theory and many-body quantum physics have been made in recent years, leading 
to a rigorous derivation of both the stationary \cite{LSY1,LiS,BCPY} and dynamic aspects \cite{ESY1,ESY2,P} of the theory.

One of the most striking features of superfluids is their response to rotation of the confining trap. Typical experiments (see \cite{Fe1} for further references) start by cooling 
a Bose gas in a magneto-optical trap below the critical temperature for Bose-Einstein condensation. If the trap is set in rotational motion the gas stays to begin with at rest in the inertial frame, but if the rotation speed exceeds some critical value (called $\Ofirst$ below), vortices 
are nucleated and remain stable over a long  time span. This behavior is in strong contrast with that of a classical fluid,  which in the stationary state rotates like a rigid body and thus remains at rest in the rotating frame.

In a rotationally symmetric trap, the main two experimentally tunable parameters are the rotation speed $\Orot$ and the strength 
of interparticle interactions,  written as $ 1/\eps^{2}$ with $\ep > 0$ in the sequel. The ground state of the system strongly 
depends on the relation between these two parameters.

In this paper we  study the ground state of a rotating Bose gas in the framework of the two-dimensional GP theory\footnote{Such a description is justified if the trap almost confines the gas on a plain orthogonal to the axis of rotation or, on the contrary, the trap is very elongated along the axis, in which case the behavior can be expected to be essentially independent of the coordinate in that direction.}. We consider the minimization, under a unit mass constraint, of the GP functional\footnote{The subscript `phys' stands for physical, i.e., this is the functional in the original physical variables, as opposed to the rescaled functionals introduced in Sect.\ 1.1 below. Units have been chosen such that $\bar h$ as well as the mass are equal to 1.}
\beq
\label{gGPf -}
	\ggpf[\Psi] = \int_{\mathbb R^2} {\diff} \rv \: \lf\{ \half\lf|\lf(\nabla - i \magnp\ri) \Psi \ri|^2 + 
\lf( V(r) - \half \Omega_{\rm rot}^2 r^2 \ri) |\Psi|^2 + \frac{|\Psi|^4}{\eps^2} \ri\},
\eeq
with
\beq
	\magnp : = \Orotv \wedge \rv= \Orot r \mathbf{e}_{\vartheta},
\eeq
where  $ \Orotv$ is the rotational velocity,  $r=|\mathbf r|$ with $\mathbf r\in\mathbb R^2$ is the distance from the rotation axis, and $ \mathbf{e}_{\vartheta} $ stands for the unit vector in the direction transverse to $\mathbf r$.
The confinement is provided by a  potential of the form  
\beq
	\label{potential -}
	V(r) : = k  r^s + \half \Oosc^2 r^2
\eeq
with $k>0$. We restrict to the case $2<s<\infty $ (anharmonic, \lq soft' potentials) and shall mainly study the Thomas-Fermi (TF) limit $\ep \to 0$. 
The case $s=2$ is special in many respects because the potential $-\Orot ^2 r ^2$
due to the centrifugal force is also a quadratic function of $r$. As a consequence an upper bound is set to 
the allowed values of the rotational speed and different physics is expected in the regime where the centrifugal force 
nearly compensates the trapping force (see \cite{A} and references therein, \cite{LSY2} and \cite{LeS}). 

We have studied before \cite{CDY1,CRY,CPRY1,R2} the case of a flat trap with the unit disc as boundary.   It corresponds formally to the potential
\beq \label{flat trap}
	V(r) = 
	\begin{cases}
		0,		&	\mbox{if} \:\: r \leq 1,	\\
		+\infty,	&	\mbox{if} \:\: r > 1,
	\end{cases}
\eeq
which amounts to restricting the integration domain to the unit disc. This model can formally be regarded as the $s\to +\infty$ 
limit of the theory with the trap \eqref{potential -} but we emphasize that this statement has to be taken with care. Indeed, as revealed 
by the analysis of the present paper  (see also \cite{CPRY2}, Sect.\ III) the limit $s\to +\infty$ cannot be interchanged with the TF limit $\ep \to 0$, 
which is the basis of most rigorous treatments of the GP theory, including the present one. 

Provided the  confining potential is stronger than quadratic, i.e., $s>2$, one can distinguish three critical speeds at which major phase transitions occur 
in the BEC (references in the literature concerning these phenomena include \cite{CD,Fe2,FB,FJS,FZ,KTU,KB,KF}):
\begin{itemize}
\item When $\Orot < \Ofirst$ the condensate is vortex-free. Vortices start to appear when the rotation speed exceeds $\Ofirst$ 
and form regular patterns. An adaptation of the methods of \cite{AAB,IM1,IM2,AJR}  shows that\footnote{Note that in the cited papers the GP functional is rescaled from the outset,  as in Eq.\ \eqref{GPfo} below, and in the rescaled variables the critical velocity is $\omega_{c1}\propto |\log \eps|$.}
$\Ofirst \propto \eps^{4/(s+2)}|\log\eps|$. This is valid for all traps, including the harmonic trap $s=2$ and the flat trap, since the centrifugal force is negligible in this regime.
\item When $\Ofirst \ll \Orot \ll \Osec$ vortices are densely packed and uniformly distributed in the condensate. It is observed both numerically 
and experimentally that they arrange themselves into triangular arrays (Abrikosov lattices with hexagonal unit cells). It remains an open problem to provide 
a rigorous proof of this fact but a heuristic argument based on an electrostatic analogy and Newton's Theorem can be found in \cite{CY} and \cite{SS3} contains a rigorous proof for a simplified model. The distribution of vorticity has been shown to be uniform in this regime for the flat trap case both 
with Neumann \cite{CY} and Dirichlet boundary conditions \cite{CPRY1}. In this paper we adapt {\rm this} result to the case of soft potentials. 
\item When $\Orot$ reaches the second critical speed $\Osec$, the centrifugal force dips a hole  with strongly depleted matter  
density around the center of the trap. Vortices are however still uniformly distributed in the annular bulk where the mass 
of the condensate resides. In the flat trap case the second critical speed is of order $\ep ^{-1}$ \cite{CDY1, CRY,CPRY1}. As shown in \cite{CDY2} and discussed further in Section \ref{sec:vortex lattice} below, this order is changed in the case of a soft potential \eqref{potential -} to $\ep ^{\frac{4}{s+2}-1}$, provided $s>2$. Notice that in the limit $s\to \infty$ we recover the result for the flat trap case, despite the subtlety of the exchange of limits  (see below).
Here the behavior of the gas  differs markedly from that in a purely harmonic trap where centrifugal effects start to 
compensate the trapping potential in the regime $\Orot = \OO (1)$. In this case the state of the condensate is well approximated by a wave function in the lowest Landau level  of the magnetic Hamiltonian corresponding to the first term in \eqref{gGPf -}
\cite{ABN,AB,LSY2} 
and the GP theory ultimately loses its validity \cite{LeS}. 
\item When $\Orot$ is further increased above a third critical speed $\Othird$, a new phase transition occurs where 
vortices are expelled from the bulk of the condensate. All the vorticity resides in the central hole created by the centrifugal force, 
resulting in a \emph{giant vortex phase}.  One of the main results of this paper is that, for a potential of the form \eqref{potential -} with $s>2$,
 $\Othird$ is of order $\ep ^{-4\frac{s-2}{s+2}}$. This converges to $\ep ^{-4}$ when $s\to +\infty$, which is significantly larger than the value of $\Othird$ in the flat trap case. The latter is,  to leading order in $\ep$,  given by $(2/ (3\pi))\,( \ep ^2 |\log \ep|) ^{-1}$ \cite{CRY,R2}. The limits $\ep \to 0$ and $s\to +\infty$ can thus {\it not} be interchanged in this regime. As we explain further below, this is due to  significantly  different physical properties of the condensate in the two kinds of traps  close to the  transition regime.
 \item When $\Orot \gg \Othird$, the BEC is in a giant vortex phase and the energy is given to an excellent approximation by a trial function with a single multiple vortex at the origin (of degree $\Omega_{\rm rot}$ to leading approximation),  accounting for the macroscopic circulation. Nevertheless we prove that rotational symmetry is  still broken 
in any true ground state, namely that no ground state can be an eigenfunction of angular momentum.
\end{itemize}

The TF approximation in the regime $\Ofirst \ll \Orot \ll \Othird$ was studied in \cite{CDY2} where a preliminary upper bound for the vortex contribution was given. Precise upper and lower bounds are provided in Theorem 1.1 below.
The case of large rotation speeds with \emph{fixed} coupling constant $\ep ^{-2}$ has been considered before \cite{R1} 
for the special case of a quartic potential (i.e., $s=4$). The present contribution improves the results of that paper in several directions. 
In particular, although the appearance of a giant vortex state has been proved in \cite{R1}, no rigorous estimate of the third critical speed could  be provided. Such an estimate is given in Theorem 1.3 below.

In the remaining subsections of this introduction we define our setting more precisely and state our main results. We start by discussing the 
scalings we shall use to study the functional \eqref{gGPf -}. We then state our results about the regimes $\Ofirst \ll \Orot \ll \Othird$ and 
$\Orot \propto \Othird$ in two different subsections and finally present our symmetry breaking theorem.

For a concise discussion of our results and methods  the reader is referred to the companion paper \cite{CPRY2}.  

\subsection{Scalings of the GP Functional}

In this paper we focus on the analysis of the minimization of the GP energy functional \eqref{gGPf -}
under the mass constraint
\[
	\lf\| \Psi \ri\|_2^2 : = \int_{\R^2} \diff \rv \: \left| \Psi \right|^2 = 1.
\]
The trapping is given by the potential
 \eqref{potential -}
%
with $ \Oosc < \Orot $, and $ 0 < \eps \ll 1 $, i.e., we study the TF limit of the model. The power $ s $ in \eqref{potential -}  characterizes the homogeneous trap and we assume that $ 2<s < \infty $. We also introduce the parameter
\beq
	\label{Oeff}
	\Oeff : = \sqrt{\Orot^2 - \Oosc^2},
\eeq
so that the effective potential in \eqref{gGPf -} can be written in the form
\beq
	\label{eff potential}
	V(r) - \half \Orot^2 r^2 = k r^s - \half \Oeff r^2.
\eeq
Since we are interested in exploring the rotation regime $ \Orot \to \infty $ but  want also to keep track of the effect of the quadratic term in the expression above  we shall assume that 
\beq
	\label{oeff cond}
	\Oeff^2 = \gamma \Orot^2
\eeq
for some given $ 0 < \gamma \leq 1 $. 

Since the profile of a minimizer of \eqref{gGPf -} expands if $ \Orot \to \infty $ and/or $ \eps \to 0 $ it is natural to rescale all lengths as well as the rotational velocity. In fact it is convenient to employ two different scalings  in different asymptotic regimes: If 
\beq
	\label{subcritical}
	\Orot \lesssim \eps^{-\frac{s-2}{s+2}},
\eeq
 we  define $\xv \in \R ^2,\: x = |\xv|$, $\psi(\xv)$, $\omega$ and $\aavo$ by writing 
\beq
 	\label{rescaling 1}
	\rv =  \reps \xv,	\hspace{1cm}	\Psi(\rv) =  \reps^{-1} \psi(\xv),	\hspace{1cm}	\Orot =  \reps^{-2} \omega,	
	\hspace{1cm}	\aavo  = \omega  x \mathbf{e}_{\vartheta},	
\eeq
where
\beq
	\label{reps}
	\reps : = \lf( k \eps^{2} \ri)^{-\frac{1}{s+2}} = \OO\lf(\eps^{-\frac{2}{s+2}}\ri).
\eeq
With these scalings the functional \eqref{gGPf -}
 can be rewritten as
\beq	
	\label{geno rel}
	\ggpf[\Psi] = \reps^{-2} \gpfo[\psi],
\eeq
with
\beq
	\label{GPfo}
	\gpfo[\psi] : = \int_{\R^2} \diff \xv \lf\{ \half \lf| \lf( \nabla - i\aavo \ri) \psi \ri|^2 + \eps^{-2} \lf[ x^s \lf| \psi \ri|^2 - \half \gamma \eps^2 \omega^2 x^2 |\psi|^2 + |\psi|^{4} \ri] \ri\}.
\eeq
If  on the other hand
\beq
	\label{supercritical}
	\Orot \gtrsim \eps^{-\frac{s-2}{s+2}},
\eeq
a natural scaling parameter is given by the position $ \rmax $ of the unique minimum point of the effective potential \eqref{eff potential}, which is explicitly given by
\beq
	\label{rm}
	\rmax : = \lf( \frac{\gamma \Orot^2}{s k} \ri)^{\frac{1}{s-2}} = \OO \lf(\Orot^{\frac{2}{s-2}}\ri).
\eeq
 Writing now
\beq
 	\label{rescaling 2}
	\rv =  \rmax \xv,	\hspace{1cm}	\Psi(\rv) = \rmax^{-1} \psi(\xv),	\hspace{1cm}	\Orot =  \rmax^{-2} \Omega,	\hspace{1cm}	\aavoo  = \Omega x \mathbf{e}_{\vartheta},
\eeq
the GP energy functional \eqref{gGPf -}
 becomes
\beq	
	\ggpf[\Psi] = \rmax^{-2} \lf[ \gpfoo[\psi]  + \lf( \salf - \half \ri) \gamma \Omega^2 \ri],
\eeq
with
\beq
	\label{GPfoo}
	\gpfoo[\psi] : = \int_{\R^2} \diff \xv \lf\{ \half \lf| \lf( \nabla - i \aavoo \ri) \psi \ri|^2 + \gamma \Omega^2 W(x) |\psi|^2 + \eps^{-2} |\psi|^{4} \ri\},
\eeq
where the effective potential is given by
\beq
	\label{W}
	W(x) : = \frac{x^s - 1}{s} - \frac{x^2 - 1}{2}.
\eeq
Note that for convenience we have extracted a negative additive constant in \eqref{geno rel} exploiting the $L^2$-normalization of both $ \Psi $ and $ \psi $, in order to obtain a positive trapping potential $ W \geq 0 $ with
\beq
	\inf_{x \in \R^+} W(x) = W(1) = 0.
\eeq
Note also that 
\beq
\Omega_{\rm rot}  =(k \eps^2)^{\frac{2}{s+2}} \omega = (sk/\gamma)^{\frac{2}{s+2}}\Omega^{\frac{s-2}{s+2}}
\eeq
and 
the two rescaled velocities, $\omega$ and $\Omega$,  are related by
\beq
	\label{scaling relation}
\omega=(s/\gamma)^{\frac{2}{s+2}} \eps^{-\frac{4}{s+2}}\Omega^{\frac{s-2}{s+2}}.
\eeq
For $s\to\infty$, i.e., a flat trap, $\omega=\Omega=\Omega_{\rm rot}$. 

To find the ground state energy and wave function of the system, we minimize the functionals \eqref{GPfo} or \eqref{GPfoo} on the domain 
\beq
	\label{minimization dom}
	\gpdom : = \lf\{ \psi \in H^1(\R^2) : \: \lf\| \psi \ri\|_2 = 1 \ri\}.
\eeq
We   introduce the notation
\beq
	\label{GPe} 
	\gpeo : = \inf_{\psi \in \gpdom} \gpfo[\psi],	\hspace{1cm}	\gpeoo : = \inf_{\psi \in \gpdom} \gpfoo[\psi],
\eeq
for the ground state energies.  Any corresponding minimizer will be denoted by $ \gpmo $ or $ \gpmoo $ respectively. The GP minimizer is 
in general not unique because vortices can break the rotational symmetry (see Section \ref{sec: introsym}) but any minimizer satisfies 
in $ \R^2 $ a variational equation (GP equation) which is in the case of \eqref{GPfo}
\beq
	\label{GP variationalo}
	- \half \Delta \gpmo - \omega L_z \gpmo + \eps^{-2} \lf[ r^s + \half (1-\gamma) \eps^2 \omega^2 x^2 + 2  \lf| \gpmo \ri|^2 \ri] \gpmo  = \chemo \gpmo,
\eeq
whereas for \eqref{GPfoo} it becomes
\beq
	\label{GP variationaloo}
	- \half \Delta \gpmoo - \Omega L_z \gpmoo + \half \Omega^2 x^2 \gpmoo + \gamma \Omega^2 W(x) \gpmoo + 2 \eps^{-2} \lf| \gpmoo \ri|^2 \gpmoo = \chemoo \gpmoo.
\eeq
Here $ L_z = - i \partial_{\vartheta} $ stands for the $z$-component of the angular momentum and the chemical potentials $ \chemo $ and $ \chemoo $ are fixed by imposing the $L^2$-normalization condition on $ \gpmo $ and $ \gpmoo $ respectively  and are given by
\beq
	\label{chem}
	\chemo = \gpeo + \eps^{-2} \int_{\R^2} \diff \xv \: \lf| \gpmo \ri|^4,	\hspace{1cm}	\chemoo  = \gpeoo + \eps^{-2} \int_{\R^2} \diff \xv \: \lf| \gpmoo \ri|^4.
\eeq

The functional \eqref{GPfo} is used to analyze the behavior of the superfluid when $\Ofirst \ll \Orot \ll \Osec$, which corresponds to $|\log \ep |\ll \om \ll \ep ^{-1}$ whereas the scaling leading to \eqref{GPfoo} is more suited to the regime $\Orot \gg \Osec $, corresponding to $\ep ^{-1} \ll \Om$. The transition to a giant vortex phase occurs when $\Om \propto \ep ^{-4}$. 

For the sake of simplicity we will often omit the suffix $ \omega $ or $ \Omega $ in $ \gpmo $, $ \gpmoo $, $ \gpeo $, etc., when it is clear the regime under discussion and/or the statement applies to all the regimes taken into account in the paper, i.e., $ |\log\eps| \ll \omega $, $ \Omega \lesssim \eps^{-4} $.

\subsection{Between $ \Ofirst $ and $ \Othird $: Vortex Lattice and Emergence of a Hole}\label{sec:vortex lattice}

As we are going to see the leading order contribution to the GP energy $ \gpe $ is determined in this regime by the density $ \gpd : = | \gpm |^2 $, not the phase of the order parameter. Correspondingly we introduce a TF-like functional acting only on a density $ \rho $: If 
\beq
	\label{lattice condition 1}
	|\log\eps| \ll \omega \lesssim \eps^{-1},
\eeq
we set
\beq
	\label{tffo}
	\tffo [\rho] =  \eps^{-2} \int_{\mathbb R^2} \diff \xv \left[\lf( x^s + \rho \ri) \rho - \half \gamma \eps^2 \omega^2 x^2 \rho \right],
\eeq
 whith $ \rho $ belonging to the domain
 \beq
	\label{tfdom}
	\tfdom := \left \{ \rho \in L^2(\mathbb R^2), \lf\| \rho \ri\|_1 = 1 : \: \rho \geq 0, x^s \rho \in L^1 (\mathbb R^2) \right\}.
 \eeq
The minimization of $ \tff $ is  straightforward (see Section \ref{TF: sec}) and both the energy and minimizer have explicit expressions: The latter is the compactly supported function
\beq
	\label{tfmo}
	\tfmo(x) = \half \lf[ \eps^2 \tfchemo - x^s + \half \gamma \eps^2 \omega^2 x^2 \ri]_+,
\eeq
where $ \tfchemo $ is determined by the $ L^1$-normalization and $ [ \: \cdot \: ]_+ $ stands for the positive part. We also set
\beq
	\label{tfe}
	\tfeo : = \inf_{\rho \in \tfdom} \tffo[\rho] = \tffo\lf[ \tfmo \ri].
\eeq
On the other hand if\footnote{Notice that by the scaling relation \eqref{scaling relation}, $ \Omega = \OO(\eps^{-1}) $ if and only if $ \omega = \OO(\eps^{-1}) $.}
\beq
	\label{lattice condition 2}
	\eps^{-1} \lesssim \Omega \ll \eps^{-4},
\eeq
the scaling defined in \eqref{rescaling 2} becomes more useful and it is convenient to  define the TF functional as
\beq
	\label{tffoo}
	\tffoo[\rho] = \int_{\mathbb R^2} \diff \xv \left[ \eps^{-2} \rho^2 + \gamma \Omega^2 W(x) \rho \right],
\eeq
with $ \rho \in \tfdom $. The TF energy and minimizer,  defined analogously to \eqref{tfe}, are now denoted by $ \tfeoo $ and $ \tfmoo $ respectively. The explicit expression for the latter is 
\beq
	\label{tfmoo}
	\tfmoo(x) = \half \lf[ \eps^2 \tfchemoo - \gamma \eps^2 \Omega^2 W(x) \ri]_+.
\eeq

As we are going to see (see Section \ref{TF: sec} for further details) there are critical values $ \omegac, \Omegac = \OO(\eps^{-1}) $ of the angular velocity satisfying
\beq
	\Omegac = \rmax^2 \reps^{-2} \omegac,
\eeq
such that if $ \omega > \omegac $ (resp. $ \Omega > \Omegac $) then the TF density $ \tfmo $ (resp. $ \tfmoo $) vanishes at the origin. More precisely, if $ \Omega > \Omegac $, there exists a disc of radius $ \xin > 0 $ so that
\beq
	\tfmoo(x) = 0,	\hspace{1cm}	\mbox{for any} \:\: 0 \leq x \leq \xin.
\eeq
 We also point out that by the rescalings \eqref{rescaling 1} and \eqref{rescaling 2}, $ \omega < \omegac $ holds if and only if $ \Omega < \Omegac $, so that the conditions $ |\log\eps| \ll \omega < \omegac $ and $ \Omega_c \leq \Omega \ll \eps^{-4} $ cover all the possible asymptotic regimes between $ \Ofirst $ and $ \Othird $.

We can now formulate the result about the energy asymptotics of $ \gpe $:
 
	\begin{teo} [{\bf GP ground state energy}] 
		\label{lattice energy: teo}	
		\mbox{}	\\
		If $ |\log\eps| \ll \omega < \omegac $ as $ \eps \to 0 $, then
		\begin{equation}
			\label{gpeo asympt}
			\gpeo = \tfeo +\tx\frac 12 \omega \lf|\log \left( \eps^2\omega \right)\ri| (1 + o(1)),
		\end{equation}
		while, if $ \Omegac \leq \Omega \ll \eps^{-4} $,
		\begin{equation}
			\label{gpeoo asympt}
			\gpeoo = \tfeoo + \tx\frac{1}{6} \Omega \lf|\log \left({\varepsilon^4 \Omega} \right)\ri| (1 + o(1)).
		\end{equation}
	\end{teo} 

	As we have anticipated a transition takes place when the angular velocity gets larger than some critical value of order $ \OO(\eps^{-1}) $: A direct consequence of the asymptotics \eqref{gpeo asympt} and \eqref{gpeoo asympt} (see, e.g., \cite[Proposition 3.1]{CPRY1}) is that
	\beq
		\lf\| |\gpm|^2 - \tfm \ri\|_2 \leq o(1) \lf\| \tfm \ri\|_{\infty},
	\eeq
	i.e., the density $ |\gpm|^2 $ is close in $L^2(\R^2) $ to the TF density $ \tfm $, but when $ \Omega $ gets larger than $ \Omegac $ the latter develops a hole and a macroscopic region where $ \gpm $ is exponentially small in $ \eps $ appears close to the center of the trap. Next proposition shows that the threshold of such a transition is precisely given by $ \Omegac $, i.e.,
	\beq
		\label{Osec}
		\Osec = (sk/\gamma)^{\frac{2}{s+2}}\Omegac^{\frac{s-2}{s+2}},	\hspace{1cm}	\Omegac = 2 \sqrt{\frac{2(s+2)}{\pi \gamma (s-2)}} \lf( \frac{2}{s} \ri)^{\frac{2}{s-2}} \eps^{-1}.
	\eeq

	\begin{pro}[{\bf Emergence of a hole in $ \gpm $}]
		\label{hole: pro}
		\mbox{}	\\
		If, as $ \eps \to 0 $, $ \Omega \geq \Omega_0\eps^{-1}$ with $\Omega_0>\ep \Omegac$ 
		 there exists $ 0 < \xh < 1 $ such that
		\beq
			\label{hole}
			\lf| \gpm(\xv) \ri| = \OO(\eps^{\infty}),
		\eeq
		for any $ \xv $ such that $ 0 \leq x \leq \xh $.
	\end{pro}

We conclude this section by stating a result about the vortex distribution and  showing that throughout the parameter region $ |\log\eps| \leq \omega \lesssim \eps^{-1} $ (resp. $ \eps^{-1} \lesssim \Omega \ll \eps^{-4} $), the vorticity of $ \gpmo $ (resp. $ \gpmoo $) is uniformly distributed  on a length scale $ \Omega^{-1/2} $ inside $ \tfsuppo $ (resp. $ \tfsuppoo $)  (see also the Remark after Theorem \ref{vortex distribution: teo}). For technical reasons, however,  we are not able to prove the result in the whole support of $ \tfm $, but we have to restrict the analysis to the slightly smaller region 
\beq
	\label{bulk sub}
	{\mathcal R}_{\rm bulk} : = \lf\{ \xv \in \R^2 : \: x_< \leq x \leq x_> \ri\},
\eeq
where for some positive parameter  $ \beta : = \beta (\eps,\Omega)  $ such that $ \beta \to 0  $ as $ \eps \to 0 $ (see \eqref{gamma} for its precise definition)
\beq
	\label{x>x<}
	x_{<} : = 
	\begin{cases}
		0,							&	\mbox{if} \:\:\: |\log\eps| \ll \omega < \omegac,	\\
		\xin + \beta (\eps\Omega)^{2/3},		&	\mbox{if} \:\:\: \Omegac \leq \Omega \ll \eps^{-4},
	\end{cases}
	\hspace{1cm}
	x_{>} : = \xout - \beta (\eps\Omega)^{2/3}.
\eeq
 Here $x_{\rm in}$ and $x_{\rm out}$ are the inner and outer radii of the support of the TF density. The definitions \eqref{x>x<} imply that
\beq
	\lf| \tfsupp \setminus {\mathcal R}_{\rm bulk} \ri| = \OO\lf(\beta (1 + (\eps\Omega)^{2/3})^{-1} \ri) \ll \lf| \tfsupp \ri|,	
\eeq
and
\beq
	\lf\| \gpm \ri\|_{L^2({\mathcal R}_{\rm bulk})} = 1 - o(1),
\eeq
i.e., $ {\mathcal R}_{\rm bulk} $ contains the bulk of the mass of $ \gpm $.

In the following theorem, we describe the vorticity distribution in the bulk of the condensate for rotation speeds between the first and third critical speeds. As it is usual in the Ginzburg-Landau (GL) theory, this is done by investigating the asymptotic properties of a certain \emph{vorticity measure} whose definition is part of the statement of the theorem. Note that we do not tackle the question of the precise location of the vortices, which should be on the sites of a hexagonal lattice. Proving this fact rigorously remains a major open problem in both GP and GL theories, see \cite{SS3} for the most recent results in this direction. 

 In the statement of the theorem  we use the notation $\mathcal B(\xv,\rho)$ for a closed disc (\lq ball') with center at $\xv$ and radius $\rho$.

	\begin{teo}[\textbf{Uniform distribution of vorticity}]
		\label{vortex distribution: teo}
		\mbox{}	\\
		If $ |\log\eps| \ll \omega \lesssim \eps^{-1} $ or $ \eps^{-1} \lesssim \Omega \ll \eps^{-4} $ as $ \eps \to 0 $, there exists a finite family of disjoint balls $ \{ \ba_i \} : = \{ \ba(\xv_i, \varrho_i) \} \subset {\mathcal R}_{\rm bulk} $, $ i = 1, \ldots, N $, such that 		\ben
			\item 
$ \varrho_i \leq \OO\left(\omega^{-1/2}+\Omega^{-1/2}\right) $, 
$ \sum \varrho_i^2 \leq (1 + (\eps\Omega)^{2/3})^{-1} $;
			\item $ \lf|\gpm\ri| > 0 $ on $ \partial \ba_i $, $ i = 1, \ldots, N $.
		\een
		Moreover, setting $ d_i : = \deg\{\gpm, \partial \ba_i\} $ and defining the vorticity measure as  
		\beq
			\label{vorticity measure}
			\nu : =  2\pi\sum_{i = 1}^N d_i \delta(\xv- \xv_i), 		
		\eeq
		then, for any set $ \set \subset {\mathcal R}_{\rm bulk} $ such that $ |\set| \gg \Omega^{-1} |\log(\eps^4\Omega)|^2 $ as $ \eps \to 0 $,  
		\beq
			\frac{\nu(\set)}{\Omega |\set|} \underset{\eps \to 0}{\longrightarrow} 1.
		\eeq
	\end{teo}

	\begin{rem}({\it Vortex balls})
		\mbox{}	\\
		By a closer inspection of the proof of the above Theorem \ref{vortex distribution: teo}, one can realize that there is more information  contained in the family of balls $ \{ \ba_i \} $: Each ball might contain a very large number of vortices but almost all vortices are enclosed by those balls. More precisely, there exists a `good region' in the condensate where all the possible zeros of $ \gpm $ are contained in $ \cup \ba_i $ and outside  this set $ \gpm $ is vortex free. There is also a complementary `bad region' where we are not able to prove such a thing, but a key point of the proof of Theorem \ref{vortex distribution: teo} is that the good region covers most of the condensate, i.e., the ratio between the areas of the bad and good regions is  $o(1)$ as $\ep \to 0$.
	\end{rem}

\subsection{The Regime $ \Omega \propto \eps^{-4} $: Transition to a Giant Vortex State}\label{sec: introGV}

In this section we consider $\Omega$ given as 
\begin{equation}\label{giant vortex omega}
\Om = \frac{\Om_0}{\ep ^4} 
\end{equation}
where $\Om_0$ is a fixed constant while $\varepsilon\to 0$.

A  \lq giant vortex state',  that turns out to be an excellent approximation to the ground state for large $\Omega_0$, is obtained by minimizing the GP functional over functions of the form\footnote{Here $ \lfloor \:\:\: \rfloor $ 
stands for the integer part.} $ f(\xv) e^{i \lfloor \Omega \rfloor \vartheta} $ with $ f $ {\it real}. The giant vortex energy
functional is given by
\beq
\label{gvf}
\gvf[f] : = \gpfoo[f e^{i \lfloor \Omega \rfloor \vartheta}]=\int_{\R^2} \diff \xv \lf\{ \half \lf| \nabla f \ri|^2 + \Omega^2 U(x) f^2 + \eps^{-2} f^{4} \ri\},
	\eeq
	where 
	\beq
		\label{pot U}
		U(x) : =  \half \lfloor \Omega \rfloor^2 \Omega^{-2} x^{-2} + \half x^2 + \salf \gamma \lf( x^s - 1 \ri) - \half \gamma \lf( x^2 - 1 \ri) - \lfloor \Omega \rfloor \Omega^{-1}.
	\eeq
By standard arguments (see Proposition \ref{gvf minimization: pro}) one can show that there exists a unique minimizer $ \gvm $ of \eqref{gvf} in 
the domain
	\bdm
		\gvdom : = \lf\{ f \in H^1(\R^2) : \: f = f^* \geq 0, \: x^s f^2 \in L^1(\R^2), \: \lf\| f \ri\|_2 = 1 \ri\}.
	\edm
In addition it is easily seen that $\gvm$ is radial and does not vanish except at the origin.
	We also set
	\beq
		\label{gve}
		\gve : = \inf_{f \in \gvdom} \gvf[f].
	\eeq
The leading order of the GP energy and density will be given with very good accuracy by the minimization of the simpler giant vortex energy 
functional. It is thus of importance to understand the properties of the latter. When $\Om \ll \ep ^{-4}$, its minimization leads to 
a profile that is to leading order of TF-type as discussed in the preceding subsection.
A salient feature of the regime \eqref{giant vortex omega} is the fact that, when $\Omega_0$ gets large, the giant vortex density $\gvm$ changes from a TF profile 
to a gaussian distribution centered on the circle of radius $1$.
 This is most conveniently understood  by  using 
an auxiliary one-dimensional energy functional (recall that $\gvm$ is radial)
\bml{\label{gv aux func}
{\mathcal E}^{\mathrm{aux}}[{\text{\it f}}\,]= 2 \pi \int_{\mathbb R} \diff x \left\{\half |f'|^2+\half \Omega^2\alpha^2(x-1)^2f^2+\eps^{-2}f^4\right\}
 =	\\
2\pi \Omega\int_{\mathbb R} \diff y \left\{\half |\hat f'|^2+\half \alpha^2y^2\hat f^2+\Omega_0^{-1/2}\hat f^4\right\},
}
where
\beq
	\label{alpha}
	\al : = \sqrt{4 + 2(s-2)}.
\eeq 
One obtains ${\mathcal E}^{\mathrm{aux}}$ from \eqref{gvf}  by expanding the potential around its minimum at $ x = 1$ and neglecting terms beyond 
quadratic in the expansion.  
The variable transformation $y=\Omega^{1/2}(x-1)$, $\hat f(y)=\Omega^{-1/4}f(x)$ is then used to blow-up  to the characteristic size of the domain where the density is confined. 
It is clear that all three terms in \eqref{gv aux func} are of 
the same order of magnitude but the importance of the last term diminishes with increasing $\Omega_0$. The quadratic part of the energy 
is simply the energy functional of the harmonic oscillator, so without the last term the 
minimizer is the Gaussian
\beq\label{gv aux gaussian}
 \gosc(y)=\pi^{-1/4}\alpha^{1/4}\exp\{-\half\alpha y^2\}.
\eeq

We expect and prove (see Section \ref{sec: gv dens profile}) that $\gvm$ is close (after rescaling) to this function when $\Om_0$ is 
large.
This behavior of the density profile is very different from that given by the TF theory discussed before. The fact that the transition to a 
gaussian profile occurs in the same parameter regime  as the transition to a giant vortex phase is at the origin of several new features of the physics of the third phase transition in a soft potential as compared 
to the theory in a flat trap. This does not seem to have been noticed in the literature before: As far as we know, the profile is always assumed to be of TF type in previous papers dealing with the transition to a giant vortex state, with the exception of \cite{R1}.

To begin with, it is more delicate to define the bulk of the condensate in the regime $\Om_0$ large. The gaussian  mass distribution has 
much slower decay and longer tails than a TF-like profile. As a consequence the mass is much less concentrated and it is 
necessary to consider an annulus of thickness much larger than the characteristic scale  $\Omega^{-1/2}=\Omega_0^{-1/2}\eps^2$ of the gaussian distribution: 
\begin{equation}\label{Abulk}
 \ab := \lf\{ \xv \in \R ^2 : \: 1- c |\log \ep | ^{1/2}\Om ^{-1/2} \leq x \leq 1+ c  |\log \ep | ^{1/2}\Om ^{-1/2}\ri\}, \hspace{1cm} c < \sqrt{\frac{2}{\al }}.
\end{equation}

Our main result  about the regime \eqref{giant vortex omega} is the following:

	\begin{teo}[\textbf{Absence of vortices in the bulk}]
		\label{teo: giant vortex}
		\mbox{}	\\
		If $ \Omega $ is given by \eqref{giant vortex omega}, there exists a finite constant $ \bar\Omega_0 $ such that for any $ \Omega_0 > \bar\Omega_0 $, no minimizer $ \gpm $ has a zero inside $ \ab $ if $ \eps $ is sufficiently small.
		\newline
		More precisely, for any $ \xv \in \ab $,
		\beq
			\label{gv point diff}
			 \lf|\gpm \lf(\xv\ri)\ri| = \gvm (x) \left( 1 + \OO(|\log \ep| ^{-a})\right) 
		\eeq
for any $a>0$.
	\end{teo}

The estimate \eqref{gv point diff} justifies the notation $\ab$: Indeed we have 
\[
\int_{\ab} \diff \xv \: \lf|\gpm \ri|^2 = 1 - o (1),
\]
because such an estimate holds true for $\gvm$, so $\ab$ contains the bulk of the mass.
Theorem \ref{teo: giant vortex} then  characterizes the transition to a giant vortex  phase where the bulk of the condensate resides in an annulus free of vortices but with a macroscopic circulation around the origin. Recalling the results of Theorem \ref{vortex distribution: teo} we deduce that the third phase transition occurs when $\Om$ is of order $\ep ^{-4}$. Indeed, 
for speeds much smaller than this threshold, vortices do occur in the annular  bulk, whereas they are absent if $\Om$ is a sufficiently large 
multiple of $\ep ^{-4}$.

The core of the proof of Theorem \ref{teo: giant vortex} is a detailed analysis of the asymptotic behavior of the GP ground state 
energy, leading to the following estimate:

	\begin{teo}[\textbf{Ground state energy asymptotics}]
		\label{gv energy: teo}	
		\mbox{}	\\
		If $ \Omega $ is given by \eqref{giant vortex omega} with $ \Omega_0 > \bar\Omega_0 $  as in Theorem \ref{teo: giant vortex}, then as $ \eps \to 0 $
		\beq
			\label{gv energy asympt}
			\gpe = \gve + \OO(|\log \ep| ^{9/2}).
		\eeq
	\end{teo}

We are also able to estimate the degree of the GP minimizer, ensuring that there is a macroscopic phase circulation around the central 
low-density hole. This is the content of the following  theorem.

	\begin{teo}[\textbf{Asymptotics for the degree}]
		\label{teo: gv degree}	
		\mbox{}	\\
		If $ \Omega $ is given by \eqref{giant vortex omega} with $ \Omega_0 > \bar\Omega_0 $  as in  Theorem \ref{teo: giant vortex} and
$R$ is any radius satisfying 
\begin{equation}\label{radius degree}
 R = 1 + \OO (\Om ^{-1/2}),  
\end{equation}
then as $ \eps \to 0 $
		\beq
			\label{gv degree asympt}
			\deg\left( \gpm,\; \dd B_{R} \right) = \Om  + \OO(\Om_0 |\log \ep| ^{9/4}).
		\eeq
	\end{teo}

As already mentioned, the physics of the giant vortex phase transition  in a 
homogeneous trap turns out to differ markedly from the corresponding regime in a flat trap that was studied in \cite{CRY,CPRY1,R2}. As a 
result, the technique of proof is different and does not rely on any  construction of vortex balls. The  remaining part of this subsection aims at elucidating the main steps in the proof.  To bring out the salient points technical details are skipped and most of the 
following equations will be written in a simplified form without the remainder terms.
 We remark also that the method to be presented is  simpler than that in 
\cite{R1} and leads to sharper results, although some ideas are common to both papers.

As a starter let us note that, in   the soft potentials ($s<\infty$)  we consider, the centrifugal force does not squeeze the condensate  against a hard wall as in the flat potential ($s=\infty$) considered in \cite{CPRY1}  although it still concentrates the scaled density  in an annulus around $x=1$.  The soft potential, however, allows more spreading of the matter density than a trap with a hard wall, so the density is lower and vortices therefore energetically less costly.  This explains the fact that the third critical speed is considerably larger in a homogeneous potential than  could be expected from the analysis of the flat trap model.

It had been recognized before \cite{R1} that for extreme rotation speeds the density profile would be gaussian, but  when studying the giant vortex transition a key difficulty originates from the fact that the change from the TF profile to the gaussian distribution takes place in the same parameter regime  as the one where vortices are expelled from the bulk of the condensate. In a sense \emph{two} distinct transitions 
happen at the same time 
in the regime \eqref{giant vortex omega} and this is the main reason  why the analysis we perform in this paper is rather different from 
\cite{CRY, CPRY1}.

Indeed, in \cite{CRY, CPRY1} the potential vortex cores were much smaller than the  width of the annular bulk of the condensate and vortices disappeared for energetic reasons. To analyze this phenomenon precisely we had to adapt the  method of vortex balls originally developed in the GL theory (see \cite{SS2} and references therein) and used to evaluate $\Ofirst$ \cite{IM1,AJR}.

In a homogeneous trap the cores of potential vortices turn out to be comparable in size to the width of the annular bulk  when $\Om \sim \Othird$. It is thus not  clear whether they disappear because of their energy cost or are 
simply expelled from the condensate because they are too large to fit in. This explains why we have not been able to give a precise estimate of the third critical speed: Potential vortices are too large to treat them using the GL technology. Their energy cost can not be evaluated as  accurately as in \cite{CRY,R2} and we have to rely on coarser  estimates.

The starting point and the basic methodology is however reminiscent  to that employed in  \cite{CRY,CPRY1} and in works dealing with the first critical speed 
in GP theory \cite{AAB,IM1,IM2}. We first decouple the energy functional, an idea originating in \cite{LM} and used repeatedly in the literature.
Actually it is more convenient to restrict the integration domain before performing the energy decoupling. For technical reasons   it is not sufficient to work on $\ab$   (defined by  Eq.\ \eqref{Abulk}) and we consider instead a larger domain 
\begin{equation}\label{sketch anne}
\anne = \left\{ \xv \in \R ^2 : \: 1 - \ete \ep ^2 \leq x \leq 1 + \ete \ep ^2\right\}
\end{equation}
where we have some freedom in the choice of  $\ete\gg 1$. It should in any case satisfy the condition $\ete \gg |\log \ep| ^{1/2}$ in order for $\anne$ to contain $\ab$ but since we shall also need to consider $\mathcal A_{\sqrt\ete}$ we even require $\ete \gg |\log \ep|$. 
If $\ete$ is chosen large enough, both the true GP minimizer and the giant vortex density profile are exponentially small in $ \eps $ outside 
$\anne$ and one can restrict the analysis to this annulus at the price of negligible remainders.  On the other hand, the proof of energy estimates and the absence of vortices requires that $\ete$ is not too large and we shall eventually choose $\ete=|\log\eps|^{3/2}$, cf.\ Eq. \eqref{ete choice}. 

By defining a giant vortex functional reduced to $\anne$ and minimizing it, we obtain an energy $\gvee$ and a profile $\gvme$ supported in 
the annulus $\anne$. We then write
 any test function $\psi$ as
\beq \label{sketch v}
\psi(\xv) = \gvme(x) e^{i\lfloor \Om \rfloor \vartheta} v(\xv),  
\eeq
which defines $v$ on $\anne$.  Since $\gvme$ is without zeros, the function $v$
 contains all possible zeros of $\psi$ in the annulus. 
The variational equation for $\gvme$ leads to 
\beq \label{sketch decouple}
\gpf [\psi] \simeq \gvee  + \Ee [v] 
\eeq
with
\beq\label{sketch: eeta}
\Ee [v] =\int_{\anne} \diff \xv\:  \gvme ^2 \left\{\half|\nabla v|^2- \rmagnp \cdot (iv,\nabla v) + \eps^{-2}\gvme ^2 (1-|v|^2)^2\right\},
\eeq
where $ \rmagnp =\left(\Omega x-\lfloor \Omega \rfloor x^{-1} \right)\mathbf e_\vartheta$ and $(iv,\nabla v)=\frac {\rm i}2(v\nabla v^*-v^*\nabla v)$.  
It is   plausible that if we write
\beq \label{sketch u}
\gpm(\xv) = \gvme(x) e^{i\lfloor \Om \rfloor \vartheta} u(\xv),  
\eeq
 then the function $u$ defined on $\anne$ should minimize $\Ee$ under the constraint
\beq \label{sketch: constraint}
 \int_{\anne} \diff \xv \: \gvme ^2 |u| ^2 = 1.
\eeq
This is not fully rigorous because some remainder terms in \eqref{sketch decouple} are neglected, most of which can be controlled using the exponential smallness of 
$\gpm$ outside of $\anne$. Nevertheless, studying the reduced problem of minimizing $\Ee$ under the constraint \eqref{sketch: constraint} gives some valuable insight on the 
full proof we are to develop in Section \ref{sec: GV}.

The appropriate trial function is given by $v = 1$, which corresponds to a pure giant vortex ansatz. This gives
\[
 \Ee [u] \leq 0
\]
and the gist of the proof is the obtainment of the corresponding lower bound, proving that the giant vortex trial function optimizes 
the energy  up to small remainders. To this end one needs to control the only possibly negative contribution originating in the second term of \eqref{sketch: eeta}. 
This is most conveniently done using a potential function defined as in \cite{AAB,CRY}:
\begin{equation}\label{sketch: potential}
F(x) = \int_{1 - \eps^2 \ete}^{x} \diff t \: \lf( \Om t - \lfloor \Om \rfloor t^{-1}\ri) \gvme^2(t). 
\end{equation}
This function satisfies $\nabla ^{\perp} F = \gvme^2 \rmagnp $ and vanishes on the inner boundary of $\anne$. Using Stokes' formula on the 
second term of \eqref{sketch: eeta} leads to 
\begin{equation}\label{sketch: ipp}
\Ee [u] =\int_{\anne} \diff \xv \lf\{  \half \gvme ^2 |\nabla u|^2 +F \mu + \eps^{-2}\gvme ^4 (1-|u|^2)^2 \ri\} - F (1+\ep ^2 \ete)
\int_{\dd B_{1+\ep ^2 \ete}} \diff \sigma \: (iu,\dd _{\tau} u).  
\end{equation}
The quantity $\mu$ is defined as 
\begin{equation}\label{sketch: vorticity}
\mu := \curl (iu,\nabla u) 
\end{equation}
and is often referred to as the \emph{vorticity measure} because in many situations involving vortices one can prove that such a measure is close to measure of the form \eqref{vorticity measure}: a sum of Dirac
 masses counting the vortices of 
$u$ with multiplicity. It is not completely clear that $\mu$ deserves such a denomination in the regime we consider: As already emphasized,
 the cores of potential vortices\footnote{The size of vortex cores can be estimated by taking it as a variational parameter in a trial function with prescribed zeros and minimizing the sum of the first and third terms in \eqref{sketch: eeta}.} are too large to employ the  jacobian estimate of \cite{JS} that would put the comparison with a sum of 
Dirac masses on a rigorous ground. For illustration we shall employ the terminology nevertheless. It is  justified by the fact that only the term involving 
$\mu$ in $\Ee[u]$ could favor the nucleation of vortices. Indeed the other bulk terms are positive and the boundary term should morally be 
negligible as we now explain.

Since $\Omega$  is very large it is plausible that one can replace it with its integer part at negligible cost. Also the density profile should be almost 
symmetric about  the radius $1$. If it were exactly symmetric, $F(1+\ete \ep ^2 )$ would be extremely small\footnote{As $ t - t ^{-1}$ is to a very good approximation antisymmetric with respect to $1$, $F(1+\ete \ep ^2 )$ would reduce to the integral of an odd function on a domain symmetric with respect to the origin, which vanishes.} 
and the boundary term in \eqref{sketch: ipp} could be neglected, reducing the proof 
of the lower bound to the estimate of the second term in \eqref{sketch: ipp}.

To deal with the term $\int_{\anne} F \mu$ we first remark  that due to the Cauchy-Schwarz inequality we have
\beq\label{sketch: control vorticity}
 \left| \mu \right| \leq |\nabla u| ^2
\eeq
pointwise, and in the regime \eqref{giant vortex omega} $\gvme$ and $F$ are of the same order of magnitude. It is thus plausible that 
$\gvme^2 |\nabla u| ^2$ can control $F \mu$ pointwise. This is in  marked contrast 
to the situation in a flat trap where $\gvme$ is much smaller than $ F $, making  the method of  vortex balls   indispensable
to improve the precision of the estimates. 

For $\Om_0$ large enough, $\gvme$ is approximately gaussian, so that one can use the ansatz \eqref{gv aux gaussian} to compute $ F $ and obtain
\begin{equation}\label{sketch: gv controls F}
|F (x) | \leq \al^{-1} \gvme^2(x) 
\end{equation}
for any $\xv \in \ab$ and $\Omega_0$ large enough. As $\al > 2$ (see \eqref{alpha}), combining \eqref{sketch: control vorticity} with \eqref{sketch: gv controls F} we can conclude that the 
kinetic energy term in \eqref{sketch: ipp} dominates the term involving the vorticity measure, which leads to a lower bound confirming that the giant vortex ansatz is optimal (up to negligible remainders). As a by-product of the proof one can also obtain an estimate of the third term in \eqref{sketch: eeta} that contradicts the presence of vortices in $\ab$: The technique, involving an estimate of $\left\Vert \nabla u \right\Vert_{L^{\infty}}$,  
originates in \cite{BBH} and has already been used in the flat trap case \cite{CRY,CPRY1}.

What remains to be understood is the way in which one can justify the approximations made above. We have assumed that $\Omega$ could be replaced by its integer part. This is harmless and the corresponding remainder terms are ignored in what follows. More importantly we have assumed that
\begin{itemize}
\item $\gvme$ is a radial Gaussian centered on $x = 1$ for $\Om_0$ large enough, which allows to prove \eqref{sketch: gv controls F};
\item in particular $ \gvme $ is symmetric with respect to $1$ so that the integration by parts in \eqref{sketch: ipp} produces only very small remainder terms.  
\end{itemize}

It is  fairly easy to prove that $\gvme$ is close to a Gaussian around the center of the bulk profile where the matter density is large enough. 
However, we need to justify the pointwise estimate \eqref{sketch: gv controls F} in a domain that is much larger than the bulk, 
including a region where the density $\gvme$ is extremely small. To obtain an appropriate lower bound to $\gvme$, a subsolution is constructed for the corresponding variational equation. For the method 
to apply it is necessary to restrict to a smaller annulus: We are able to prove that $\gvme$ is pointwise close to the optimal Gaussian 
only in $\anns$. The region $\anne \setminus \anns$ is dealt with\footnote{There is some freedom in the choice of $\anns$, the essential point being
that we can prove the appropriate lower bound on $\gvme$ only in a smaller domain than $\anne$.}
 using the exponential smallness of $\gpm$ in the low density hole.

As for the boundary term  in \eqref{sketch: ipp} one can prove that $\gvme$ is to a very good approximation symmetric around $x=1$. Indeed, as we have explained, 
the profile becomes gradually gaussian when $\Om_0$ is increased, because the weight of the quartic component of the energy decreases. 
The reason we have separated the discussion of the symmetry from that of the approximation by a Gaussian is that $\gvme$ is very close to a 
symmetric function even for $\Om_0$ small where the gaussian approximation is not expected to be valid. In some sense $\gvme$ is much more 
symmetric than it is gaussian. To understand this, one can consider the energy functional \eqref{gv aux func} obtained by keeping only the quadratic term in the Taylor 
expansion of the potentials: After rescaling, the minimizer is symmetric for any $\Om_0$ and one can prove that it 
approximates $\gvme$ much better\footnote{One can compare the remainders in \eqref{gvme point est} and \eqref{gvme sym est} below.} than the Gaussian 
\eqref{gv aux gaussian}.

 Making use of this property 
we show that $F(1 + \eps^2 \ete)$ is indeed very small. Then one could hope to conclude that the boundary term is negligible if a proper estimate of the phase circulation 
\[
\int_{\dd B_{1+\ep ^2 \ete}} \diff \sigma \: (iu,\dd _{\tau} u)   
\]
 were available. This is however not easy to obtain because the integral is located in a region where $\gvme$ is very small, casting doubt 
 on the possibility to prove any estimate of the term above using the kinetic energy $\half \int_{\anne} \gvme ^2 |\nabla u| ^2$.

To get around this point we modify our strategy as follows: Instead of the potential function \eqref{sketch: potential} we use 
{\it two} different potentials $F_1$, $F_2$ vanishing respectively on $\dd B_{1-\ep ^2 \ete}$ and $\dd B_{1+\ep ^2 \ete}$ . The first is used to calculate the part of the momentum term residing between the inner boundary of $\anne$ and the maximum point of $\gvme$, while we use the second in the complementary region, leading to a formula of the form 
\begin{equation}\label{sketch:final}
\Ee [u] \simeq \int_{\R ^2} \diff \xv \lf\{ \half \gvme ^2 |\nabla u|^2 +(F_1 +F_2) \mu + \eps^{-2}\gvme ^4 (1-|u|^2)^2 \ri\} - (F_1 (1) - F_2 (1))
\int_{\dd \B_1} \diff \sigma \: (iu,\dd _{\tau} u),  
\end{equation}
with the convention that $F_1 = 0$ for $r \geq 1 $ and $F_2 = 0 $ for $r \leq 1$. In some sense we are using  
a discontinuous potential, which has the effect to produce a boundary term on $\dd B_1$, i.e., in the middle of the bulk\footnote{More 
precisely the jump is at the maximum point of the profile, which is close to 1 but not exactly equal to it.}.  An important point is that $|F_1 +F_2| < \half \gvme ^2 $ when $\Om_0$ is large, which means that the analysis of the bulk terms can be carried out exactly as sketched above.

On the other hand the jump of the potential function $F_1 (1) - F_2 (1)$ is roughly speaking given by $F(1+ \eps^2\ete)$ with $F$ defined in \eqref{sketch: potential} and is thus very small because of the symmetry 
of $\gvme$. The key point is that this information can be combined with an efficient estimate of the phase circulation of $u$ on $\dd B_1$. Indeed, the
boundary term is now located where the matter density is large, which allows to control the phase circulation of $u$ in terms of the kinetic
 term in the energy (see Lemma \ref{lem:degree}). 
This control also leads to the degree estimate of Theorem \ref{teo: gv degree} once $u$ has been proved to be vortex free.

The trick of using a discontinuous potential function to make boundary terms tractable in the study of functionals such as \eqref{sketch decouple} 
seems to be new. In situations considered before \cite{AAB,CRY,CPRY1}, the TF-like profile allowed to employ other methods for estimating 
boundary terms. These fail in the context of the present paper due to the long tails of the approximately gaussian density profile.

\subsection{Rotational Symmetry Breaking}
\label{sec: introsym}

The nucleation of vortices in a rotating superfluid trapped in a radial potential has an important connection with the symmetry properties of the ground state. For example, when $\Orot$ is kept below the first critical speed, the ground state is expected to be  unique and rotationally symmetric. This has recently been proved in \cite{AJR} at least in the TF limit: The ground state for $\Orot <  \Ofirst$ is equal to the ground state without rotation,  i.e., for $\Orot = 0$, which is a radial function. When there is only one vortex in the superfluid, it will sit at the center of the trap and the ground state, although not radial any longer, is still an eigenfunction of the angular momentum, i.e., it is of the form $f(r) e^{id\vartheta}$ with $f$ a real-valued function and $d \in \Z $.

On the other hand the occurrence of more than one vortex in the GP minimizer $ \gpm $ is automatically associated with a breaking of the rotational symmetry of the functional at the level of the ground state (see, e.g., \cite{Seir,CDY1} for a detailed discussion of this phenomenon). Therefore Theorem \ref{vortex distribution: teo} guarantees that  rotational symmetry is actually broken for any $ |\log\eps| \ll \omega $, $ \Omega \ll \eps^{-4} $, since the bulk of the condensate is filled with a large number of vortices, namely of the order of the angular velocity $ \omega $ or $ \Omega $.

However, the giant vortex transition described in Section \ref{sec: introGV} seems to suggest a restoration of the symmetry under rotations: For instance the proof of Theorem \ref{gv energy: teo} shows that a giant vortex trial function with all the vorticity concentrated at the origin can reproduce the GP ground state energy up to subleading order corrections. In addition the bulk of the condensate is vortex free (see Theorem \ref{teo: giant vortex})  so all vorticity must reside in a region where the density is low. 

In the following theorem we prove that the rotational symmetry is never restored even though the angular velocity crosses the giant vortex threshold.  This suggests that there is a nontrivial distribution of vortices in the low density hole and it is an interesting open problem to locate them,  or else find other reasons for the symmetry breaking.

	\begin{teo}[{\bf Rotational symmetry breaking}]
		\label{symmetry break: teo}
		\mbox{}	\\
		If $ \Omega \gtrsim \eps^{-4} $, no minimizer $ \gpm $ of the GP functional is an eigenfunction of angular momentum for $ \eps $ sufficiently small.
	\end{teo}

Note that the theorem is stated for $ \Omega \gtrsim \eps^{-4} $, i.e., in the giant vortex regime. This is the regime where one could a priori hope that the ground state of the theory is exactly of the form $f(r) e^{id\vartheta}$. 

The same result has already been proven for smaller angular velocities in \cite[Theorem 2]{Seir}: Assume for instance that $ |\log\eps| \lesssim \omega \ll \eps^{-1} $, then it is not difficult to realize that the total winding number of any minimizer must be $ \OO(\omega) $. Therefore if one picks some finite $ d > s/2 + 1 $, Theorem 2 in \cite{Seir} proves the instability of any symmetric vortex of degree $ n = \OO(\omega) $:  Condition (2.29) in \cite{Seir}   is indeed always satisfied for $ \eps $ sufficiently small, since the chemical potential is of order $ \OO(\eps^{-2}) $, whereas the r.h.s of (2.29) is $ \OO(\omega^2) $.

\subsection{Organization of the Paper}

The proofs of the results stated above are given in the following sections. The vortex lattice regime is discussed in Section 2, beginning in Sect. 2.1 with a comparison between the TF energy and density with those of an energy functional where the vector potential responsible for the appearance of vortices has been dropped. Section 2.2 contains the proof of upper and lower bounds to the vortex contributions to the energy. The uniform distribution of vorticity is proved in Section 2.3.

Section 3 is concerned with the giant vortex transition. After some preliminary estimates on the ground state energy and minimizing wave functions in Sect. 3.1, we study in Sect. 3.2 the giant vortex density profile. The essential results here are the estimates on the deviation from a Gaussian (Proposition 3.5) and from another profile symmetric about $x=1$ (Proposition 3.7). The basic energy estimates are given in Sect. 3.3 and in Sect. 3.4 they are combined with gradient estimates to establish the giant vortex transition. 

Finally, breaking of rotational symmetry is proved in Section 4.

\section{The Regime $ |\log\eps| \ll \omega, \ \Omega \ll \eps^{-4} $}

\subsection{The TF Energies and Densities}
\label{TF: sec}



We start by summing up some important information about the minimization of the TF functionals $ \tffo $ and $ \tffoo $ (Equations \eqref{tffo} and \eqref{tffoo}): We often refer to \cite{CDY2} for a detailed analysis and omit here the proof details.

The unique TF minimizer  of the functional \eqref{tffo} for $ |\log\eps| \ll \omega \lesssim \eps^{-1} $ is the density \eqref{tfmo}, i.e.,
\bdm
	\tfmo(x) = \half \left[ \eps^2 \tfchemo - x^s + \half \gamma \eps^2 \omega^2 x^2 \ri]_+ 
\edm
with $ \tfchemo = \tfe + \eps^{-2} \| \tfmo \|^2_2 $. The support of $ \tfmo $ is always a compact set since  $ s > 2 $, and there exists a critical value of the angular velocity $ \omegac = \OO(\eps^{-1}) $ such that
\beq
	\label{tfsupp}
	\tfsuppo = 
	\begin{cases}
		\lf\{ \xv \in \R^2 : \: x \in [0,\xout] \ri\},		&	\mbox{if} \:\: \omega \leq \omegac,	\\
		\lf\{ \xv \in \R^2 : \: x \in [\xin,\xout] \ri\},		&	\mbox{if} \:\: \omega > \omegac
	\end{cases}
\eeq
with $\xin>0$.  It is given by
\beq
	\label{omegac}
	\omegac = 2^{\frac{s}{s+2}} \gamma^{-1/2} \lf( \frac{2(s+2)}{\pi (s-2)} \ri)^{\frac{s+2}{2(s-2)}} \eps^{-1}.
\eeq
If $ \omega = \omegac $, the outer radius is
\bdm
	\xout = \lf( \half \gamma \eps^2 \omegac^2 \ri)^{\frac{1}{s-2}}. 
\edm
For generic $ \omega $ neither $ \xin $ nor $ \xout $ are explicitly given but can be obtained by solving the equations $ \tfmo(\xinout) = 0 $. Both radii are increasing functions of $ \omega $, whereas both $ \tfchemo $ and $ \xout - \xin $, i.e., the width of $ \tfsuppo $, decrease as $ \omega $ increases.

When $ \omega $ becomes larger so that  $ \omega \gtrsim \eps^{-1} $ (which is the same as $\Omega\gtrsim \eps^{-1} $) we use the definition \eqref{tffoo} for the TF functional. The TF minimizer is in this case
\beq
	\label{tfm 2}
	\tfmoo(x) : = \half \lf[ \eps^2 \tfchemoo - \gamma \eps^2 \Omega^2 W(x) \ri]_+.
\eeq
The support of $ \tfmoo $ is always compact,  but  if $ \Omega \gtrsim \eps^{-1} $ is sufficiently large, the density vanishes in a neighborhood of the origin. More precisely one can define an analogue,  denoted by $ \Omegac $, of the critical value $ \omegac $ as the smallest value of $ \Omega $ such that $ \tfmoo $ vanishes at the origin $ x = 0 $. A simple computation shows that
\beq
	\label{Omegac}
	\Omegac = \rmax^2 \reps^{-2} \omegac = 2 \sqrt{\frac{2(s+2)}{\pi \gamma (s-2)}} \lf( \frac{2}{s} \ri)^{\frac{2}{s-2}} \eps^{-1},
\eeq
i.e., it agrees with what is implied by the scalings \eqref{rescaling 1} and \eqref{rescaling 2}.
\newline
If $ \Omega \gg \eps^{-1} $ all the expressions simplify and, by using the Taylor expansion 
\beq
	\label{Taylor W}
	W(x) = \half (s-2)(1 - x)^2 + \OO\lf( |1 - x|^3\ri),
\eeq   
and the fact that $ \tfmoo(\xinout) = 0 $, one obtains
\beq
	\tfchemoo = \half \gamma (s-2) \Omega^2 \lf( 1 - \xinout \ri)^2 \lf(1 + \OO(|1-\xinout|) \ri),
\eeq
so that the normalization condition yields
\beq
	\label{tfsupp width}
	\lf|1 - \xinout \ri| = \lf[ \frac{3}{2 \pi \gamma (s-2)} \ri]^{1/3} (\eps\Omega)^{-2/3} \lf( 1 + \OO((\eps \Omega)^{-2/3}) \ri),
\eeq
and $ \tfchemoo = \OO(\eps^{-4/3} \Omega^{2/3}) $, which in turn implies that
\beq
	\label{tf sup}
	\lf\| \tfmoo \ri\|_{L^{\infty}(\R^2)} = \OO((\eps\Omega)^{2/3}).
\eeq
Note also that for $ \Omega \gg \eps^{-1} $ the density can be written as
\beq
	\tfmoo(x) = \tx\frac{1}{4} \gamma (s-2) \eps^2\Omega^2 \lf[ (1-\xinout)^2 - (1-x)^2 + \OO\lf((\eps\Omega)^{-1}\ri) \ri]_+.
\eeq

For further convenience we introduce another reference density, which is obtained by minimizing a GP-like energy obtained by adding some radial kinetic energy to $ \tff $, i.e., if $ |\log\eps| \ll \omega < \omegac $,
\beq
	\label{hgpfo}
	\hgpfo[f] : = \int_{\mathbb R^2} \diff \xv \left\{ \half |\nabla f|^2  - \half \gamma \omega^2 x^2 f^2 + \eps^{-2} \left[\lf( x^s + f^2 \ri) f^2 \right] \right\},
\eeq
and, if $ \Omegac \leq \Omega \ll \eps^{-4} $,
\begin{equation}
 	\label{hgpfoo}
	\hgpfoo[f] : = \int_{\mathbb R^2} \diff \xv \left\{ \half |\nabla f|^2 + \gamma \Omega^2 W(x) f^2 + \varepsilon^{-2} f^4 \right\},
\end{equation}
where $ f \in \hgpdom $ and
\beq
	\label{hgpdom}
	\hgpdom : =  \lf\{ f \in H^1(\R^2), f \geq 0 : \: \lf\| f \ri\|_2 = 1 \ri\}.
\eeq
We also set
\beq
 	\label{hgpe}
	\hgpeo : = \inf_{f \in \hgpdom} \hgpfo[f],	\hspace{1cm}	\hgpeoo : = \inf_{f \in \hgpdom} \hgpfoo[f],
\eeq
and in both cases we denote (with some abuse of notation) by $ g_0 $ the associated unique  nonnegative minimizer.
\newline
The energy $ \hgpeo $ (resp. $ \hgpeoo $) is trivially bounded from below by $ \tfeo $ (resp. $ \tfeoo $). Next proposition proves that it is also bounded from above by the same quantity up to smaller order corrections:

	\begin{pro} [{\bf Estimates for $ \hgpeo $ and  $ \hgpeoo $}] 
		\label{hgpe ub: pro}
		\mbox{}	\\
		If $ 0 \leq \omega < \omegac $ as $ \eps \to 0 $ then
		\beq
			\label{hgpeo ub}
			\tfeo \leq \hgpeo \leq	\tfeo + \OO(|\log\eps|),
		\eeq
		while, if $ \Omegac \leq \Omega \ll \varepsilon^{-4} $, 
		\begin{equation}
			\label{hgpeoo ub}
			\tfeoo \leq \hgpeoo \leq \tfeoo + \OO( (\eps \Omega)^{4/3}  |\log(\varepsilon^4 \Omega)|).
		\end{equation}
	\end{pro} 

	\begin{rem}({\it Radial kinetic energy of $ \tfmoo $})
		\mbox{}	\\
		Using \eqref{tfsupp width} as well as \eqref{tf sup} and the Taylor expansion \eqref{Taylor W}, one easily obtains that
		\beq
			\tfeoo = \OO(\eps^{-4/3} \Omega^{2/3}),
		\eeq
		which is much larger than the error term in \eqref{hgpeoo ub}, i.e., the radial kinetic energy of the (regularized) TF profile $ \tfmoo $, for any $ \Omega \ll \eps^{-4} $, since
		\bdm
			(\eps^4 \Omega)^{2/3}|\log(\eps^4\Omega)| \ll 1.
		\edm
	\end{rem}
	We  also remark that employing the variational equation for $g_0$ one can show that $\hat E_\Omega^{\rm GP}\geq 
	E_\Omega^{\rm TF}+ c (\eps\Omega)^{4/3}|\log (\eps^4\Omega)|$ with some $c>0$.
	\medskip

	\begin{proof}
		The result is proven by evaluating the functional $ \hgpfo $ (resp. $ \hgpfoo $) on a suitable regularization\footnote{The regularization is necessary since the kinetic energy density is otherwise not integrable around $ x = \xinout $.} of $ \sqrt{\tfmo} $ (resp. $ \sqrt{\tfmoo} $) like in \cite{CY}. The error terms in \eqref{hgpeo ub} and \eqref{hgpeoo ub} are precisely given by the kinetic energy of such functions. For instance, if $ \Omega \geq \Omegac $, it is clear that the kinetic energy of the TF profile far from $ \xinout $  is of order $ (\eps\Omega)^{4/3} $, but one has to regularize the density in a neighborhood of $ \xin $ and $ \xout $: By a simple optimization, one can show that the width of such regions has to be of order $ (\eps^4 \Omega)^a (\eps\Omega)^{-2/3} \ll (\xout - \xin) $ for some $ a > 0 $, which produces the additional factor $ |\log(\eps^4 \Omega)| $ (see, e.g., \cite[Proposition 2.1]{CPRY1}). 
	\end{proof}

With some slight abuse of notation we denote by $ g_0 $ the minimizer of either \eqref{hgpfo} or \eqref{hgpfoo}. It has other useful properties which can be easily proven: It is radial and strictly positive and has a unique maximum. Moreover one can show that it is suitably close to $ \tfm $: A trivial consequence of \eqref{hgpeo ub} and \eqref{hgpeoo ub} is (see, e.g., \cite[Proposition 2.1]{CRY}):
\begin{eqnarray}
	\label{l2 est g0 1}
	\lf\| g_0^2 - \tfmoo \ri\|_2^2 \leq \OO(\eps^2|\log\eps|),	\hspace{0,2cm}	&\mbox{if}&  0 \leq \omega < \omegac;	\\
	\label{l2 est g0 2}
	\lf\| g_0^2 - \tfmoo \ri\|_2^2 \leq \OO(\eps^{10/3} \Omega^{4/3} |\log(\eps^4\Omega)|), \hspace{0,2cm}	 &\mbox{if}&  \Omegac \leq \Omega \ll \eps^{-4}.
\end{eqnarray}
One can also show as in \cite[Proposition 2.1]{CRY} that
\beq
	\label{linfty est g0}
	\lf\| g_0 \ri\|_{L^{\infty}(\R^2)}^2 \leq 
	\begin{cases}
		\lf\| \tfmo \ri\|_{\infty} (1 + o(1))= \OO(1),			&	\mbox{if} \:\:\: 0 \leq \omega < \omegac;	\\
		\tfmoo(1) (1 +o(1))  = \OO(\eps^{2/3} \Omega^{2/3})	 	&	\mbox{if}  \:\:\: \Omegac \leq \Omega \ll \eps^{-4}.
	\end{cases}
\eeq
The estimates \eqref{l2 est g0 1} and \eqref{l2 est g0 2} can actually be improved:

	\begin{pro}[{\bf Pointwise estimate for $g_0$}]
		\label{point est g0: pro}
		\mbox{}	\\
		If $ 0 \leq \omega < \omegac $\footnote{The condition $ \omega < \omegac $ should be thought of for $ \omega = \OO(\eps^{-1}) $ as the stronger requirement $ \omega = \omega_0 \eps^{-1} $ for some constant $ \omega_0 < \omegac \eps $. Similarly the condition $ \Omega \geq \Omegac $ means actually $ \Omega \geq (1 - o(1)) \Omegac $ and, if $ \Omega $ is below $ \Omegac $, $ \xin $ has to be set equal to $ 0 $.} as $ \eps \to 0 $, then
		\beq
			\label{point est g0 1}
			\lf| g_0^2(x) - \tfmo(x) \ri| \leq \OO(|\log\eps|^{-1}),
		\eeq
		uniformly in $ \xv \in \Ao $ with 
		\beq
			\label{Ao}
			\Ao : = \lf\{ \xv \in \R^2 : \: 0 \leq x \leq \xout - |\log\eps|^{-1} \ri\}.
		\eeq
		If on the other hand $ \Omegac \leq \Omega \ll \eps^{-4} $
		\beq
			\label{point est g0 2}
			\lf| g_0^2(x) - \tfmoo(x) \ri| \leq \OO(\eps^{4/3} \Omega^{1/3} |\log(\eps^4\Omega)|^{3/2}) \: \tfmoo(x) = o(1) \: \tfmoo(x),
		\eeq
		uniformly in $ \xv \in \A_{\Omega} $ with
		\beq
			\label{Aoo}
			\Aoo : = \lf\{ \xv \in \R^2 : \: \xin + (\eps \Omega)^{-2/3} |\log(\eps^4\Omega)|^{-1} \leq x \leq \xout - (\eps \Omega)^{-2/3} |\log(\eps^4\Omega)|^{-1} \ri\}.
		\eeq
	\end{pro}

	\begin{rem}({\it Bulk of the condensate})
		\mbox{}	\\
		The annulus $ \Aoo $ contains automatically the bulk of the mass provided $ \Omega \ll \eps^{-4} $: Using \eqref{tfsupp width} it is simple to verify that
		\beq
			\lf\| \tfmoo \ri\|_{L^1(\Aoo)} = 1 - \OO(|\log(\eps^4\Omega)|^{-2})
		\eeq
		and
		\beq
			\label{tfmoo Aoo lb}
			\tfmoo(x) \geq C (\eps \Omega)^{2/3} |\log(\eps^4\Omega)|^{-1} > 0,	
		\eeq
		for any $ \xv \in \Aoo $. The same obviously holds true for $ \Ao $ as well: 
		\beq
			\label{tfmo Ao lb}
			\tfmo(x) \geq C |\log\eps|^{-1} > 0,		\hspace{1cm}	\lf\| \tfmo \ri\|_{L^1(\Ao)} = 1 - \OO(|\log\eps|^{-2}).
		\eeq	
		As a direct consequence of the pointwise estimates \eqref{point est g0 1} and \eqref{point est g0 2}, $ \Ao $ and $ \Aoo $ contain the bulk of the mass of $ g_0 $ too:
		\begin{eqnarray}
			\label{bulk g0 1}
			\lf\| g_0 \ri\|^2_{L^2(\Ao)} = 1 - \OO(|\log\eps|^{-1}), \hspace{0,2cm}	&	\mbox{if} & 0 \leq \omega < \omegac,	\\
			\label{bulk g0 2}
			\lf\| g_0 \ri\|^2_{L^2(\Aoo)} = 1 - \OO(|\log(\eps^4\Omega)|^{-2}), \hspace{0,2cm}		&	\mbox{if} & \Omegac \leq \Omega \ll \eps^{-4}.
		\end{eqnarray}
	\end{rem}

	\begin{proof} 
		The proof is done exactly as in \cite[Proposition 2.6]{CRY} (see also \cite[Proposition 2.3]{CPRY1}), i.e., by identifying local super- and subsolutions to the variational equation solved by $ g_0 $, which can be rewritten
		\beq
			\label{variational eq g0}
			- g_0^{\prime\prime} - x^{-1} g_0^{\prime} = 4 \eps^{-2} \lf[ \tilde{\rho} - g_0^2 \ri] g_0,
		\eeq
		where
		\beq
			\label{tilderho}
			\tilde{\rho}(x) : =
			\begin{cases}
				 \half \left[ \eps^2 \hchemo - x^s + \half \gamma \eps^2 \omega^2 x^2 \ri],	&	\mbox{if} \:\:\: 0 \leq \omega < \omegac,	\\	
				\half \lf[ \eps^2 \hchemoo - \gamma \eps^2 \Omega^2 W(x) \ri],	&	\mbox{if} \:\:\: \Omegac \leq \Omega \ll \eps^{-4},
			\end{cases}
		\eeq
		and the chemical potentials $ \hchemo $ and $ \hchemoo $ are fixed by the $L^2$-normalization of $ g_0 $. Note that $ \tilde\rho $ differs from the TF densities $ \tfmo $ and $ \tfmoo $ only by the chemical potential. However combining \eqref{l2 est g0 1}, \eqref{l2 est g0 2} with \eqref{linfty est g0}, one can easily obtain the following estimate of the difference between the chemical potentials:
\begin{eqnarray}
	\label{hgpchem est 1}
	\lf| \hchemo - \tfchemo \ri| \leq \OO(\eps^{-1}|\log\eps|^{1/2}),	\hspace{0,2cm}									&	\mbox{if} & 0 \leq \omega < \omegac;	\\
	\label{hgpchem est 2}
	\lf| \hchemoo - \tfchemoo \ri| \leq \OO(\Omega |\log(\eps^4\Omega)|^{1/2}),	\hspace{0,2cm}		&	\mbox{if} & \Omegac \leq \Omega \ll \eps^{-4}.
\end{eqnarray} 
	This in turn yields a pointwise estimate of the difference between $ \tilde\rho $ and the TF densities, which guarantees that $ \tilde\rho $ is strictly positive inside $ \Ao $ or $ \Aoo $.
	
	The rest of the proof consists in a local analysis of the variational equation \eqref{variational eq g0} on the scales $ \eps $ or $ \eps^{2/3} \Omega^{-1/3} |\log(\eps^4 \Omega)| $, if $ \omega \leq \omegac $ or $ \Omegac \leq \Omega \ll \eps^{-4} $ respectively. Note that in this second case the blow-up scale remains much smaller than $ (\eps \Omega)^{-2/3} |\log(\eps^4\Omega)|^{-1} $ which appears in the definition of $ \Aoo $.
\end{proof}

We conclude this section by stating another pointwise estimate of $ g_0 $ which shows the exponential smallness in $ \eps $ of $ g_0 $ outside of $ \tfsupp $.

	\begin{pro} [{\bf Exponential smallness of $ g_0 $}]
		\label{g0 exp small: pro}
		\mbox{}	\\
		There exists a constant $C>0$ such that, if $ |\log\eps| \ll \omega $ and $ \Omega \ll \eps^{-4} $ as $ \eps \to 0 $,
		\beq
			\label{g0 exp small}
			g_0^2 (x) \leq 
			\begin{cases}
				C  \exp \left\{ - \eps^{-1/2} \: \mathrm{dist}\lf(\xv, \tfsuppo\ri) \right\},	&	\mbox{if} \:\:\: |\log\eps| \ll \omega < \omegac,	\\
				C (\eps\Omega)^{2/3}  \exp \left\{ - \Omega^{1/2} \: \mathrm{dist}\lf(\xv, \tfsuppoo\ri) \right\},	&	\mbox{if} \:\:\: \Omegac \leq \Omega \ll \eps^{-4},
			\end{cases}
		\eeq
		for any $ \xv \in \R^2 $.
	\end{pro}

	\begin{proof}
		See, e.g., \cite[Propostion 2.2]{CRY}: For instance in the first case, i.e., when $ |\log\eps| \ll \omega < \omegac $, one can show that if $ x \geq \xout + \sqrt{\eps} $, the r.h.s. of \eqref{variational eq g0} is negative and $ C \exp\{ - \eps^{-1/2} (x - \xout) \} $ is a supersolution to \eqref{variational eq g0} for $ x \geq \xout $ and a suitably large constant $ C > 0 $. In the other case when $ \Omegac \leq \Omega \ll \eps^{-4} $ one can investigate \eqref{variational eq g0} separately for $ x \leq \xin $ and $ x \geq \xout $ and identify appropriate supersolutions there.
	\end{proof}

\subsection{Energy Asymptotics and Emergence of a Hole}

In this section we first prove Theorem \ref{lattice energy: teo}, i.e., the asymptotics of the GP ground state energy, and show how this implies a phase transition at $ \Osec $, that is the occurrence of a macroscopic hole in the condensate.

\begin{proof}[Proof of Theorem \ref{lattice energy: teo}]
	For the sake of brevity we omit the proof of \eqref{gpeo asympt} and focus on the regime $ \Omegac \leq \Omega \ll \eps^{-4} $. The strategy as well as most of the estimates are basically the same for \eqref{gpeo asympt} and \eqref{gpeoo asympt} but the vanishing of the TF density for $ x \leq \xin $ makes the second case slightly tougher to deal with.

	As usual the result follows from a comparison of appropriate upper and lower bounds to the GP energy $ \gpeoo $, but in order to simplify the analysis it is convenient to extract the leading order term $ \tfeoo $ from the outset. This can be done by using a very standard trick (see, e.g., \cite{LM}), i.e., a splitting of the energy: Exploiting the positivity of $ g_0 $, we define a function $ u_0 \in L^2_{\mathrm{loc}}(\R^2) $ as
	\beq
		\label{u0}
		\gpmoo(\xv) = : g_0(x) u_0(\xv),
	\eeq
	and compute, using the variational equation \eqref{variational eq g0} solved by $ g_0 $,
	\beq
		\label{gpeoo splitting}
		\gpeoo = \hgpeoo + \F[u_0],
	\eeq
	where $ \F $ is the weighted GL-type functional
	\beq
		\label{F functional}
		\F[u] : = \int_{\R^2} \diff \xv \: g_0^2 \lf\{ \half \lf| \lf( \nabla - i \aavoo \ri) u \ri|^2 + \eps^{-2} g_0^2 \lf(1 - |u|^2 \ri)^2 \ri\}.
	\eeq
	We can then focus on the reduced energy $ \F[u_0] $ and prove upper and lower bounds for it. Note that in the upper bound there is a normalization condition to fulfill, i.e.,
	\beq
		\label{u0 normalization}
		\int_{\R^2} \diff \xv \: g_0^2 |u_0|^2 = 1.
	\eeq
	In the lower bound on the opposite we are going to remove this restriction and minimize w.r.t.\ any $ u \in L^2(\Aoo) $, where $ \Aoo $ is given by \eqref{Aoo}.

	{\it Upper bound:} The trial function which is going to be used in the upper bound proof is very similar to the one in \cite[Proof of Proposition 4.1]{CY}: Roughly speaking we pick a function $ \trialu $ which is approximately a phase factor, i.e., such that $ |\trialu| \simeq 1 $, but at the same time contains vortices of unit degree on a regular lattice  intersected with the support of the TF density $ \tfmoo $. 
	More precisely we denote by $ \latt  $ a regular lattice   with fundamental cell $ \cell $ (triangular, square or hexagonal) and lattice spacing $ \ell $ and define
		\beq
		\label{trialu}
		\trialu(\xv) : = c \: \xi(\xv) \: v(\xv),
	\eeq
	where $ c \simeq 1 $ is a normalization constant to enforce \eqref{u0 normalization}, $ v $ a phase factor (here we denote a point $ \xv = (x_1,x_2) \in \R^2 $ by the complex notation $ \zeta = x_1 + i x_2 $) given by
	\begin{equation}
		\label{phase v}
		v(\xv) : = \prod_{\zeta_i \in \latt \cap \Aoo} \frac{\zeta - \zeta_i}{|\zeta - \zeta_i|} = : \exp\lf\{ i \varphi(\xv) \ri\},
	\end{equation}
	and $ \xi $ regularizes the singularities of $ v $ at the lattice points
	\begin{equation}
		\label{xi}
		\xi(\xv) : = 
		\begin{cases}
			1,				& 	\text{if} \:\:\: |\xv - \xv_i| > t, \:\: \text{for all } \xv_i \in \latt,	\\
 			t^{-1} |\xv - \xv_i|, 	&	\text{if} \:\:\: |\xv - \xv_i| \leq t.
		\end{cases}
	\end{equation}
	Here $ t > 0  $ is a variational parameter (vortex core radius), such that
	\beq
		\label{t condition 1}
		t \ll \Omega^{-1/2}.
	\eeq
	We also choose the lattice spacing $ \ell $ in such a way that 
	\begin{equation}
		\label{cell area}
		|\cell| = \pi {\Omega}^{-1},
	\end{equation}
	which guarantees the `neutrality' of each cell of the lattice in the `electrostatic analogy' discussed in \cite[Proof of Proposition 4.1]{CY}. Since $ |\tfsuppoo| = \OO( (\eps \Omega)^{-2/3}) $, the total number of lattice points inside $\tfsuppoo$ is $ \OO(\eps^{-2/3} \Omega^{-1/6}) \gg 1 $, as long as $ \Omega \ll \eps^{-4} $, and it is not difficult to verify that
	\beq
		\label{norm const c}
		c = 1 + \OO(\eps^{-2/3} \Omega^{-1/6} t^2) = 1 + o(1),
	\eeq
	thanks to \eqref{t condition 1}.

	The evaluation of $ \F[\trialu] $ can be decomposed into two different estimates of the kinetic energy and interaction energy respectively. We are going to show that the former satisfies the inequality 
	\beq
		\label{vortex energy}
		 \half \int_{\mathbb R^2} \diff \xv \: g_0^2 \lf| \lf( \nabla - i \aavoo \ri) \trialu \ri|^2 \leq \half \lf( 1 + o(1) \ri) \Omega \lf|\log \left( t^2 \Omega \right)\ri|.
 	\eeq
	A simple computation yields (recall that $ \xi $ is real)
	\bml{
 		\label{vortex en 1}
		\half \int_{\mathbb R^2} \diff \xv \: g_0^2 \lf| \lf( \nabla - i \aavoo \ri) \trialu \ri|^2  = \half c^2 \int_{\mathbb R^2} \diff \xv \: g_0^2 \xi^2 \lf| \nabla \varphi - \aavoo \ri|^2 + c^2 \sum_{\xv_i \in \latt} \int_{|\xv - \xv_i| \leq t}  \diff \xv \: g_0^2 |\nabla \xi|^2 \leq	\\
		\half (1 + o(1)) \int_{\At} \diff \xv \: g_0^2 \xi^2 \lf| \nabla \varphi - \aavoo \ri|^2 + C \Omega^{1/2} t + \OO(\eps^{\infty}),
	}
	where $\At $ is chosen in such a way that \eqref{g0 exp small} yields the pointwise bound\footnote{Note that \eqref{g0 exp small} also implies an exponential decay in $ x $ of $ g_0 $, which is used in order to estimate the remainder in \eqref{vortex en 1}.} $ g_0 = \OO(\eps^{\infty}) $ \footnote{The symbol $\OO (\ep ^{\infty})$ denotes a quantity vanishing faster than any 
power of $\ep$.} outside $ \At $.   A possible choice is
	\beq
		\label{At}
		\At : = \lf\{ \xv \in \R^2\,:\,\,  x_{\rm in}-\Omega^{-1/2}|\log(\eps^4\Omega)|^{2}\leq x\leq x_{\rm out}+ \Omega^{-1/2}|\log(\eps^4\Omega)|^{2} \ri\}.
	\eeq
	Now we can apply the electrostatic computation in \cite{CY} to estimate the leading term in \eqref{vortex en 1}: The first term on the r.h.s. of \eqref{vortex en 1} becomes
	\beq
		\label{vortex en 2}
 		\half \int_{\At} \diff \xv \: g_0^2 \xi^2 \lf|  \nabla \varphi - \aavoo \ri|^2 \leq \half \lf( 1 + \OO(t \Omega^{1/2}) \ri) \sum_{\xv_i \in \latt \cap \At} \sup_{\xv \in \celli} g_0^2(x) \lf[ \pi |\At | |\log(t^2\Omega)| + \OO(1) \ri],
	\eeq
	as in \cite[Eq. (4.37)]{CY}. In the above expression we have actually overestimated the contribution in $ \At \setminus \Aoo $ since $ v $ contains no vortices there and therefore one would easily get an upper bound without the logarithmic term in the cell contained in  $ \At \setminus \Aoo $.	
	
	Using a Riemann sum approximation for the integral of $ g_0^2 $ in $ \At $ we can replace the sum in \eqref{vortex en 2} and obtain
	\beq
		\half \int_{\At} \diff \xv \: g_0^2 \xi^2 \lf|  \nabla \varphi - \aavoo \ri|^2 \leq \half \lf( 1 +  o(1) \ri) \Omega \lf[ |\log(t^2\Omega)| + \OO(1) \ri],
	\eeq
	thanks to \eqref{t condition 1} and
	\beq
		\lf\| g_0 \ri\|^2_{L^2(\At)} \leq \lf\| \tfmoo \ri\|_{L^1(\Aoo)} + o(1) + \lf\| g_0 \ri\|^2_{L^2(\At\setminus\Aoo)} \leq 1 + o(1) + \OO(|\log(\eps^4\Omega)|^{-2}) = 1 + o(1),
	\eeq
	since $ g_0^2 \leq (\eps\Omega)^{2/3} |\log(\eps^4\Omega)|^{-1} $ in $ \At \setminus \Aoo $ and $ \Omega^{-1/2} |\log(\eps^4\Omega)|^{2} \ll (\eps\Omega)^{-2/3} |\log(\eps^4\Omega)|^{-1} $ for any $ \Omega \ll \eps^{-4} $.

	The upper bound \eqref{vortex energy} is then proven and it remains to compute the interaction energy of $ \trialu $: For any $ \Omegac \leq \Omega \ll \varepsilon^{-4}$ one can show that
    	\beq
		\label{interaction energy}
		\eps^{-2} \int_{\R^2}\diff \xv \: g_0^4  \lf(1 - \lf|\trialu\ri|^2\ri)^2 \leq C \eps^{-4/3} \Omega^{5/3} t^2.
	\eeq
	Indeed we note that $ |\trialu| = 1 $ in $ \R^2 \setminus \Aoo $, so that
	\bml{
 		\int_{\R^2}\diff \xv \: g_0^4  \lf(1 - \lf|\trialu\ri|^2\ri)^2 = \sum_{\xv_i \in \latt \cap \Aoo} \int_{\ba(\xv_i,t)} \diff \xv \: g_0^4 \lf( 1 - t^{-2} \lf| \xv - \xv_i \ri|^2 \ri)^2  \leq	\\
		C (\eps\Omega)^{4/3} t^2 |\cell|^{-1} |\Aoo| \leq C \eps^{2/3} \Omega^{5/3} t^2.
	}
	Combining \eqref{gpeoo splitting} with \eqref{vortex energy} and \eqref{interaction energy}, we finally obtain
	\bml{
 		\label{gpeoo ub}
 		\gpeoo \leq \hgpeoo +  \half \lf( 1 + o(1) \ri) \Omega \lf|\log \left( t^2 \Omega \right)\ri| + C \eps^{-4/3} \Omega^{5/3} t^2 + o(1) \leq	\\
		\tfeoo + \tx\frac{1}{6} \lf( 1 + o(1) \ri) \Omega |\log(\eps^4\Omega)| + \OO(\Omega).
	}	
	where we have chosen the vortex core radius $ t = \eps^{2/3} \Omega^{-1/3} $ and used \eqref{hgpeoo ub}. Note that our choice of $ t $ is compatible with \eqref{t condition 1} as long as $ \Omega \ll \eps^{-4} $. Moreover the radial kinetic energy of the profile $ g_0 $, which is $ \OO((\eps\Omega)^{4/3} |\log(\eps^4\Omega)|) $ according to \eqref{hgpeoo ub}, is much smaller than the remainder $ \OO(\Omega) $ if again $ \Omega \ll \eps^{-4} $.

	{\it Lower Bound:} A lower estimate matching \eqref{gpeoo ub} can be proven by following the same strategy as in \cite[Section 5]{CY}, i.e., first decomposing the support of $ \tfmoo $ into cells, then blowing-up the energy $ \F[u_0] $ on each cell and finally using a lower bound for the GL functional with an applied magnetic field between the second and third critical values (see, e.g., \cite{SS1,SS2}). 
	
	The starting point is a restriction of the integration in $ \F $ to the domain $ \Aoo $, which exploits the positivity of the integrand. Then we decompose $ \Aoo $ into cells: Let $ \tilde\latt $ be the lattice\footnote{For convenience we have chosen a square regular lattice but any regular lattice covering the space would have done the job.}
 	\beq
		\tilde\latt = \left\{ \xv_i = (m \tilde \ell, n \tilde \ell), \: m,n \in \mathbb Z \right\}
 	\eeq
	with cells $ \tcelli $ and lattice spacing satisfying
	\beq
		\label{spacing condition}
		\Omega^{-1} |\log(\eps^4\Omega)| \ll \tilde\ell^2 \ll \eps^{-4/3} \Omega^{-4/3},
	\eeq
	we easily obtain the lower bound
	\begin{multline}
		\label{cell energy 0}
		\F[u_0] \geq \sum_{\xv_i, \tcelli \subset \Aoo} \int_{\tcelli} \diff \xv \: g_0^2 \lf\{ \half \lf| \lf( \nabla - i \aavoo \ri) u_0 \ri|^2 + \eps^{-2} g_0^2 \lf(1 - |u_0|^2 \ri)^2 \ri\} \geq	\\
		(1-o(1)) \sum_{\xv_i, \tcelli \subset \Aoo} g_0^2(x_i) \: \F^{(i)} [u_0],
	\end{multline}
 	where we have used \eqref{point est g0 2} and \eqref{tfmoo Aoo lb}, which imply the lower bound
	\beq
		\label{g0 A00 lb}
		g_0^2(x) \geq C (\eps\Omega)^{2/3} |\log(\eps^4\Omega)|^{-1},	\hspace{1cm}	\mbox{for any} \:\: \xv \in \Aoo,
	\eeq
	and introduced the functional
	\beq
		\label{cell energy 1}
		\F^{(i)} [u_0] : = \int_{\tcelli} \diff \xv \lf\{ \half \lf| \lf( \nabla - i \aavoo \ri) u_0 \ri|^2 + \eps^{-4/3} \Omega^{2/3} |\log(\eps^4\Omega)|^{-1} \lf(1 - |u_0|^2 \ri)^2 \ri\}.
	\eeq
	Note that the conditions \eqref{spacing condition} on the spacing $ \tilde\ell $, which are essential for the proof, are compatible as long as $ \Omega \ll \eps^{-4} $: The upper bound on $ \tilde\ell $ is  needed  in order that the spacing is much smaller than the width of $ \tfsuppoo $ as required for the Riemann approximation we shall use.  On the other hand, the lower bound in  \eqref{spacing condition} implies that the cells of the lattice $ \tilde\latt $ are much larger than those of the lattice considered in the proof of the upper bound.  The reason is that we would like each cell of the new lattice $ \tilde\latt $ to contain a  large number of vortices, diverging to $\infty$ as $ \eps \to 0 $. In fact, this is crucial in order to make the whole proof strategy work since the GL estimate from \cite{SS1,SS2} that we are going to use holds only  under this condition.
	
	We now blow-up the energy in the cell $ \tcelli $ and set $ \sv : =  \tilde\ell^{-1} (\xv - \xv_i) $,
	\beq
		\tilde u_0(\sv) : = u_0 (\xv_i + \tilde\ell\xv), \hspace{1cm}	\taavoo(\sv) : = \tilde\ell \: \aavoo(\xv_i + \tilde\ell \sv),
	\eeq
	so that $ \F^{(i)} $ becomes
	\beq
		\label{cell energy 2}
		\F^{(i)} [u_0] = \tilde{\F}^{(i)} [\tilde u_0] : = \int_{\cell_1} \diff \sv \lf\{ \half \lf| \lf( \nabla - i \taavoo \ri) \tilde u_0 \ri|^2 + \epsilon^{-2} \lf(1 - |\tilde u_0|^2 \ri)^2 \ri\},
	\eeq
	where $ \cell_1$ stands for  a unit square centered at the origin and we have introduced a new  small parameter $ \epsilon $ defined as
	\beq
		\label{epsilon}
		\epsilon : = \tilde\ell^{-1} \eps^{2/3} \Omega^{-1/3} |\log(\eps^4\Omega)|^{1/2},
	\eeq	
	 which is much less than $(\eps^{4} \Omega)^{1/3} |\log(\eps^4\Omega)|^{1/2}\ll  1$ if $ \Omega \ll \eps^{-4} $ thanks to \eqref{spacing condition}.  The rescaled vector potential $ \taavoo $ is explicitly given by
	\beq
		\label{scaled vect}
		\taavoo(\sv) = \Omega \tilde\ell \: \mathbf{e}_z \wedge \xv_i + \Omega \tilde\ell^2 \: \mathbf{e}_z \wedge \sv, 
	\eeq
	and the corresponding magnetic field is 
	\beq
		\hex = \mathrm{curl}\, \taavoo = \Omega \tilde\ell^2.
	\eeq 
	Employing  gauge invariance one obtains the lower bound
	\beq
		\label{cell energy 3}
		\tilde{\F}^{(i)} [\tilde u_0] \geq \inf_{u \in H^1(\cell_1)} \int_{\cell_1} \diff \sv \lf\{ \half \lf| \lf( \nabla - i \hex \mathbf{e}_z \wedge \sv \ri) u\ri|^2 + \epsilon^{-2} \lf(1 - |u|^2 \ri)^2 \ri\},
	\eeq
	where the applied magnetic field $ \hex $ satisfies the conditions 
	\beq
		\label{magnetic field condition}
		|\log\epsilon| \ll \hex \ll \epsilon^{-2},
	\eeq
	thanks to \eqref{spacing condition}.

	We can now borrow a lower bound to the GL functional with an applied magnetic field between the first and second critical fields (see, e.g., \cite{SS1,SS2}) to estimate the r.h.s. of \eqref{cell energy 3} and obtain
	\beq
		\tilde{\F}^{(i)} [\tilde u_0] \geq  (1-o(1)) \hex \log\lf(\epsilon^{-1} \hex^{-1/2} \ri) \geq \tx\frac{1}{6} (1 - o(1)) \Omega \tilde\ell^2 |\log(\eps^4\Omega)|.
	\eeq
	Putting together the above estimate with \eqref{cell energy 0} and using again the Riemann sum approximation we finally get
	\beq
		\F[u_0] \geq  \tx\frac{1}{6} (1 - o(1))  \Omega |\log(\eps^4\Omega)| \lf\| g_0 \ri\|^2_{L^2(\Aoo)} \geq   \tx\frac{1}{6} (1 - o(1))  \Omega |\log(\eps^4\Omega)|,
	\eeq
	by \eqref{bulk g0 2}. The splitting \eqref{gpeoo splitting} as well as the trivial bound $ \hgpeoo \geq \tfeoo $ hence complete the  proof of the  lower bound.
\end{proof}

	The occurrence of a hole in the bulk of the condensate for $ \Omega > \Omegac $  now follows as a  straightforward consequence of the energy asymptotics proven in Theorem \ref{lattice energy: teo}:
	
\begin{proof}[Proof of Proposition \ref{hole: pro}]
	The proof is very similar to the proof of Proposition \ref{g0 exp small: pro} (see also \cite[Propostion 2.2]{CRY}): By setting $ \gpd : = |\gpmoo|^2 $ one easily obtains from \eqref{GP variationaloo} the inequality
	\beq
		\label{variational ineq gpd}
		- \half \Delta \gpd \leq 2 \eps^{-2} \lf[ \tgpd - \gpd \ri] \gpd,
	\eeq
	where
	\beq
		\label{tildegpd}
		\tgpd(x) : =	\half \lf[ \eps^2 \chemoo - \gamma \eps^2 \Omega^2 W(x) \ri].
	\eeq
	Now the energy asymptotics \eqref{gpeoo asympt} implies exactly as in the derivation of \eqref{l2 est g0 2} and \eqref{linfty est g0} that 
	\beq
		|\gpmoo| \leq C (\eps\Omega)^{2/3},	\hspace{1cm}	\lf\| \gpd - \tfmoo \ri\|_2^2 \leq C \eps^2 \Omega |\log(\eps^4\Omega)|,
	\eeq
	 which in turn (see, e.g., \cite[Proposition 2.1]{CRY}) yield the estimate
	\beq
		\lf| \chemoo - \tfchemoo \ri| \leq C \eps^{-2/3} \Omega^{5/6} |\log(\eps^4 \Omega)|^{1/2}.
	\eeq
	Now if $ \Omega > \Omegac $, one can pick some $ 0 < \delta \ll 1 $ so that the above inequality implies (recall \eqref{Taylor W})
	\beq
		\tgpd(x) \leq C \eps^{4/3} \Omega^{5/6} |\log(\eps^4 \Omega)|^{1/2} - C (\eps \Omega)^{2/3} \delta^2,
	\eeq
	for any $ 0 \leq x \leq \xin - \delta (\eps\Omega)^{-2/3} $. Hence if one takes for instance
	\beq
		\delta = (\eps^4\Omega)^{1/12} |\log(\eps^4\Omega)|
	\eeq
	 that is $ \ll 1$ as long as $ \Omega \ll \eps^{-4} $, the inequality \eqref{variational ineq gpd} becomes
	\beq
		- \half \Delta \gpd + C \eps^{-2/3} \Omega^{5/6} \gpd \leq 0,	\hspace{1cm}	\mbox{for any} \:\: 0 \leq x \leq \xin - \delta (\eps\Omega)^{-2/3}.
	\eeq
	To complete the proof it suffices therefore to note that the function $ C \exp \{ - \eps^{-1/3} \Omega^{5/12} (\xin - x) \} $ is a supersolution to the equation $ - \Delta f + C \eps^{-2/3} \Omega^{5/6} f = 0 $ in the ball $ 0 \leq x \leq \xin $ with boundary condition $ f = C $ on $ \partial \ba_{\xin} $. One thus obtains
	\beq
		\gpd(x) \leq C \lf\| \gpd \ri\|_{\infty} \exp \lf\{ - \eps^{-1/3} \Omega^{5/12} (\xin - x) \ri\},
	\eeq
	for $  0 \leq x \leq \xin $ but, if $ \Omega \geq \Omega_0 \eps^{-1} $ for $ \Omega_0 > \eps\Omegac $,  then $ \xin > C > 0 $ and one can always find some $ 0 < \xh < \xin $ such that the above upper bound yields the desired estimate.
\end{proof}

\subsection{Uniform Distribution of Vorticity}

The energy asymptotics proven in Theorem \ref{lattice energy: teo} is the main ingredient in the proof of the uniform distribution of vorticity:

\begin{proof}[Proof of Theorem \ref{vortex distribution: teo}]
	The proof follows very closely \cite[Proof of Theorem 1.1]{CPRY1}, which in turn relies on \cite[Proposition 5.1]{SS1}. For the sake of brevity we only enlighten the major differences and assume that $ \Omegac \leq \Omega \ll \eps^{-4} $.
	
	The starting point is the combination of the cell decomposition \eqref{cell energy 0} with the global upper bound \eqref{gpeoo ub}, which gives (recall that $ \tilde\latt $ is a regular square lattice with spacing $ \tilde\ell $ satisfying \eqref{spacing condition})
	\beq
		\label{energy localization 1}
		\sum_{\xv_i, \tcelli \subset \Aoo} g_0^2(x_i) \: \F^{(i)} [u_0] \leq \tx\frac{1}{6} (1 + \eta) \disp\sum_{\xv_i, \tcelli \subset \Aoo} g_0^2(x_i) \tilde\ell^2 \Omega |\log(\eps^4\Omega)|,
	 \eeq
	for some quantity $ \eta : = \eta(\eps,\Omega) $ such that $ \eta \to 0 $ as $ \eps \to 0 $, $\Omega\ll \eps^{-4}$. 
 
	In order to localize the above bound, one needs some control on the density $ g_0 $ and to this purpose  we set
	\beq
		\label{gamma}
		\beta : = \beta(\eps,\Omega) = \max\lf( |\log\eta|^{-1}, \: |\log(\eps^4\Omega)|^{-1} \ri),
	\eeq
	and define  $ {\mathcal R}_{\rm bulk} $ as in \eqref{bulk sub}, i.e., $ {\mathcal R}_{\rm bulk} = \{\xv:\, \xin + (\eps\Omega)^{-2/3} \beta \leq  x \leq \xout -  (\eps\Omega)^{-2/3} \beta \} $, so that $ {\mathcal R}_{\rm bulk} \subset \Aoo $ and
	\beq
		g_0^2(x) \geq C (\eps\Omega)^{2/3} \beta,	\hspace{1cm}	\mbox{for any} \:\: \xv \in {\mathcal R}_{\rm bulk}.
	\eeq
	Such a bound unfortunately is not sufficient to complete the localization of the energy and a certain amount of bad cells has to be rejected: We say that a cell $ \tcelli $ is good if 
	\beq
		\label{good cells sub}
		 \F^{(i)} [u_0] \leq \tx\frac{1}{6} \lf[1 + \min\lf(\sqrt{\eta}, \: \eps^2\sqrt{\Omega}\ri) \ri] \Omega \tilde\ell^2 |\log(\eps^4\Omega)|.
	\eeq
	Now given any set $ \set \subset {\mathcal R}_{\rm bulk} $ with area much larger than $ \tilde\ell^2 $ \eqref{energy localization 1} implies that the total number of bad cells contained inside $ \set $ (denoted by $ N_{B} $) is much smaller than the total number of cells $ N $, i.e., more precisely
	\beq
		N_{B} \leq \sqrt{\eta} \beta^{-1} N \leq \min\lf[ \sqrt{\eta} |\log\eta|, \: \eps^2\sqrt{\Omega} |\log(\eps^4\Omega)| \ri] N \ll N.
	\eeq
	On the other hand one can use the upper bound \eqref{good cells sub} inside good cells to prove the final result exactly as in \cite[Proof of Theorem 1.1]{CPRY1}. Note that the choice $ \tilde\ell = \Omega^{-1/2} |\log(\eps^4\Omega)| $ is compatible with \eqref{spacing condition} and yields the lower bound on the measure of $ \set $, i.e., $ |\set| \gg \tilde\ell^2 = \Omega^{-1} |\log(\eps^4\Omega)|^2 $.	
\end{proof}

\section{The Giant Vortex Regime $ \Omega \sim \eps^{-4} $}\label{sec: GV}

Since throughout this section we deal only with the regime $ \Omega \sim \eps^{-4} $, we remove any suffix from $ \gpmoo $, $ \gpeoo $, etc. for the sake of simplicity and denote those quantities by $ \gpm $, $ \gpe $ etc..

\subsection{Preliminary Estimates for $ \gpe $ and $ \gpm $}
\label{sec: GV GP}

We first introduce and recall some notation: As in \eqref{alpha} we define a parameter $ \alpha $ as
\bdm
	\al^2 : = 4 + \gamma (s-2),
\edm
and a vector field $ \mathbf{B} $ by
\begin{equation}\label{defi B first}
\mathbf{B} = \left(\Om r - \lfloor \Om \rfloor r^{-1} \right) \mathbf{e}_{\vartheta}. 
\end{equation}
A Taylor expansion of the potential $ U(x) $ (see \eqref{pot U}) around $ x = 1 $ (e.g., for $ \half \leq x \leq \hbox{$\frac{3}{2}$} $) yields 
\bml{
	\label{Taylor U}
	U(x)  = \half B^2(1) + \lf( 1 - \into^2 \Omega^{-2} \ri) \lf(x - 1\ri) + \half \lf[ 1 + 3 \into^2 \Omega^{-2} + \gamma(s-2) \ri] \lf(x-1\ri)^2 + \OO\lf( (x-1)^3 \ri) =	\\
	\half \al^2 (x-1)^2 + \OO(\Omega^{-2}) + \OO\lf( (x-1)^3 \ri).
}
Moreover we denote by $ \gosc $ the normalized ground state of the one-dimensional harmonic oscillator ($ y \in \R $)
\beq
	\label{1d osc}
	\hosc : = - \half \partial_y^2 + \half \al^2 y^2,
\eeq
i.e.,
\beq
	\label{osc gs}
	\gosc(y) : = \pi^{-1/4} \al^{1/4} \exp \lf\{ - \half \al y^2 \ri\}.
\eeq
The starting point is a rough but useful upper bound on $  \gpe $ which has several important consequences:
	
	\begin{pro}[\textbf{Estimates for $ \gpe $ and $ \gpm $}]
		\label{gp prel est: pro}
		\mbox{}	\\
		If $ \Omega = \Omega_0 \eps^{-4} $ with $  \Omega_0 > 0  $ and $ \eps $ is sufficiently small,
		\beq
			\label{gpgs ub gv}
			\gpe \leq \OO(\Omega),	\hspace{1cm}	\lf\| \gpm \ri\|_{\infty}^{2} \leq \OO(\eps^2 \Omega).
		\eeq
		Moreover there exists a finite constant $ C $ such that, for any $ \xv \in \R^2 $,  
		\beq
			\label{gpm exp small gv}
			\lf| \gpm(\xv) \ri|^2 \leq C \eps^{-2} \exp \lf\{- \Omega^{1/2} \lf|1 - x \ri| \ri\}.
		\eeq
	\end{pro}

	\begin{proof}
		In order to prove the upper bound on $ \gpe $, it suffices to test the GP functional on a suitable rescaling of the harmonic oscillator ground state \eqref{osc gs} times a phase containing a giant vortex of degree $ \into $ at the origin. More precisely we set
		\beq
			\label{gp trial gv}
			\trial(\xv) : = \ftrial(x) e^{i \into \vartheta},
		\eeq
		\beq
			\label{ftrial}
			\ftrial(x) : = c_{\Omega}
 				\begin{cases}
					\Omega^{1/4} \gosc\lf(\Omega^{1/2} (1-x)\ri),	&	\mbox{if} \:\: x \geq 1 - \eps,	\\
					f(x),					&	\mbox{if} \:\: 1 - 2 \eps \leq x \leq 1 - \eps,	\\
					0,					&	\mbox{if} \:\: 0 \leq x \leq 1 - 2\eps,
				\end{cases}
		\eeq
where $ c_{\Omega} $ is a normalization constant and $ f $ a smooth function such that
		\bdm
			f(1 - 2 \eps) = 0,	\hspace{1cm}	f(1 - \eps) = \gosc(\Omega^{1/2} \eps) = \OO(\eps^{\infty}),
		\edm
		which also allows to assume that $ \lf\|\nabla f \ri\|_{\infty} = \OO(\eps^{\infty}) $. A simple computation gives
		\bmln{
			1 = \lf\| \trial \ri\|_2^2 = 2 \pi c_{\Omega}^2 \Omega^{1/2} \int_{1 - \eps}^{\infty} \diff x \: x \: \gosc^2(\Omega^{1/2}(x-1)) + \OO(\eps^{\infty}) = 	\\
			2 \pi c_{\Omega}^2 \int_{- \eps \Omega^{1/2}}^{\infty} \diff y (1 + \Omega^{-1/2} y) \gosc^2(y) + \OO(\eps^{\infty}) = 2 \pi c_{\Omega}^2 \int_{- \infty}^{\infty} \diff y \:  \gosc^2(y) + \OO(\Omega^{-1/2}),
		}
		so that
		\beq
			\label{normalization gv}
			c_{\Omega}^2 = (2 \pi)^{-1} \lf( 1 + \OO(\Omega^{-1/2}) \ri).
		\eeq
		By using \eqref{Taylor U} and the exponential decay of $ \gosc $ one has	
		\bml{
 			\label{gp ub gv}
			\gpfoo\lf[\trial\ri] = 2\pi c_{\Omega}^2 \int_{0}^{\infty} \diff x \: x \lf\{ \half \lf( \ftrial^{\prime} \ri)^2 + \Omega^2 U(x) \ftrial^2 + c_{\Omega}^2 \eps^{-2} \ftrial^{4} \ri\} =	\\
			2\pi c_{\Omega}^2 \int_{1 - \eps}^{1 + \eps} \diff x \: x \lf\{ \half \lf( \ftrial^{\prime} \ri)^2 + \Omega^2 U(x) \ftrial^2 + c_{\Omega}^2 \eps^{-2} \ftrial^{4} \ri\} + \OO(\eps^{\infty}) =	\\
			2 \pi c_{\Omega}^2 \Omega \int_{- \eps \Omega^{1/2}}^{\eps \Omega^{1/2}} \diff y (1 + \Omega^{-1/2} y) 	\lf\{ \half \lf( \gosc^{\prime} \ri)^2 + \half \al^2 y^2 \gosc^2 + c_{\Omega}^2 \eps^{-2} \Omega^{-1/2} \gosc^{4} \ri\} + \OO(\Omega^{-1}) + \OO(\Omega \eps^3) =	\\
			2 \pi c_{\Omega}^2 \Omega \lf[ \frac{\al}{2} + c_{\Omega}^2 \sqrt{\frac{\al}{2 \pi \Omega_0} } + \OO(\Omega^{-1/2}) \ri] = \Omega \lf[ \frac{\al}{2} + \frac{1}{2\pi} \sqrt{\frac{\al}{2 \pi \Omega_0} } + \OO(\Omega^{-1/2}) \ri].
		}
		The second statement in \eqref{gpgs ub gv} is a consequence of the variational equation \eqref{GP variationaloo} which yields 
		\beq
			\label{gp ineq gv}
			- \half \Delta \gpd \leq  \eps^{-2} \lf[ \eps^2 \chemGP  -  \gamma \eps^2\Omega^2 W  - 2\gpd \ri] \gpd \leq \eps^{-2} \lf[ \eps^2 \chemGP - 2\gpd \ri] \gpd,
		\eeq
		where 
		\beq
			\label{gpdoo}
			\gpd : = \lf| \gpm \ri|^2
		\eeq
		and we have used the positivity of $ W $. Now since $ \Delta \gpd \leq 0 $ at any maximum point of $ \gpd $, it immediately follows that 
		\beq
			\label{sup est gpm}
			\lf\| \gpd \ri\|_{\infty} \leq \half \eps^2 \chemGP,
		\eeq
		which together with the definition of the chemical potential implies
		\bdm
			\chemGP = \gpe + \eps^{-2} \lf\| \gpm \ri\|^4_4 \leq \gpe + \half \chemGP,
		\edm
		so that by \eqref{gp ub gv} 
		\beq
			\label{chemoo ub}
			\chemGP \leq 2 \gpe \leq \OO(\Omega) 
		\eeq
		and thus \eqref{gpgs ub gv} is proven.

		In order to prove the pointwise estimate \eqref{gpm exp small gv} we exhibit an explicit supersolution to the variational equation \eqref{GP variationaloo}.
		We first consider a ball $ \ba(1- a \eps^2) $ centered at the origin of radius $ 1 - a \eps^2 $, with $ a = \OO(1) $ sufficiently large: The monotonicity of $ W $ together with the Taylor expansion $ W(x) = \half (s-2)(1 - x)^2 + \OO\lf( (1 - x)^3\ri) $, yields the bound
		\bdm
			W(x) \geq C a^2 \eps^4,
		\edm
		for some $ C > 0 $ and for any $ 0 \leq x \leq 1 - a \eps^2 $. Therefore inside $ \ba(1 - a \eps^2) $ one has by \eqref{gp ineq gv} and \eqref{chemoo ub}
		\bdm
			- \half \Delta \gpd \leq  \eps^{-2} \lf[ \eps^2 \chemGP  -  C a^2 \eps^6 \Omega^2 \ri] \gpd \leq - \half C  a^2 \eps^4 \Omega^2 \gpd,
		\edm
		if $ a $ is taken sufficiently large. In the same region the function
		\beq
			\label{fsup}
			\fsup(x) : = C_{a} \lf\| \gpd \ri\|_{\infty} \exp \lf\{ - \Omega^{1/2}(1-x) \ri\},
		\eeq
		satisfies
		\bdm
			- \half \Delta \fsup + \half C a^2 \eps^4 \Omega^2 \fsup = \half \lf[ - \Omega^{1/2} - \Omega x  +  C  a^2 \eps^4 \Omega^2 \ri] \fsup > 0,
		\edm
		if $ a $ is again large enough. Therefore $ \fsup $ provides a supersolution, since at the boundary $ \partial \ba(1 - a \eps^2) $ one has
		\bdm
			\fsup(1 - a \eps^2) = C_{a} \lf\| \gpd \ri\|_{\infty} \exp\lf\{- a \Omega_0^{1/2} \ri\} = \lf\| \gpd \ri\|_{\infty} 
		\edm
		if $ C_{a} $ is taken equal to $  \exp\{a \Omega_0^{1/2}\} $. In conclusion, by the maximum principle (see, e.g., \cite[Exercise 2, p. 317]{T} or \cite[Theorem 1, p. 508]{E}) 
		$ \gpd(x) \leq \fsup(x) $ for any $ \xv \in \ba_{1 - a \eps^2} $ but the result can be trivially extended to the whole ball 
		$ \ba_1 $ thanks to the monotonicity of $ \fsup $.
		
		To complete the proof of the pointwise estimate \eqref{gpm exp small gv}, one has simply to repeat the argument above in 
		the complement of the region $ \ba_{1 + a\eps^2} $, exploiting the monotonicity of $ W $ there, and replace the supersolution $ \fsup $ with the function
		\bdm
			C_{a} \lf\| \gpd \ri\|_{\infty} \exp \lf\{ - \Omega^{1/2}(x-1) \ri\}.
		\edm
	\end{proof}

	The decay of $ \gpm $ outside of the bulk of the condensate can be estimated more precisely by exploiting the Taylor expansion \eqref{Taylor W}:
	
	\begin{pro}[\textbf{Exponential decay of $ \gpm $}]\label{pro: gpm exp decay}
		\mbox{}	\\
		 If $ \Omega = \Omega_0 \eps^{-4} $ with $ \Omega_0 > 0 $ and $ 1 \ll \lambda \ll |\log\eps|^{b} $ for some $b>0$ as $ \eps \to 0 $, there exist 
		 two finite constants $ c > 0 $ and $ C < \infty $ independent of $ \eps $ and $ \lambda $ such that
		\beq
			\label{gpm exp decay}
			\lf| \gpm(\xv) \ri|^2 \leq C \eps^{-2} \exp\lf\{ - c \Omega \lf(|1 - x| - \eps^2 \lambda \ri)^2 \ri\},
		\eeq
		for any $ \xv $ satisfying the condition
		\beq
			\label{gpm decay cond}
			\eps^2 \lambda \leq |1 - x| \leq \eps^2 \lambda^{a},	\hspace{1cm}	\mbox{with} \:\: a \geq 2.
		\eeq
	\end{pro}

	\begin{proof}
		We first consider the region $ x \leq 1 - \eps^2 \lambda $ and prove that the function
		\beq
			\fsup(x) : = \lf\| \gpd \ri\|_{\infty} \lf[ \exp\lf\{ - c \Omega (1 - \eps^2 \lambda - x)^2 \ri\} +  \exp\lf\{ - \sqrt{\Omega_0} \lambda^{2a+1} \ri\} \ri]
		\eeq
		provides a supersolution to \eqref{gp ineq gv} in the region $ \eps^2 \lambda \leq 1-x \leq \eps^2 \lambda^{2a+1} $, for any $ a \geq 2 $:
		\bml{
 			\label{gpm var ineq 1}
 			\lf[ - \half \Delta + \gamma \Omega^2 W(x)  + 2 \eps^{-2} \fsup^2(x) - \chemGP \ri] \fsup \geq	\\
			\lf[ \half c \Omega + \half (s - 2 - c^2) \Omega^2 (1 - x)^2 - C \eps^{-2} \lambda^{6a+3} - C \Omega \ri] \fsup \geq	\\
			 \eps^{-4} \lf[ \half (s - 2 - c^2) \Omega_0^2 \lambda^2 - C \Omega_0 \ri] \fsup \geq 0,
		}
		if $ c^2 < s - 2 $ and $ \eps $ is sufficiently small (recall that by assumption $ \lambda \to \infty $ as $  \eps \to 0 $). Moreover at the boundary 
		\beq
			\fsup(1-\eps^2 \lambda) \geq  \lf\| \gpd \ri\|_{\infty} \geq \gpd(1-\eps^2\lambda),	\hspace{1cm}	\fsup(1-\eps^2\lambda^{2a+1}) \geq \gpd(1-\eps^2\lambda^{2a+1}),	
		\eeq
		thanks to \eqref{gpm exp small gv}. Therefore $ \gpd \leq \fsup $ and in the region \eqref{gpm decay cond}
		\bdm
			 \exp\lf\{ c \Omega (1 -\eps^2 \lambda - x)^2 - \sqrt{\Omega_0} \lambda^{2a+1} \ri\} \leq \exp\lf\{ - \sqrt{\Omega_0} \lambda^{2a} (\lambda - C) \ri\} = o(1).
		\edm
		The proof for $ x \geq 1 + \eps^2 \lambda $ is identical.
	\end{proof}

\subsection{The Giant Vortex Density Profile}\label{sec: gv dens profile}

In this section we investigate the properties of the giant vortex profile and the associated energy functional defined in \eqref{gvf}.   

\begin{pro}[\textbf{Minimization of $ \gvf $}]
		\label{gvf minimization: pro}
		\mbox{}	\\
		There exists a minimizer $ \gvm $ of \eqref{gvf} that is unique up to a sign, radial and can be chosen to be strictly 
		 positive for $ \xv \neq 0 $. It solves the variational equation
		\beq
			\label{gvm variational eq}
			-\half \Delta \gvm + \Omega^2 U(x) \gvm + 2 \eps^{-2} \gvm^3 = \gvchem \gvm,
		\eeq
		with $ \gvchem = \gve + \eps^{-2} \lf\| \gvm \ri\|_4^4 $. Moreover $ \gvm $ has a unique maximum and, if $ \Omega = \Omega_0 \eps^{-4} $ with $ \Omega_0 > 0 $,
		\beq
			\label{gve asympt}
			\gve = \Omega \lf[ \frac{\al}{2} + \frac{1}{2\pi} \sqrt{\frac{\al}{2 \pi \Omega_0} } + \OO(\Omega_0^{-3/4}) + \OO(\Omega^{-1/2}) \ri].
		\eeq 	
		In addition there exists a finite constant $ C $ such that
		\beq
			\label{linfty est gvm}
			\sup_{x\in\mathbb R^2}\lf | \sqrt{2\pi} \Omega^{-1/4} \gvm(x) -  \gosc\lf(\Omega^{1/2}(1-x)\ri) \ri | \leq C \Omega_0^{-1/4} \eps^{-1},
		\eeq
		Finally $ \gvm $ decays exponentially far from $ x = 1 $: If $1\ll \lambda\ll |\log\eps|^b$ for some $b>0$ there exist constants $ c > 0 $ and $C<\infty$ independent of $\eps$ and $\lambda$ such that, for any $ \eps^2 \lambda \leq |1-x| \leq \eps^2 \lambda^a $, $ a \geq 2 $,
		\beq
			\label{gvm exp decay}
			\gvm^2(x) \leq C \eps^{-2} \exp\lf\{ - c \Omega \lf(|1 - x| - \eps^2 \lambda \ri)^2 \ri\}.
		\eeq		
	\end{pro}

	\begin{proof}
		All the properties of $ \gvm $ can be deduced by standard arguments (see, e.g., \cite[Proposition 2.3]{CRY} or \cite[Proposition 4.1]{CPRY1}). For instance the existence of a unique maximum point of $ \gvm $ can be proven by an adaptation of the argument in \cite[Proposition 2.2]{CPRY1}, while the pointwise estimate \eqref{gvm exp decay} can be proven exactly as \eqref{gpm exp decay}.
		
		The energy upper bound in \eqref{gve asympt}, i.e.,
		\beq
			\label{gve ub}
			\gve \leq \Omega \lf[ \frac{\al}{2} + \frac{1}{2\pi} \sqrt{\frac{\al}{2 \pi \Omega_0} } + \OO(\Omega^{-1/2}) \ri],
		\eeq
		has already been proven in \eqref{gp ub gv}. In order to prove a lower bound matching with \eqref{gp ub gv}, one needs first to show that
		\beq
			\label{l2 est gvm}
			\lf\| \gvm(x) - (2\pi)^{-1/2} \Omega^{1/4} \gosc\lf(\Omega^{1/2}(1-x)\ri) \ri\|_{L^{2}(\R^2)} = \OO(\Omega_0^{-1/4}).
		\eeq
		After a rescaling of all lengths in the functional \eqref{gvf}, i.e., setting $ y : = \Omega^{1/2} (1-x) $, and denoting
		\beq
			\tgvm(y) : = (2\pi)^{1/2} \Omega^{-1/4} \gvm\lf(1-\Omega^{-1/2}y\ri),
		\eeq	
		one obtains
		\bml{
 			\label{gve lower bound}
 			\gve = \gvf\lf[\gvm\ri] \geq \Omega \lf(1 - \OO(\Omega^{-1/2}) \ri) \int_{-\infty}^{
\frac{1}{2}\Omega^{1/2}} \diff y \lf\{ \half \lf| \nabla \tgvm \ri|^2 + \half \al^2 y^2 {\tgvm}^2 + (2\pi)^{-1} \Omega_0^{-1/2} {\tgvm}^{4} \ri\} \geq	\\
			\Omega \lf(1 - \OO(\Omega^{-1/2}) \ri) \bra{\tgvm} \hosc \ket{\tgvm}
		}
		where we have used the Taylor expansion \eqref{Taylor U} and $\hosc$ is defined in \eqref{1d osc}. Strictly speaking the 
		expectation value on the r.h.s. of the above expressions is computed in $ L^2(-\infty, \half\Omega^{1/2}) $ but, exploiting 
		the exponential smallness \eqref{gvm exp decay} of $ \gvm $ at $ y = \half \Omega^{1/2} $, it is not difficult to slightly extend
		$ \tgvm $ smoothly  to $ y $ larger than $  \half \Omega^{1/2} $ in such a way that the mean value can be computed in 
		$ L^2(\R) $. The error $ \OO(\eps^{\infty}) $ due to this procedure can safely be included in the remainder 
		$ \OO(\Omega_0^{-1/2}) $. 
		
		We now decompose $ \tgvm $ in a Fourier series in terms of the normalized eigenfunctions $ g_n(y) $, $ n \in \N $, of the harmonic oscillator \eqref{1d osc} (with $ g_0 = \gosc $), i.e., set
		\beq
			\label{tgvm Fourier series}
			\tgvm(y) = \sum_{n=0}^{\infty} a_n g_n(y),
		\eeq
		and use the upper bounds above to get
		\beq
 			\frac{\al |a_0|^2}{2} + \frac{3\al}{2} \sum_{n=1}^{\infty} \lf| a_n \ri|^2 \leq  \al \sum_{n=0}^{\infty} \lf(n + \half\ri)  \lf| a_n \ri|^2 =	\bra{\tgvm} \hosc \ket{\tgvm} \leq \frac{ \al}{2} + \frac{1}{2\pi} \sqrt{\frac{\al}{2 \pi \Omega_0} } + \OO(\Omega^{-1/2}),
		\eeq
		which immediately implies
		\beq
			\label{gvm fourier bound}
			\sum_{n = 1}^{\infty} \lf| a_n \ri|^2 \leq C \Omega_0^{-1/2},
		\eeq
		since $ \sum |a_n|^2 = 1 - \OO(\eps^{\infty}) $ by \eqref{gvm exp decay}. Therefore \eqref{l2 est gvm} is proven and, since $ \lf\| \tgvm \ri\|_{\infty} \leq C \Omega_0^{1/2} $ (see below), one also has
		\beq
			\lf\| \tgvm \ri\|_4^4 \geq \lf\| \gosc \ri\|_4^4 - C \Omega_0^{1/2} \lf\| \tgvm - \gosc \ri\|_2 - \OO(\eps^{\infty}) \geq \lf\| \gosc \ri\|_4^4 - C  \Omega_0^{1/4},
		\eeq
		which yields \eqref{gve asympt}. 	
		
		In fact \eqref{gvm fourier bound} has much stronger consequences than \eqref{l2 est gvm} and, if we define for any $ g \in L^2(\R) $
		\beq
			\label{osc norm}
			\lf\| g \ri\|^2_{\mathrm{osc}} : = \bra{g} \hosc \ket{g},
		\eeq
		one has
		\beq
			\label{linfty est tgvm}
 			\lf\| \tgvm - \gosc \ri\|^2_{\mathrm{osc}} = \bra{\tgvm} \hosc \ket{\tgvm} - \al \braket{\tgvm}{\gosc} + \half \al \leq C \Omega_0^{-1/2}, 
		\eeq
		where we have used \eqref{gve asympt} and the inequality
		\bdm
			\braket{\tgvm}{\gosc} \geq 1 - C \Omega_0^{-1/2},
		\edm
		which is a trivial consequence of \eqref{l2 est gvm}. On the other hand the Sobolev embedding theorem in $ \R $ immediately 
		implies
		\bdm
			\lf\| g \ri\|_{L^{\infty}(\R)} \leq C \lf\| g \ri\|_{H^1(\R)} \leq C \lf\| g \ri\|_{\mathrm{osc}},
		\edm
		and therefore \eqref{linfty est gvm} follows from \eqref{linfty est tgvm}. Note that the prefactor in \eqref{linfty est gvm}
		is bounded independently of $ \Omega_0 $ because it behaves as $\Omega_0^{-1/4} \Omega^{1/4} $.
	\end{proof}

For technical reasons which will be clearer later we also consider a functional with a different integration domain, i.e., \beq
	\label{annulus at}
	\anne : = \lf\{ \xv \in \R^2 : \: 1 - \eps^2 \ete \leq x \leq 1 + \eps^2 \ete \ri\},
\eeq
where 
\beq
	\label{ete cond 1}
	|\log\eps| \ll \ete \ll |\log\eps|^b\eeq
with $b>1$ is a parameter which will be fixed later as $\ete=|\log\eps|^{3/2}$. Note that since $ \ete \gg 1 $ the domain $ \anne $ expands on a scale $ \eps^{-2} $, 
which  ensures that it contains the bulk of the mass.
We define
\begin{equation}
	\label{domain ete}
	\gvdome := \left\{  f \in H^1 (\anne) : \: f= f ^*, \: \lf\| f \ri\|_{L^2 (\anne)} = 1 \right\}
\end{equation}
and set, for any $ f \in \gvdome $,
\beq
	\label{gvfe}
	\gvfe[f] : = \int_{\anne} \diff \xv \lf\{ \half \lf| \nabla f \ri|^2 + \Omega^2 U(x) f^2 + \eps^{-2} f^{4} \ri\}.
\eeq
We recall that 
\[
	U(x) =  \half B^2(x) + \gamma W(x) = \half \lfloor \Omega \rfloor^2 \Omega^{-2} x^{-2} + \half x^2 + \salf \gamma \lf( x^s - 1 \ri) - \half \gamma \lf( x^2 - 1 \ri) - \lfloor \Omega \rfloor \Omega^{-1} \geq 0.
\]
The ground state energy is
\beq
	\label{gvee}
	\gvee := \inf_{f \in \gvdome} \gvfe[f]
\eeq
and we denote  by $ \gvme $ any associated minimizer.

	\begin{pro}[\textbf{Minimization of $ \gvfe $}]
		\label{gvfe minimization: pro}
		\mbox{}	\\
		There exists a minimizer $ \gvme $ of \eqref{gvee} that is unique up to a sign, radial and can be chosen to be strictly 
		positive for $ \xv \neq 0 $.   Inside $ \anne $ it solves the variational equation
		\beq
			\label{gvme variational eq}
			-\half \Delta \gvme + \Omega^2 U(x) \gvme + 2 \eps^{-2} \gvme^3 = \gvcheme \gvme,
		\eeq
		with boundary conditions $ \gvme^{\prime}(1 \pm \eps^2 \ete) = 0 $ and $ \gvcheme = \gvee + \eps^{-2} \lf\| \gvme \ri\|_4^4 $. Moreover $ \gvme $ has a unique global maximum at $ \gvmaxe $.
		If $ \Omega = \Omega_0 \eps^{-4} $ with $ \Omega_0 > 0 $, then $ \lf\| \gvme \ri\|_{\infty}^{2} = \OO(\eps^2 \Omega) $ 
		and 
		\beq
			\label{gvee asympt}
			\gvee =  \lf(1 + \OO(\eps^{\infty}) \ri) \gve.	
		\eeq
	\end{pro}

	\begin{proof}
		The only result which deserves a discussion is the energy upper bound \eqref{gvee asympt} since anything else can be dealt with as in the proof of Propositions \ref{gp prel est: pro} and \ref{gvf minimization: pro}. However it suffices to test the functional $ \gvfe $ on a trial function of the form $ c_{\eps} \gvm $, where $ c_{\eps} $ ensures the normalization in $ L^2(\anne) $, to get
		\bdm
			\gvfe\lf[ c_{\eps} \gvm \ri] \leq c_{\eps}^4 \gvf\lf[ \gvm \ri] = c_{\eps}^4 \gve,
		\edm
		since $ c_{\eps} \geq 1 $ and the energy is positive so that we can extend the integration domain from $ \anne $ to the whole of $ \R^2 $. It remains then to estimate the normalization constant, but the exponential decay \eqref{gvm exp decay} together with the conditions \eqref{ete cond 1} on $ \ete $ guarantees that $ c_{\eps} = 1 + \OO(\eps^{\infty}) $.
	\end{proof}

In the next proposition we prove the analogue of \eqref{gvm exp decay}, i.e., an estimate of the decaying rate of $ \gvme $. Actually we also state more refined pointwise estimates of $ \gvme $ which will be crucial in the proof of the absence of vortices.

	\begin{pro}[\textbf{Pointwise estimates for $ \gvme $}]
		\label{gvme point est: pro}
		\mbox{}	\\
		If $ \Omega = \Omega_0 \eps^{-4} $ with $ \Omega_0 > 0 $, there exists constants $ C,c > 0 $, such that for any  $x$ with $ |1-x| \geq \eps^2 \ete^{1/4} $,
		\beq
			\label{gvme exp decay}
			\gvme^2(x) \leq C \eps^{-2} \exp\lf\{ - c \Omega \lf(|1 - x| - \eps^2 \ete^{1/4} \ri)^2 \ri\}.
		\eeq		
		Moreover if in addition $ \Omega_0 > \bar{\Omega}_0 $ with $ \bar{\Omega}_0 = \OO(1) $ large enough but independent of $ \eps $, then 
		\beq
			\label{gvme point est}
			\gvme(x) = \lf(1 + \OO(\Omega_0^{-1/4}) \ri) (2\pi)^{-1/2} \Omega^{1/4} \: \gosc \lf(\Omega^{1/2}(1-x)\ri)
		\eeq
		uniformly in $ \xv \in \anns $.
	\end{pro}

	\begin{rem}({\it Maximum of $ \gvme $})		
		\mbox{}	\\
		A straightforward consequence of \eqref{gvme point est} is the estimate
		\beq
			\label{gvme sup est}
			\lf\| \gvme \ri\|^2_{\infty} \leq \half \pi^{-3/2} \sqrt{\al\Omega} \lf(1 + c\Omega_0^{-1/4} \ri).
		\eeq
		On the other hand since $ \gvme^{\prime}(\gvmaxe) = 0  $ and $ \gvme^{\prime\prime}(\gvmaxe) \leq 0 $, the variational equation \eqref{gvme variational eq} implies the following estimate of the maximum position
		\bdm
			\lf(1 - \gvmaxe \ri)^2 \leq 2 \lf(\al \Omega\ri)^{-2} \lf[ \gvcheme - 2 \eps^{-2} \gvme^2(\gvmaxe) \ri] \leq \lf(\al \Omega\ri)^{-1} \lf( 1 + C \Omega_0^{-1/2} \ri).
		\edm
		In fact, \eqref{linfty est gvme} below yields a better estimate:
		\bml{
 			\half \pi^{-3/2} \sqrt{\al\Omega} \lf(1 - C \Omega_0^{-1/4} \ri) = \lf(1 - C \Omega_0^{-1/4}) \ri) (2\pi)^{-1/2} \Omega^{1/4} \: \gosc^2(0) \leq \gvme^2(1) \leq 	\\
			\gvme^2(\gvmaxe) \leq \half \pi^{-3/2} \sqrt{\al\Omega} \lf(1 + C \Omega_0^{-1/4} \ri) \exp\lf\{ - \al \Omega \lf( 1 - \gvmaxe \ri)^2 \ri\},
		}
		which implies
		\beq
			\label{gvmaxe est}
			\lf| 1 - \gvmaxe \ri| \leq C \Omega_0^{-5/8} \eps^2.
		\eeq
	\end{rem}

	\begin{proof}[Proof of Proposition \ref{gvme point est: pro}]
		The proof of \eqref{gvme exp decay} is basically identical to the proof of \eqref{gvm exp decay}. The only difference is 
		due to the fact that $ \gvme $ does not vanish at the boundary of $ \anne $. This implies that any legitimate supersolution 
		to the variational equation \eqref{gvme variational eq} must be larger than $ \gvme $ at $ \partial \anne $. Alternatively 
		one can simply repeat the proof of \eqref{gpm exp decay} or \eqref{gvm exp decay}, but extend $ \gvme $ outside $ \anne $ 
		to a function which decays to zero. A simple way to achieve that is to construct a function $ \tgvme $ which equals $ \gvme $
		in $ \anne $ and satisfies the inequality
		\beq
			\label{var eq extended}
			-\half \Delta \tgvme + \Omega^2 U(x) \tgvme + 2 \eps^{-2} \tgvme^3 \leq \gvcheme \tgvme,
		\eeq
		 in the region $ \eps^2 \ete \leq |1-x| \leq \eps^2 \ete^2 $, together with Dirichlet boundary conditions 
		 $ \tgvme(1\pm \eps^2 \ete^2) = 0 $. Note that $ \tgvme $ is a priori only a weak subsolution because it is in general 
		only continuous at $ 1 \pm \eps^2 \ete $. However any supersolution to \eqref{var eq extended} in the interval $ |1 - x| \leq \eps^2 \ete^2 $  provides a pointwise upper bound to $ \gvme $ in $ \anne $ and, acting as in the proof of \eqref{gvm exp decay}, it is not difficult to realize that the r.h.s. of \eqref{gvme exp decay} is a supersolution.
		
		The second statement is obtained by combining the application of the maximum principle to the variational equation \eqref{gvme variational eq} far from $ \gvmaxe $ with elliptic estimates like \eqref{linfty est gvm} in the central region around $ \gvmaxe $ where $ \gvme $ is large enough. 
		A direct consequence of the decay estimate \eqref{gvme exp decay} is indeed that one can apply the analysis contained 
		in the proof of Proposition \ref{gvf minimization: pro}. In particular the analogues of \eqref{gve lower bound} - \eqref{linfty est tgvm} hold true, which yield as in the proof of \eqref{linfty est gvm},
		\beq
			\label{linfty est gvme}
			\lf\| \sqrt{2\pi} \Omega^{-1/4} \gvme(x) -  \gosc\lf(\Omega^{1/2}(1-x)\ri) \ri\|_{L^{\infty}(\anne)} \leq C \Omega_0^{-1/4} \eps^{-1}.
		\eeq
		Note that \eqref{gvme exp decay} is needed in order to extend the analysis to the whole real line exactly as in the discussion after \eqref{gve lower bound}.
		The estimate above is already sufficient to obtain \eqref{gvme point est} in the region where $ \gvme $ is large enough, i.e., sufficiently close to $ \gvmaxe $. However far from $ \gvmaxe $ a different approach is needed. We start by proving the following upper bound: For any $ \xv \in \anne $ such that
		\beq
			\label{condition exp small}
			 \lf| 1 - x \ri| \geq (\al \Omega_0)^{-1} \eps^2,
		\eeq
		there exists a constant $ c $ so that
		\beq
			\label{gvme point ub}
			\gvme^2(x) \leq \lf( 1 + c \Omega_0^{-1/4} \ri) \half  \pi^{-3/2} \sqrt{\al \Omega_0} \:  \exp \lf\{ - \alpha \Omega \lf(\lf|1 - x \ri| - 2 (\al \Omega_0)^{-1} \eps^2\ri)^2 \ri\}.
		\eeq
		As above we use an extension $ \tgvme $ of $ \gvme $ to the interval $ |1-x| \leq \eps^2 \ete^2 $ and restrict for simplicity the analysis to the region 
		\beq
			\label{interval ub 2}
			1 - \eps^2 \ete^2 \leq x \leq 1 - \lf(\al\Omega_0\ri)^{-1} \eps^2,
		\eeq
		where we consider the function
		\beq
			\fsup^2(x) : = \lf( 1 + c \Omega_0^{-1/4} \ri) \half \pi^{-3/2} \sqrt{\al \Omega_0} \: \exp\lf\{-  \al \Omega \lf(1 - \lf(\al\Omega_0\ri)^{-1} \eps^2 -  x \ri)^2 \ri\},
		\eeq
		At the boundary of the interval \eqref{interval ub 2} one certainly has $ \tgvme \leq \fsup  $ thanks to \eqref{linfty est gvme}. Moreover by the following estimate of the chemical potential
		\beq
			\label{gvcheme est}
			\gvcheme = \Omega \lf[ \frac{\al}{2} + \frac{1}{\pi} \sqrt{\frac{\al}{2 \pi \Omega_0} } + \OO(\Omega_0^{-3/4}) + \OO(\Omega^{-1/2}) \ri],
		\eeq
		which is a consequence of \eqref{gve asympt}, \eqref{gvee asympt} and \eqref{linfty est gvme}, one obtains 
		\bml{
 			\label{gvme var ineq 2}
 			\lf[ - \half \Delta + \Omega^2 U(x) + 2 \eps^{-2} \fsup^2(x) - \gvcheme \ri] \fsup(x) \geq		\lf[ \half x^{-1} \al\Omega \lf(2x  - 1 + \lf(\al\Omega_0\ri)^{-1} \eps^2 \ri)  + \ri.	\\
			\lf. \al\Omega_0 \eps^{-6} \lf( 1  - \half \lf(\al\Omega_0\ri)^{-1} \eps^2 - x \ri) + 2 \eps^{-2} \fsup^2(x) - \gvcheme  - \OO(1) \ri] \fsup(x) \geq	\\
			\sqrt{\al\Omega_0}\eps^{-4}  \lf[ \sqrt{\al\Omega_0} \eps^{-2} \lf( 1  - \half \lf(\al\Omega_0\ri)^{-1} \eps^2 - x \ri) + \pi^{-3/2} \exp\lf\{ - \al \Omega \lf(1 - \lf(\al\Omega_0\ri)^{-1} \eps^2 -  x \ri)^2 \ri\} - \ri.	\\
			\lf. 2^{-1/2} \pi^{-3/2} \lf( 1 + \OO(\Omega_0^{-1/4}) \ri) \ri] \fsup(x).
		}
		Now we estimate for any  $\delta_1$ such that $ 1 - \delta_1 (\al\Omega)^{-1/2} \leq 1 - \lf(\al\Omega_0\ri)^{-1} \eps^2 $
		\bmln{
 			\sqrt{\al\Omega_0}\eps^{-2} \lf( 1  - \lf(\al\Omega_0\ri)^{-1} \eps^2 - x \ri) + \pi^{-3/2} \exp\bigg\{ - \al \Omega \lf(1 - \lf(\al\Omega_0\ri)^{-1} \eps^2 -  x \ri)^2 \bigg\} \geq 	\\
			\half \lf(\al\Omega_0\ri)^{-1/2} + \pi^{-3/2} - \pi^{-3/2} \al \Omega \lf(1 - \lf(\al\Omega_0\ri)^{-1} \eps^2 -  x \ri)^2 \geq (1 - \delta_1^2) \pi^{-3/2} - C \Omega_0^{-1/2}, 
		}
		where we have used the trivial inequality $ e^{-x} \geq 1 - x $. On the other hand if $ x \leq 1 - \delta_1 (\al\Omega)^{-1/2} $
		\bdm
			\sqrt{\al\Omega_0} \eps^{-2} \lf( 1  - \half \lf(\al\Omega_0\ri)^{-1} \eps^2 - x \ri) \geq \delta_1 \lf( 1 - C \Omega_0^{-1/2}\ri).
		\edm 
		Hence for any  $\delta_1$ such that
		\bdm
			\frac{1}{2\pi^{3}} < \delta_1^2 < 1 - \frac{1}{\sqrt{2}},
		\edm
		(e.g., $ \delta_1 = \half $) and $ \Omega_0 $ sufficiently  large, the r.h.s. of \eqref{gvme var ineq 2} is positive, i.e., $ \fsup $ is a supersolution to \eqref{var eq extended} and \eqref{gvme point ub} follows.
		
		We now focus on the pointwise lower bound and prove that for any $ \xv \in \anns $ such that
		\beq
			\label{condition lb}
			\lf| 1 - x \ri| \geq \delta_2 \lf( \al \Omega \ri)^{-1/2},	\hspace{1cm}	\mbox{with} \:\:\log\sqrt{2} < \delta_2^2 < 1,
		\eeq
		one has the lower bound
		\beq	
			\label{gvme point lb}
			\gvme^2(x) \geq  \lf(1 - \OO(\eps^{\infty})\ri)  \half  \pi^{-3/2} \sqrt{\al \Omega_0}\exp\lf\{ - \al \Omega \lf( 1 - x \ri)^2 \ri\}.
		\eeq
		 Define, for any $ \xv \in \anne $,
		\beq
			\fsub^2(x) : = 
			\begin{cases}
				 \half \pi^{-3/2} \sqrt{\al \Omega_0} \lf[ \exp\lf\{-  \al \Omega \lf(1 -  x \ri)^2 \ri\} -  \exp\lf\{ -  \al \Omega_0 \ete^2 \ri\} \ri],	&	\mbox{for} \:\: |1-x| \geq  \delte,	\\
				 \half \pi^{-3/2} \sqrt{\al \Omega_0} \lf[ \exp\lf\{-  \al \Omega \delte^2 \ri\} -  \exp\lf\{ -  \al \Omega_0 \ete^2 \ri\} \ri],	&	\mbox{otherwise},
			\end{cases}
		\eeq
		where $ \delte : = \delta (\al\Omega)^{-1/2} $ and $ \delta = \OO(1) $ is a suitable parameter. Clearly $ \fsub $ vanishes at $ x = 1 - \eps^2\ete $ and, using \eqref{gvcheme est}, one can show that $ \fsub $ is a subsolution to \eqref{gvme variational eq} in $ \anne $: If $ \lf| 1 - x \ri| \geq \delte $ 
		\bml{
 			\label{var ineq 2}
 			- \half \Delta \fsub + \Omega^2 U \fsub + 2 \eps^{-2} {\fsub}^3 - \gvcheme \fsub \leq		\\
			\lf[ \half \al \Omega + \pi^{-3/2} \sqrt{\al\Omega_0} \exp\lf\{-\al\Omega(1-x)^2\ri\}  - \gvcheme + \OO(\eps^{-2} \ete) \ri] \fsub	\leq	\\
			2^{-1/2} \pi^{-3/2} \sqrt{\al\Omega_0} \eps^{-4} \lf[ - 1 + \sqrt{2} e^{-\delta^2} + C \Omega_0^{-1/4} \ri] \fsub \leq 0
		}
		if $ \delta^2 < \log\sqrt{2} $ and  $ \Omega_0 $ is large enough. On the other hand for $ |1 - x | \leq \delte $
		\bml{
 			- \half \Delta \fsub + \Omega^2 U \fsub + 2 \eps^{-2} {\fsub}^3 - \gvcheme \fsub \leq \lf[ \half \al \Omega \delta^2 + \pi^{-3/2} \sqrt{\al\Omega_0} e^{-\delta^2} - \gvcheme + \OO(\eps^{-2}) \ri] \fsub \leq 	\\
			\half \al \Omega \lf[ \delta^2 - 1 + C \Omega_0^{-1/2} \ri] \fsub \leq 0,
		}
		if $ \delta < 1 $ and $ \Omega_0 $ is large enough. Hence it suffices to take any $ \delta $ satisfying $ \log\sqrt{2} < \delta^2 < 1 $ to obtain that $ \fsub $ is a subsolution and $ \gvme \geq \fsub $ inside $ \anne $. However since for any $ \xv \in \anns $,
		\bdm
			\exp\lf\{\half \al \Omega (1 - x)^2 - \half \al \Omega_0 \ete^2\ri\} \leq \exp\lf\{- \half \al \Omega_0 \ete (\ete - 1) \ri\} = \OO(\eps^{\infty}),
		\edm
		thanks to \eqref{ete cond 1} and \eqref{gvme exp decay}, then \eqref{gvme point lb} is proven.
		
		Collecting \eqref{linfty est gvme}, \eqref{gvme point ub} and \eqref{gvme point lb}, we obtain \eqref{gvme point est}.
	\end{proof}

	Another crucial property of $ \gvme $ is its approximate symmetry w.r.t.\ $ x = 1 $: The potential $ U $ is symmetric w.r.t.\ $ x=1 $ up 
	to small   corrections, so one could expect that the same is true for $ \gvme $. We then introduce another reference profile 
	$ \gsym $ which is defined as the unique positive minimizer of the one-dimensional energy functional\footnote{Note that $ \esymf $ coincides with $ {\mathcal E}^{\mathrm{aux}} $ except for the different integration domain.}
	\beq\label{esymf}
		\esymf[f] : = 2 \pi \int_{-\eps^2\ete}^{\eps^2\ete} \diff x \lf\{ \half (f^{\prime})^2 + \half \al^2 \Omega^2 (1 - x)^2 f^2 + \eps^{-2} f^4 \ri\},
	\eeq
	where $ f \in \domsym $ with
	\beq
		\domsym : = \bigg\{ f \in H^1([1-\eps^2\ete,1+\eps^2\ete]) : \: f = f^*, \: 2\pi  \int_{-\eps^2\ete}^{\eps^2\ete} \diff x \: f^2(x) = 1 \bigg\}.
	\eeq
	We also define
	\beq
		\esym : = \inf_{f \in \domsym} \esym[f] = \esym[\gsym].
	\eeq
	We collect some useful properties of $ \gsym $ in the following

	\begin{pro}[\textbf{Symmetric profile $ \gsym $}]
		\label{sym profile: pro}
		\mbox{}	\\
		The profile $ \gsym $ satisfies the variational equation
		\beq
			\label{gsym variational eq}
			- \half \gsym^{\prime\prime} + \half \al^2 \Omega^2 (1-x)^2 \gsym + 2\eps^{-2} \gsym^3 = \chemsym \gsym
		\eeq
		with Neumann conditions at the boundary $ \gsym^{\prime}(1\pm\eps^2\ete) = 0 $. Moreover $ \gsym $ is symmetric under inversion w.r.t. $ 1 $, i.e., 
		\beq
			\label{gsym sym}
			\gsym(x) = \gsym(2-x),
		\eeq
		for any $ x \in [1-\eps^2\ete, 1+\eps^2\ete] $, and it has a unique maximum point at $ x = 1 $.
	\end{pro}

	\begin{proof}
		Most of the results have already been proven in a slightly different setting  in Prop.\ \ref{gvf minimization: pro} so we skip their discussion. The symmetry of $ \gsym $ can be obtained by exploiting the invariance of the variational equation under inversion or by a rearrangement argument. 
	\end{proof}

	The relation between $ \gvme $ and $ \gsym $ is discussed in the next proposition but we first need to state an estimate of the gradient of $ \gvme $. 
We actually state and prove two different estimates that we use for different purposes: \eqref{gradient gvme} below yields the $L^{\infty}$-bound one should expect from the scales of the problem, whereas \eqref{gradient gvme old} yields the expected rate of exponential decay by giving a 
direct access to the ratio $\gvme^{\prime} \gvme ^{-1}$. 

	\begin{lem}[\textbf{Estimate for $ \gvme^{\prime} $}]
		\label{gradient gvme: pro}
		\mbox{}	\\
		If $ \Omega = \Omega_0 \eps^{-4} $ with $ \Omega_0 > 0 $, then uniformly in $ \xv \in \anne $,
		\beq
			\label{gradient gvme}
			\lf| \gvme^{\prime}(x) \ri| \leq C \left(\alpha \Omega \right) ^{3/2} e^{-C \Omega (1-x)^2},
		\eeq
		\beq
			\label{gradient gvme old}
			\lf| \gvme^{\prime}(x) \ri| \leq C \Omega_0^2 \ete^3 \eps^{-2} \gvme(x).	
		\eeq
	\end{lem} 
	
	\begin{proof}
The proof of \eqref{gradient gvme} is identical to that of \cite[Proposition 2.2]{R1} (actually Step 3 of that proof). Starting from the $L ^{\infty}$-estimate for $\gvme$ we can use the variational equation to bound $\Delta \gvme$. Combining this with the Gagliardo-Nirenberg inequality 
\cite[Theorem 1]{N} we obtain the desired bound.

To prove \eqref{gradient gvme old} we can assume that $ x \leq \gvmaxe $ so that $ \gvme^{\prime} \geq 0 $, since the case $ x \geq \gvmaxe $ is identical. 
By integrating the variational equation \eqref{gvme variational eq} from $ 1-\eps^2\ete $ to $ x $ and using Neumann boundary conditions, 
we get
\bmln{
 - \half \gvme^{\prime}(x) = \int_{1-\eps^2\ete}^x \diff t \lf\{ \half t^{-1} \gvme^{\prime}(t) + \lf[ \gvcheme - 2 \eps^{-2} \gvme^2(t) - \Omega^2 U^2(t) \ri]  \gvme(t) \ri\} \geq	\\
- C \Omega^2 \gvme(x) \int_{1-\eps^2\ete}^x \diff t \: (1-t)^2 \geq - C \Omega^2 \eps^6 \ete^3 \gvme(x),
}
where we have used the monotonicity of $ \gvme $ for $ x \leq \gvmaxe $ and the sup estimate $ \gvme \leq \half \eps^2 \gvcheme $, which is 
the analogue of \eqref{sup est gpm}.
\end{proof}

	\begin{pro}[Estimate for $ \gvme - \gsym $]
		\label{gvme sym est: pro}
		\mbox{}	\\
		If $ \Omega = \Omega_0 \eps^{-4} $ with $ \Omega_0 > 0 $ large enough, 
		one has
		\beq
			\label{gvme sym est}
			\lf| \gvme(x) - \gsym(x) \ri| \leq C \Omega ^{-1/4}
		\eeq
		uniformly in $\xv \in \anne$.
	\end{pro}

	\begin{proof}
Several arguments of this proof  can be taken over from a similar analysis in \cite[Section 2]{R1}. We thus focus on the main new 
features and skip some details for  brevity.

It is more convenient to work with rescaled variables: We introduce 
\begin{eqnarray*}
\gvmet (y) &=& \left(\alpha \Omega \right) ^{-1/4} \gvme \left( \left(\alpha \Omega \right) ^{1/2} (1-x) \right)\\
\gsymt (y) &=& \left(\alpha \Omega \right) ^{-1/4} \gsym \left( \left(\alpha \Omega \right) ^{1/2} (1-x) \right)
\end{eqnarray*}
that are minimizers of the functionals
\begin{eqnarray*}
\gvfet [f]&=& 2\pi \int_{\annet} \diff y \left\{ \half \left| \dd_ y f \right| ^2 + \left(  \half y^2 + \OO \left( (\alpha\Omega)^{-1/2} \right)  \right) f^2 + \left( \alpha \Omega_0 \right) ^{-1/2} f^4 \right\} \left( 1+ \left(\alpha \Omega \right) ^{-1/2} y \right), \\
\esymft [f]&=& 2\pi \int_{\annet} \diff y \left\{  \half \left| \dd_ y f \right| ^2 +  \half y^2 f^2 + \left( \alpha \Omega_0 \right) ^{-1/2} f^4 \right\}, 
\end{eqnarray*}
where the integration is on the domain
\[
 \annet : = [-\ete \alpha ^{-1/2}, \ete \alpha ^{-1/2}].
\]
The quantity $\OO ((\alpha\Omega)^{-1/2})$ in the potential term of $\gvfet$ comes from the Taylor expansion of the original potential.
 This imprecise but concise form is convenient because the precise expression is not relevant. Strictly speaking we should take into account the behavior 
of this term at infinity (since the domain $\annet$ extends to the whole real line in the limit we are considering) but any unwanted increasing 
term is compensated by the exponential fall-off of $\gvmet$, inherited from that of $\gvme$. We do not elaborate more on this point as very 
similar  subtleties have been dealt with in \cite[Sections 2 and 4.1]{R1}. In the sequel we also 
 replace $\annet$ 
by the whole real line. Again, thanks to the fall-off of the functions under consideration, this is harmless.

We have
\beq\label{equation tilde 1}
\left(-\half \dd_y ^2 + \half y^2 + 2 \ell \gsymt ^2 \right) \gsymt = \lamsym \gsymt,
\eeq
where we have set
\[
 \ell := \left( \alpha \Omega_0 \right) ^{-1/2}
\]
and $\lamsym$ is a chemical potential fixed by the normalization, satisfying
\[
 \lamsym = \esymt + \ell \int_{\R} \diff y \: \gsymt ^4.
\]
Because of \eqref{equation tilde 1}, $\gsymt$ is an eigenstate of the Schr\"{o}dinger operator
\begin{equation}\label{hamiltonien tilde}
\hsym := -\half \dd_y ^2 + \half y^2 + 2 \ell \gsymt ^2.
\end{equation}
As it is also positive (being the unique positive minimizer of \eqref{esymf}), it must be the ground state \cite[Corollary 11.9]{LL}.
On the other hand, we have for $\gvmet$
\beq\label{equation tilde 2}
\left(-\half\dd_y ^2 + \half y^2 + 2 \ell \gvmet ^2 \right) \gvmet = \lambda \gvmet + \OO_{L ^2} \left( \left( \alpha \Omega \right) ^{-1/2} \right),
\eeq
where 
\[
 \lambda =  \gveet + \ell \int_{\R} \diff y \: \gvmet ^4 
\]
and the symbol $\OO_{\mathcal{H}} \left( d \right)$ denotes a function whose norm in the Banach space $\mathcal{H}$ is bounded by $d$.
The remainder in \eqref{equation tilde 2} has two origins (further details on a similar problem can be found in \cite[Proof of Theorem 2.1]{R1}):
\begin{itemize}
 \item The remainder in the Taylor expansion of the potential. As already explained, the exponential decay of $\gvme$ allows to control the 
errors produced by the higher order terms in the expansion.
\item The fact that the energy density in $\gvfet$ is integrated against the measure $(1 +( \alpha \Omega ) ^{-1/2} y) \diff y$. 
Again, most terms are dealt with using the fall-off of $\gvmet$. Note that there is also a term of the form $\dd_y \gvmet$ (coming from the 
replacement of the 2D Laplacian of a radial function by a double radial derivative). This one is controlled using Lemma \ref{gradient gvme: pro}.
\end{itemize}
We now write 
\begin{equation}\label{delta tilde}
 h := \gvmet - \gsymt.
\end{equation}
  It is a simple matter to prove that 
\begin{equation}\label{input tilde}
\left\Vert h \right\Vert_{L^2(\R)} \leq C \ell 
\end{equation}
by noting that both $\gvmet$ and $\gsymt$ are close to $\gosc$. Similar arguments have already been used so we omit the details. The rest of the 
proof consist in iteratively improving \eqref{input tilde}.

Taking the difference of equations \eqref{equation tilde 1} and \eqref{equation tilde 2} we see that 
\begin{equation}\label{equation tilde 3}
\left( \hsym - \lamsym \right) h = \left(\lambda - \lamsym \right)\gvmet +2 \ell \gvmet \left( \gsymt ^2 - \gvmet ^2 \right) + \OO_{L ^2} \left( \left( \alpha \Omega \right) ^{-1/2}\right).
\end{equation}
We have 
\[
 \lambda - \lamsym = \gveet - \esymt +\ell \left( \int_{\R} \diff y \: \gvmet ^4 - \int_{\R} \diff y \: \gsymt ^4\right) 
\]
and easy energetic arguments combined with Schwarz inequality yield 
\[
\left|  \lambda - \lamsym  \right| \leq  C \left( \alpha \Omega \right) ^{-1/2} + C \ell\left\Vert h \right\Vert_{L ^2 (\R)}.
\]
Similarly 
\[
 \left\Vert \gvmet \left( \gsymt ^2 - \gvmet ^2 \right) \right\Vert_{L ^2 (\R)} \leq C \ell\left\Vert h \right\Vert_{L ^2 (\R)}.
\]
Thus, multiplying \eqref{equation tilde 3} by $ h $ and integrating, we obtain
\begin{equation}\label{estime tilde}
	\bra{h} \hsym - \lamsym \ket{h} \leq C \left( \alpha \Omega \right) ^{-1/2} + C \ell \left\Vert h \right\Vert_{L ^2 (\R)} ^2,
\end{equation}
which simplifies to 
\[
	\bra{h} \hsym - \lamsym \ket{h} \leq C \max (\left( \alpha \Omega \right) ^{-1/2}, \ell ^3)
\]
by using \eqref{input tilde}. We thus have, combining the above and \eqref{input tilde} and recalling that $\gsymt$ is the ground state of 
$\hsym$ associated with the eigenvalue $\lamsym$
\beq
	\label{norm identity}
 \gvmet = \left( 1+ \OO (\ell)\right) \gsymt  + \OO_{\mathcal{H}^{\mathrm{sym}}} \left( \max \lf[ \left( \alpha \Omega \right) ^{-1/2}, \ell ^3 \ri]\right)
\eeq
where $ \mathcal{H}^{\mathrm{sym}} $ is the Banach space obtained by taking the closure w.r.t. the norm induced by $\hsym$ (which in particular dominates the harmonic oscillator norm), i.e.,
\beq
	\lf\| f \ri\|^2_{\mathcal{H}^{\mathrm{sym}}} : = \bra{f} \hsym \ket{f},
\eeq
and the remainder is orthogonal 
to $\gsymt$. 
Since both $\gsymt$ and $\gvmet$ are also normalized one obtains from \eqref{norm identity} the identity
\bdm
	1 = \lf(1 + \OO(\ell) \ri)^2 + \OO \left( \max \lf[ \left( \alpha \Omega \right) ^{-1}, \ell ^6 \ri] \right)
\edm
which implies that, since $ \ell \ll 1 $,
\beq
	 \OO(\ell) = \OO \left( \max \lf[ \left( \alpha \Omega \right) ^{-1}, \ell ^6 \ri] \right),
\eeq
and therefore $ \ell \leq (\al \Omega)^{-1} $, thus proving from \eqref{norm identity} that
\begin{equation}\label{input tilde 2}
\left\Vert h \right\Vert_{ \mathcal{H}^{\mathrm{sym}}} \leq C  \left( \alpha \Omega \right) ^{-1/2}.
\end{equation}
To obtain \eqref{gvme sym est} it only remains 
\begin{itemize}
\item to recall that we are dealing with 1D functions, thus the $\mathcal{H}^{\mathrm{sym}}$-norm that controls the $H ^1$-norm also controls the $L ^{\infty}$-norm;
\item to scale back the variables.
\end{itemize}
\end{proof}

\subsection{Energy Asymptotics and Consequences}

In this section we prove the energy asymptotics  in the regime $\Omega\sim \eps^{-4}$ and derive some bounds on reduced energies that will lead to the proof of the absence of vortices
in the next section.
As usual the  asymptotics is derived from upper and lower bounds to the energy, the former being quite simple to obtain (see below)
 but the latter requiring more work. 

	The starting point is an energy decoupling  already used in Section 2.2 (in the context of GL and GP theories the technique originates in \cite{LM}). 
	\begin{pro}[\textbf{Reduction to $ \anne $}]
		\label{gpe decoupling: pro}
		\mbox{}	\\
		Defining the function $u\in H^1(\anne)$ by setting for $ \xv \in \anne $
		\beq
			\label{gpm decomposition}
			\gpm(\xv) = : \gvme(x) u(\xv) \exp\lf\{ i \into \vartheta \ri\},
		\eeq
		the following holds
		\beq \label{decoupling pre}
		 \gpe \geq \gvee + \Ee[u],
		\eeq
		with $ \mathbf{B} $ given by \eqref{defi B first} and
		\beq \label{reduced en pre}
		\Ee[u] := \int_{\anne} \diff \xv \: \gvme^2 \lf\{ \half \lf| \nabla u \ri|^2 - \mathbf{B} \cdot \lf(iu, \nabla u\ri) + \eps^{-2} \gvme^2 \lf(1 - |u|^2 \ri)^2 \ri\}.
		\eeq
		Moreover $u$ satisfies the following equation on $\anne$ for some $\lambda \in \R$  
		\begin{equation}\label{equation u}
		  -\nabla \left( \gvme^2 \nabla u \right) - i\gvme ^2 \mathbf{B} \cdot \nabla u + 2\frac{\gvme ^4}{\ep ^2} \left(|u| ^2 - 1 \right) u = \lambda \gvme ^2 u .
		\end{equation}
		
	\end{pro}

	\begin{proof}
		We first use the positivity of the energy density to restrict the integration in the GP functional to the annulus $ \anne $. 
		A standard energy decoupling for which we refer to \cite[Proposition 3.1]{CRY} or \cite[Proposition 4.3]{CPRY1} leads to \eqref{decoupling pre}.
		The variational equation satisfied by $ \gvme $ is used, as well as the Neumann boundary conditions which ensure that there are 
		no boundary terms in the decoupling. 
		The derivation of the equation \eqref{equation u} uses the same ingredients, we refer to \cite[Lemma 4.3]{CRY}.
	\end{proof}

	In order to conclude the proof of Theorem \ref{gv energy: teo}, it suffices to prove a lower bound to $ \Ee[u] $ that
 matches the trivial bound which is obtained by taking $ u = 1 $, i.e., using $ \gvme \exp\{ i \into \vartheta \} $ as a trial function for $ \gpf $. This function has of course to be extended to the whole space and normalized. Using \eqref{gvme exp decay} to bound $ \gvme $ 
outside $ \anne $ it is not difficult to estimate the remainders due to these operations and obtain
	\beq \label{gpe upbound}
		\gpe \leq \gvee + \OO(\eps^{\infty}).
	\eeq

To prove the corresponding lower bound we have to analyze in some details the reduced functional $\Ee[u]$. As is common in the context of GP theory
(see, e.g., \cite{AAB,IM1,CRY,CPRY1}), a crucial step is an integration by parts of the angular momentum term. For technical reasons it will be 
convenient to first replace $ \mathbf{B} $ with 
\begin{equation}
 \label{rmagnp}
	\rmagnp(x) : = \Om \lf(r - r^{-1}\ri) \mathbf{e}_{\vartheta}
\end{equation}
i.e., replace $\into$ with $\Om$. We will also reduce the integration domain to $\anns$ where the crucial $L^{\infty}$-estimate \eqref{gvme point est}
holds true.
These considerations lead to the definition of two potential functions that are `anti-derivatives' of $\gvme^2 \rmagnp$ vanishing at the boundaries
of $\anns$:  
	\beq
		\label{potentials}
		F_1(x) : =\Om \int_{1 - \eps^2 \sqrt{\ete}}^{x} \diff t \: \lf( t - t^{-1}\ri) \gvme^2(t),	
\hspace{1cm}	F_2(x) : = -\Om \int_{x}^{1 + \eps^2 \sqrt{\ete}} \diff t \: \lf( t- t^{-1}\ri) \gvme^2(t).
	\eeq
The use of two potential functions is a novelty of the present paper: When $ x $ is smaller than $ \gvmaxe $ we have $\nablap F_1 = \gvme ^2 \rmagnp$
and will use $F_1$ as potential function, whereas $ F_2 $ will be used for $ x \geq \gvmaxe $. In  a sense this means that we use a potential function that is discontinuous at 
$\gvmaxe$.

 This apparently strange strategy is actually very useful to deal with the boundary term that the integration by parts of the angular
momentum term produces. With these definitions we  shall have to estimate a boundary term on the circle of radius $\gvmaxe$, i.e., in the middle 
of the bulk, where $\gvme$ is large. This is much easier than estimating the term located on a boundary of $\anns$ where the density is small,  but such a term 
would be produced by considering a single potential function as in \cite{AAB,CRY}.   
The necessity of this approach is due to the particular situation under consideration. The density profile $\gvme$ behaves as a 
Gaussian and the long tails of such a function impose the use of a larger domain (which is apparent in the fact that $\anne$ is much thicker 
than the characteristic length of the Gaussian) than in former situations \cite{AAB,CRY,CPRY1}
where the profile was of TF type. This requirement rules out the trick we used in  \cite{CPRY1}, Sect.\ 4.3,  to avoid having 
to estimate any boundary term at all. 

	We now turn to the proofs of two lemmas about the potential functions we just introduced. 
We first estimate the difference between $F_1$ and 
$F_2$ at $\gvmaxe$, in other words the jump of our `discontinuous potential'. This quantity will appear in the boundary term produced by the 
integration by parts.
The second lemma is a pointwise estimate showing that $F_1$ and $F_2$ can be controlled by the density, which is required for the treatment
 of the bulk term.

	\begin{lem}[\textbf{Estimate for $ F_1 - F_2 $}]
		\label{pot est sym: lem}
		\mbox{}	\\
		For any $ \xv \in \anns $ and in particular for $ x = \gvmaxe $
		\beq \label{F1-F2}
			\lf| F_1(x) - F_2(x) \ri| \leq C \Omega_0 \ete ^{3/2}.
		\eeq
	\end{lem}

	\begin{proof}
We compute
\[
F_1(x) - F_2(x) = \Omega \int_{1-\eps^2\sqrt{\ete}}^{1+\eps^2\sqrt{\ete}} \diff t \: \lf( t - t^{-1}\ri) \gvme^2(t)
\]
and note that a Taylor expansion yields 
\[
 t-\frac{1}{t} = 2 (t-1) + \OO (\ep ^4 \ete) 
\]
on the domain we are considering. We thus have
\bml{
 	 F_1(x) - F_2(x) = 2 \Omega \int_{1-\eps^2\sqrt{\ete}}^{1+\eps^2\sqrt{\ete}} \diff t \: (t-1) \gvme^2(t) + \OO( \Omega_0 \ete ) = \\
	 2 \Omega \int_{1-\eps^2\sqrt{\ete}}^{1+\eps^2\sqrt{\ete}} \diff t \: (t-1) \lf(\gvme^2(t)-\gsym ^2 (t) \ri) + \OO( \Omega_0 \ete ),
}
where we have used the $L^2$-normalization of $\gvme$ and the symmetry of $\gsym$ to obtain the first and second error terms respectively. Finally 
\[
 \bigg| \Omega \int_{1-\eps^2\sqrt{\ete}}^{1+\eps^2\sqrt{\ete}} \diff t \: (t-1) \lf(\gvme^2(t)-\gsym ^2 (t) \ri) \bigg| \leq C\Omega_0 \ete ^{3/2}
\]
by using the crucial estimate \eqref{gvme sym est} and the upper bound on $\gvme$ \eqref{gvme sup est}. Recalling that $\ete \gg 1$, \eqref{F1-F2} 
is proved.
\end{proof}
	
	In the following statement we will denote by $ \chi_{\set} $ the characteristic function of  a set $ \set $.

	\begin{lem}[\textbf{Pointwise estimate for $ F_1 $ and $ F_2 $}]\label{lem:F bulk}
		\mbox{}	\\
		If $ \Omega = \Omega_0 \eps^{-4} $ and $ \Omega_0 > \bar{\Omega}_0 $ with $ \bar{\Omega}_0 = \OO(1) $ large enough but independent of $ \eps $, then for any $ \xv \in \anns $
		\beq
			\chi_{[1-\eps^2\sqrt{\ete}, \gvmaxe]}(x) |F_1(x)| +  \chi_{[\gvmaxe, 1+\eps^2\sqrt{\ete}]}(x) |F_2(x)| \leq \al^{-1}  \lf( 1 + C\Omega_0^{-1/4} \ri) \: \gvme^2(x),
		\eeq
		where the constant $C>0$ does not depend on $\xv$.
	\end{lem}

	\begin{proof}
		We can assume without loss of generality that $ x \leq \gvmaxe $ because in the opposite case the proof is essentially the same. 		For any $ \xv \in \anns $, \eqref{gvme point est} yields
		\bml{
 			\lf| F_1(x) \ri| \leq \half  \pi^{-3/2} \sqrt{\al \Omega_0} \Omega \lf( 1 + C\Omega_0^{-1/4} \ri)   \int_{1-\eps^2\sqrt{\ete}}^x \diff t \: (1 - t^2) \exp \lf\{ - \alpha \Omega \lf(1 - t \ri)^2 \ri\} \leq	\\
			 \pi^{-3/2} \sqrt{\al \Omega_0} \Omega \lf( 1 + C\Omega_0^{-1/4} \ri) \int^{\infty}_{1 - x} \diff z \: z \exp\lf\{-\al \Omega z^2 \ri\} \leq	\\
			 \half \pi^{-3/2} \al^{-1/2} \Omega_0^{1/2} \lf( 1 + C\Omega_0^{-1/4} \ri) \exp\lf\{-\al \Omega (1-x)^2 \ri\}  \leq  \al^{-1} \lf( 1 + C\Omega_0^{-1/4} \ri) \: \gvme^2(x).
		}
	\end{proof}

The two above lemmas will be combined with some information about the function $u$. Again, we distinguish the ingredients needed for estimating
the boundary term from those needed for the bulk term. The latter are quite simple actually: The bulk term involves the vorticity measure
\begin{equation}\label{muv}
 \muv:= \curl(iu,\nabla u)= \half i\: \curl \left(u \nabla u^* - u^{*} \nabla u \right),
\end{equation}
a quantity that is easily related to the kinetic energy density $|\nabla u| ^2$. Note that in our setting we do not need to  apply tools 
involving a construction of vortex balls to obtain  a useful estimate.

\begin{lem}[\textbf{Bound on the vorticity measure}]\label{lem:vorticity}\mbox{}\\
We have
\begin{equation}\label{vortic bound}
|\muv| \leq |\nabla u| ^2 
\end{equation}
pointwise in $\anne$.
\end{lem}

\begin{proof}
A straightforward computation shows that
\[
 \muv = - 2 \mathrm{Im} \left[ \dd_1 u \: \lf( \dd_2 u \ri)^* \right]
\]
where $\dd_i = \dd_{x_i}$ for $i=1,2$ and $\mathrm{Im}$ stands for the imaginary part. On the other hand
\[
 \left| \dd_1 u \: \lf( \dd_2 u \ri)^* \right| \leq \half \left( \left| \dd_1 u \right| ^2 +\left| \dd_2 u \right| ^2 \right),
\]
which completes the proof.
\end{proof}

The estimate of the boundary term will  rely on the fact that the degree of $u$ around the circle of 
radius $\gvmaxe$ is in  a certain sense controlled by the kinetic energy. Note that we do not know yet that $u$ does not vanish on this circle, so 
\eqref{degree estimate} below is not truly a degree estimate. Note also that, as will be clear from the proof, the estimate holds because 
the circle we are interested in lies in a region where $\gvme$ is large enough. This is the main motivation for the 
trick
of choosing a potential function discontinuous along this particular circle. If the potential function was continuous there, we would 
have to estimate the degree on a circle were $\gvme$ is small (namely on one of the boundaries of $\anns$). This would be much 
more difficult, if not impossible.

\begin{lem}[\textbf{Circulation estimate}]\label{lem:degree}\mbox{}\\
Let $R$ be a radius satisfying
\begin{equation}\label{radius degree 2}
 R = 1 +\OO (\Om ^{-1/2}).
\end{equation}
We have, for some constant $C>0$ and any $\delta >0$,
\begin{equation}\label{degree estimate}
 \left| \int_{\dd B_{R}} \diff \sigma \: (iu,\dd_{\tau} u) \right| \leq C \bigg[ \left( \Om ^{-1/2} + \delta \right) \int_{\anns} \diff \xv \: \gvme ^2 |\nabla u| ^2 + \delta^{-1} \bigg]
\end{equation}
\end{lem}

\begin{proof}
We use a smooth radial cut-off function $\chi$ with support in $[\tilde{R}, R ]$, for some radius 
\begin{equation}\label{Rtilde}
\tilde{R} = R - c \Om ^{-1/2}
\end{equation}
with $ c > 0 $.
Obviously one can impose
\begin{eqnarray}\label{cutoff chi}
 \chi (R) &=& 1, \nonumber\\
\chi(\tilde{R}) &=& 0,  \nonumber\\
|\chi| &\leq& 1, \nonumber \\
| \nabla \chi| &\leq& c^{-1} \Om ^{1/2}.
\end{eqnarray}
We then use Stokes' formula to obtain
\[
 \int_{\dd B_{R}} \diff \sigma \: (iu,\dd_{\tau} u) = \int_{\dd B_{R}} \diff \sigma  \: \chi (iu,\dd_{\tau} u) = - \int_{\R ^2} \diff \xv \: \nabla ^{\perp} \chi \cdot (iu,\nabla u)
+ \int_{\R ^2}  \diff \xv \: \chi \muv,
\]
where the integrals on $\R^2$ are actually reduced to the support of $\chi$. If we choose $c$ small enough in \eqref{Rtilde}, a combination of 
\eqref{gvme point est} and \eqref{radius degree 2} shows that 
\[
 \gvme ^2 \geq C \Om ^{1/2}
\]
in this region. Indeed, because of the assumptions \eqref{radius degree 2} and \eqref{Rtilde}, the region of interest is located at a distance of at 
most $\OO(\Om ^{-1/2})$ from $1$, i.e., in the middle of the bulk (see, e.g., Remark 3.1).
We thus estimate, using also \eqref{vortic bound} and \eqref{cutoff chi}, 
\bml{
  \left| \int_{\dd B_{R}}  \diff \sigma \: (iu,\dd_{\tau} u) \right| \leq  C \Om ^{1/2}  \int_{\tilde{R} \leq x \leq R} \diff \xv \: |u||\nabla u|
+ \int_{\tilde{R} \leq x \leq R} \diff \xv \: |\muv| \leq \\
C \int_{\anns} \diff \xv \: \gvme ^2 |u||\nabla u|+ C\Om ^{-1/2} \int_{\anns} \diff \xv \:\gvme ^2 |\nabla u| ^2 \leq \\
C  \delta \int_{\anns} \diff \xv \:\gvme ^2 |\nabla u| ^2 + C  \delta^{-1} \int_{\anns} \diff \xv \: \gvme ^2 |u| ^2+ C\Om ^{-1/2} \int_{\anns} \diff \xv \: \gvme ^2 |\nabla u| ^2.
}
There only remains to recall that $\gvme  |u| = |\gpm| $ is normalized in $L^2 (\R ^2)$ to complete the proof.
\end{proof}


We are now able to conclude the proof of the energy lower bound. A corollary of the proof below is the following useful estimate that will 
allow to complete the proof of the absence of vortices in the next subsection.

\begin{cor}[\textbf{Estimates for reduced energies}]\label{cor: red energy upbound}\mbox{}\\
Denoting
\begin{equation}\label{red red energy}
\Fe[u]:=\int_{\anne} \diff \xv \: \gvme^2 \lf\{ \half \lf| \nabla u \ri|^2 + \eps^{-2} \gvme^2 \lf(1 - |u|^2 \ri)^2 \ri\},
\end{equation}
one has for $\Om_0$ larger than some $\bar{\Om}_0 = \OO(1) $ 
\begin{equation}\label{red energy upbound}
\left| \Ee[u] \right|+ \Fe[u] \leq C \Omega_0 ^2 \ete ^{3}
\end{equation}
for some finite constant $C < \infty$. 
\end{cor}

\begin{proof}[Proof of Theorem \ref{gv energy: teo} and Corollary \ref{cor: red energy upbound}]
The starting point is the energy decoupling of Proposition \ref{gpe decoupling: pro} and the trivial upper bound \eqref{gpe upbound}. We now bound from below the reduced energy $ \Ee[u] $.

As announced we first replace $ \mathbf{B} $ with $\rmagnp$. The remainder produced by this operation is estimated as follows:
\bml{
 \bigg| \int_{\anne} \diff \xv \: \gvme ^2\left( \Om - \into \right) x^{-1} \mathbf{e}_{\vartheta} \cdot (iu,\nabla u) 
\bigg| \leq  \half C \beta \int_{\anne} \diff \xv \: \gvme ^2 |u | ^2 +  \half C \beta^{-1}
 \int_{\anne} \diff \xv \: \gvme ^2 |\nabla u |^2 \leq	\\
	 \half C \beta^{-1} \int_{\anne} \diff \xv \: \gvme ^2 |\nabla u |^2 + \half C \beta,
}
where we have used the normalization of $\gpm$ and $\beta > 0 $ is a constant to be fixed later in the proof.
We now reduce the integration domain to $\anns$ at the price of a second remainder (for a lower bound we neglect the positive terms
in the energy density)
\begin{multline*}
 \bigg| \int_{\anne \setminus \anns }  \diff \xv \: \gvme ^2 \rmagnp \cdot (iu,\nabla u) \bigg| \leq \half \beta \int_{\anne \setminus \anns}  \diff \xv \:
 \gvme ^2 |\rmagnp | ^2|u | ^2   \\ + \half \beta^{-1} \int_{\anne\setminus \anns}  \diff \xv \:\gvme ^2 |\nabla u |^2 
\leq  \half \beta^{-1} \int_{\anne \setminus \anns}  \diff \xv \: \gvme ^2 |\nabla u |^2 + \OO (\ep ^{\infty}).
\end{multline*}
Here we have used the exponential smallness of $|\gpm| = \gvme |u|$ outside $\anns$. This requires
\[
 \ete \gg |\log \ep|,
\]
which we are free to decide. 
After these two steps we have
\begin{multline}\label{lowbound 1}
\gpe \geq \gvee +  \int_{\anns} \diff \xv \: \gvme^2 \lf\{ \half \lf| \nabla u \ri|^2 - \Omega \rmagnp \cdot \lf(iu, \nabla u\ri) + \eps^{-2} \gvme^2 \lf(1 - |u|^2 \ri)^2 \ri\}
\\ +  \int_{\anne \setminus \anns} \diff \xv \: \gvme^2 \lf\{\half \lf| \nabla u \ri|^2 + \eps^{-2} \gvme^2 \lf(1 - |u|^2 \ri)^2 \ri\}-  \half C \beta^{-1} \int_{\anne} \gvme ^2 |\nabla u| ^2 - \half \beta -  \OO (\ep ^{\infty}).
\end{multline} 

We now turn to the main part of the proof, namely a lower bound to
\begin{equation}\label{recall red ener}
\Es[u] := \int_{\anns} \diff \xv \: \gvme^2 \lf\{ \half \lf| \nabla u \ri|^2 - \Omega \rmagnp \cdot \lf(iu, \nabla u\ri) + \eps^{-2} \gvme^2 \lf(1 - |u|^2 \ri)^2 \ri\}.
\end{equation} 
As a first step we  perform the integration by parts we have been alluding to. Since
\begin{eqnarray*}
\nabla^{\perp} F_1 & = &  \Omega \gvme^2 \rmagnp,  \hspace{1cm} \mbox{ if } \:\: 1-\ep^2 \sqrt{\ete} \leq x \leq \gvmaxe, \\
\nabla^{\perp} F_2 & = &  \Omega \gvme^2  \rmagnp,  \hspace{1cm} \mbox{ if } \:\: \gvmaxe \leq x \leq 1+\ep ^2 \sqrt{\ete}, 
\end{eqnarray*}
we have
\bml{\label{proof energy}
 \int_{\anns} \diff \xv \: \gvme^2 \lf\{ \half \lf| \nabla u \ri|^2 - \Omega \rmagnp \cdot \lf(iu, \nabla u\ri) \ri\} =	\\
\int_{\anns} \diff \xv \: \half \gvme^2 \lf| \nabla u \ri|^2 + \int_{1 - \eps^2 \sqrt{\ete} \leq x \leq \gvmaxe} \diff \xv \: F_1(x) \muv +  \int_{\gvmaxe \leq x \leq 1 +\ep ^2 \sqrt{\ete}} \diff \xv \: F_2(x) \muv 	\\
- \lf[ F_1(\gvmaxe) - F_2(\gvmaxe) \ri] \int_{\partial \ba_{\gvmaxe}} \diff \sigma \:  \lf(iu, \partial_{\tau} u\ri) , 
}
where $\muv$ is the vorticity measure defined in \eqref{muv}. We then combine Lemmas \ref{lem:F bulk} and \ref{lem:vorticity} to control 
the bulk terms produced by the integration by parts when $\Om_0 \geq \bar{\Om}_0$, the constant appearing in Lemma \ref{lem:F bulk}:
\begin{multline*}
 \int_{\anns} \diff \xv \: \half \gvme^2 \lf| \nabla u \ri|^2 + \int_{1 - \eps^2 \sqrt{\ete} \leq x \leq \gvmaxe} \diff \xv \: F_1(x) \muv 
+  \int_{\gvmaxe \leq x \leq 1 +\ep ^2 \sqrt{\ete}} \diff \xv \: F_2(x) \muv \geq \\ \left(\half - \al^{-1}\lf(1+\OO(\Om_0 ^{-1/4})\ri)\right)\int_{\anns} \diff \xv \:  \gvme^2 \lf| \nabla u \ri|^2.
\end{multline*}
For the boundary term we use Lemmas \ref{pot est sym: lem} and \ref{lem:degree} to infer
\bml{
 \lf| F_1(\gvmaxe) - F_2(\gvmaxe) \ri| \bigg|\int_{\partial \ba_{\gvmaxe}} \diff \sigma \:  \lf(iu, \partial_{\tau} u\ri) \bigg| \leq C \Omega_0 \ete ^{3/2} \left| \int_{\partial \ba_{\gvmaxe}} \diff \sigma \:  \lf(iu, \partial_{\tau} u\ri) \right| \leq \\
 C \Omega_0 \ete ^{3/2} \delta^{-1} +C \Omega_0 \ete ^{3/2} \left( \Om ^{-1/2} + \delta\right) \int_{\anns} \diff \xv \: \gvme^2 \lf| \nabla u \ri|^2.
}
We thus obtain 
\bmln{
\int_{\anns} \diff \xv \: \gvme^2 \lf\{ \half \lf| \nabla u \ri|^2 - \Omega \rmagnp \cdot \lf(iu, \nabla u\ri) \ri\}	\geq  \\
\left(\half - C \Omega_0 \ete ^{3/2} \delta - \alpha^{-1} \lf(1+\OO(\Om_0 ^{-1/4})\ri)\right)\int_{\anns} \diff \xv \: \gvme^2 \lf| \nabla u \ri|^2 
- C \Omega_0 \ete ^{3/2} \delta^{-1}.
}
We now recall that $\alpha > 2$, which is obvious from its definition \eqref{alpha}. It is thus sufficient to choose
\[
\delta = C \Omega_0 ^{-1} \ete ^{-3/2}
\]
with a sufficiently small constant $C$ to deduce from the above that, for some $C>0$,
\[
\int_{\anns} \diff \xv \: \gvme^2 \lf\{ \half \lf| \nabla u \ri|^2 - \Omega \rmagnp \cdot \lf(iu, \nabla u\ri) \ri\}\geq  
C \int_{\anns} \diff \xv \: \gvme^2 \lf| \nabla u \ri|^2 -\OO( \Omega_0 ^2 \ete ^{3}).
\]
Going back to \eqref{proof energy} we have
\[
\Es[u] \geq C \int_{\anns} \diff \xv \: \gvme^2 \lf| \nabla u \ri|^2 + \int_{\anns}\diff \xv\: \frac{\gvme^4}{\ep ^2} \left(1 -|u| ^2\right) ^2
-\OO( \Omega_0 ^2 \ete ^{3}). 
\]
Combining this with \eqref{gpe upbound} and \eqref{lowbound 1} and taking 
\[
\beta = \OO(1)
\]
large enough, we obtain for $\Om_0 > \bar{\Om}_0$
\[
\gvee + \OO (\ep ^{\infty}) \geq \gpe \geq \gvee + C \int_{\anne} \diff \xv \: \gvme^2 \lf| \nabla u \ri|^2 + \int_{\anne}\diff \xv\: \frac{\gvme^4}{\ep ^2} \left(1 -|u| ^2\right) ^2
- \OO (\ep ^{\infty})- \OO( \Omega_0 ^2 \ete ^{3}),
\]
 if again $\ete \gg |\log \ep| $. For definiteness we fix 
\beq
	\label{ete choice}
 \ete = |\log \ep| ^{3/2}
\eeq
in the sequel, a choice that indeed satisfies the requirements \eqref{ete cond 1}.
Neglecting the positive terms in the lower bound concludes the proof of 
Theorem \ref{gv energy: teo}, whereas keeping them proves Corollary \ref{cor: red energy upbound}. 
\end{proof}

\subsection{Transition to the Giant Vortex State}

The proof of the absence of vortices requires one more ingredient, namely an estimate on the $L ^{\infty}$-norm of the gradient of $u$. It is 
proved in much the same way than the corresponding results in \cite{CRY,CPRY1} and 
 we shall be brief on the details.

\begin{lem}[\textbf{Gradient estimate for u}]\label{lem:grad est u}\mbox{}\\
For any $ c >0 $, let $\Ac$ be the domain defined  as
\begin{equation}\label{defi D}
 \Ac := 
 \lf\{\xv:\, 1 - c |\log \ep| ^{1/2} \Om ^{-1/2} \leq x \leq 1+ c |\log \ep|^{1/2} \Om ^{-1/2} \ri\}.
\end{equation}
If the constant $c>0$ is chosen small enough in the above, there exists a $C>0$ such that
\begin{equation}\label{grad u}
\left\Vert \nabla u \right\Vert_{L ^{\infty} (\Ac)} \leq C \ep ^{\frac{c ^2 \alpha}{2} -2}.
\end{equation}
\end{lem}

\begin{proof}
The proof uses the equation  \eqref{equation u} and the Gagliardo-Nirenberg inequality \cite[Theorem 1]{N}
\[
 \left\Vert \nabla u \right\Vert_{L ^{\infty}} \leq C \left(\left\Vert u \right\Vert_{L ^{\infty}} ^{1/2} \left\Vert \Delta u \right\Vert_{L ^{\infty}} ^{1/2} + \left\Vert u \right\Vert_{L ^{\infty}}\right)
\]
as in \cite[Lemma 5.1]{CRY}, which leads to the following bound (we omit the details)
\begin{multline}\label{grad u proof}
\left\Vert \nabla u \right\Vert_{L ^{\infty} (\D)} \leq C \ep^{-1} \left( \left\Vert \gpm \right\Vert_{L ^{\infty} (\D)} + 
\left\Vert \gpm \right\Vert_{L ^{\infty} (\D)} ^2\left\Vert \gvme ^{-1}\right\Vert_{L ^{\infty} (\D)}\right)  \\
+ C\left\Vert \gpm \right\Vert_{L ^{\infty} (\D)}\left\Vert \gvme ^{-1} \right\Vert_{L ^{\infty} (\D)} \left(|\lambda| ^{1/2} + \left\Vert \gvme^{-1} \nabla \gvme \right\Vert_{L ^{\infty} (\D)} 
+ \left\Vert\mathbf{B} \right\Vert_{L ^{\infty} (\D)}\right)
\end{multline}
for any domain $\D \subset \anne$. We will reduce to the domain  \eqref{defi D} and bound the terms on the right-hand side of the above 
equation.
The most stringent requirement is the validity of a proper upper bound on $\gvme ^{-1}$,  that requires a lower bound for $\gvme$. Using \eqref{gvme point est} 
 one can see that we have
\[
 \gvme \geq C \ep ^{\frac{c ^2 \alpha}{2} - 1}
\]
on $\Ac$ provided $c$ is small enough. Note that with our choice \eqref{ete choice}, i.e., $\ete = |\log \ep| ^{3/2}$, $\Ac \subset \anns$, so one can use  \eqref{gvme point est} on $\Ac$.
For the sup estimate on $\gpm$ we use \eqref{gpm exp small gv}, whereas Lemma \ref{gradient gvme: pro} yields
\[
\left\Vert  \gvme^{-1} \nabla \gvme\right\Vert_{L ^{\infty} (\Ac)} \leq C  \Om_0 ^2 \eps^{-2}|\log \ep| ^3.  
\]
On $\Ac$  the vector potential $\mathbf{B}$ is easily seen to satisfy 
\[
 \left\Vert \mathbf{B} \right\Vert_{L ^{\infty} (\Ac)} \leq C \Om_0\eps^{-2}|\log \ep|.
\]
Finally, adapting the proof of \cite[Equation (4.28)]{CRY}, we find
\[
 |\lambda| \leq \lf| \Ee[u] \ri| + C \ep ^{-3} \Fe[u] ^{1/2} \leq C \Omega_0 \ete ^{3/2} \ep ^{-3},
\]
where the second inequality is a consequence of Corollary \ref{cor: red energy upbound} and thus uses the assumption that $\Om_0$ is large enough. 
Plugging these bounds in \eqref{grad u proof} we obtain the result.
\end{proof}

We emphasize that on the domain $\Ac$ defined in \eqref{defi D} in the Lemma we have the estimate
\begin{equation}\label{g lowbound D}
\gvme \geq C \ep ^{\frac{c ^2 \alpha}{2}-1}. 
\end{equation}

We are now able to conclude the proof of the absence of vortices. The idea that we use, namely the combination of bounds like 
those in Corollary \ref{cor: red energy upbound} and a gradient estimate, was first introduced in \cite{BBH}. We recall the definition of 
 $\ab$:
\[
 \ab = \lf\{ \xv \in \R ^2: \: 1- c  |\log \ep | ^{1/2} \Om ^{-1/2}\leq x \leq 1+ c |\log \ep | ^{1/2} \Om ^{-1/2}\ri\}, \hspace{1cm} c < \lf( \frac{2}{\al }\ri) ^{1/2}.
\]
\begin{proof}[Proof of Theorem \ref{teo: giant vortex}]
The argument is by contradiction. Let us assume that $|1-|u(\xv)|| \geq |\log \ep| ^{-a}$ at some point $ \xv \in \ab$ for some power $a>0$. 
Then \eqref{grad u} means that 
\[
|1-|u||\geq \half |\log \ep| ^{-a},	\hspace{1cm} \mbox{ on } \:\: \ba\lf(\xv,C\ep^{2-\frac{c ^2 \alpha}{2}} |\log\eps|^{-a}\ri)\cap \ab. 
\]
This implies
\[
 \int_{\anns} \diff \xv \: \frac{\gvme^4}{\ep ^2}\left(1-|u| ^2 \right) ^2 \geq C \ep ^{c ^2 \alpha-2} |\log \ep| ^{-4a}.
\]
This is a contradiction with \eqref{red energy upbound} because of the choices 
\[
 \ete = |\log \ep| ^{3/2},	\hspace{1cm} c < \lf( \frac{2}{\al }\ri) ^{1/2}.
\] 
Indeed, we would conclude that a negative power of $\ep$ is controlled by a power of $|\log \ep|$. We thus obtain $|1-|u|| \leq |\log \ep| ^{-a}$
 in $\ab$ for any power $a>0$, and Theorem \ref{teo: giant vortex} is proved.   
\end{proof}

The estimate of the degree of $\gpm$ is a consequence of Lemma \ref{lem:degree} and of the energy estimates of Corollary \ref{cor: red energy upbound}. 

\begin{proof}[Proof of Theorem \ref{teo: gv degree}]
Taking a radius $R$ satisfying \eqref{radius degree} we first note that the pointwise estimate in \eqref{gv point diff} implies that $\gpm$ does not vanish on $\dd B_{R} $, so that its degree is indeed well defined. We then compute
\begin{eqnarray}\label{compute degree}
2 \pi \deg (\gpm, \dd B_{R}) &=& - i \int_{\dd  B_R} \diff \sigma \: \frac{|\gpm|}{\gpm} \dd_{\tau} \left( \frac{\gpm}{|\gpm|} \right) = - i \int_{\dd B_R} 
\diff \sigma \: \frac{|u|}{u} \dd_{\tau} \left( \frac{u}{|u|} e^{i\lfloor\Om\rfloor\vartheta}\right) e^{-i\lfloor\Om\rfloor\vartheta} \nonumber
\\ &=& 2\pi \lfloor\Om\rfloor - i \int_{\dd B_R} \diff \sigma \:  \frac{|u|}{u} \dd_{\tau} \left( \frac{u}{|u|}\right),
\end{eqnarray}
the second term on the second line being the degree of $u$. Then 
\beq\label{error degree}
\bigg| \int_{\dd B_R} \diff \sigma \: \frac{|u|}{u} \dd_{\tau} \left( \frac{u}{|u|}\right) \bigg|\leq \left| \int_{\dd B_R} \diff \sigma \:  (iu,\dd_{\tau} u ) \right| 
\eeq
where we have used that $|u|$ is bounded above and below by a constant on $\dd B_R $. It remains to combine \eqref{degree estimate} 
and \eqref{red energy upbound}  (where we stick to the choice \eqref{ete choice} $\ete = |\log \ep| ^{3/2}$ for definiteness and optimize over $\delta$)
 and the result is proved.
\end{proof}

\section{Rotational Symmetry Breaking}

We complete in this section the proof of Theorem \ref{symmetry break: teo}. The strategy of the proof is identical to one followed in \cite[Section 5]{CPRY1}, so we often omit technical details.

A symmetric vortex minimizer is a function of the form $ f_n(x) \exp\{ i n \vartheta \} $ where $ f_n $ is real and $ n \in \Z $. For any given $ n $ there exists a minimizing function $ f_n $ with energy $ E_n $ but we are here interested in  a global minimizer, so we denote by $ \bar{n} $  a minimizer of $ E_n $ w.r.t. $ n \in \Z $ 
.
The associated minimizing profile $ f_{\bar{n}}(x) $ is  
radial and positive far from the origin. Moreover it has a unique maximum at $ \gvmaxe $ and is increasing for $ 0 \leq x \leq \gvmaxe $ and decreasing everywhere else. 

As in \cite[Proof of Theorem 1.6]{CPRY1} the result is proven by a direct inspection of the second variation of the GP functional around a local minimizer, which in the case of a symmetric vortex  $ f_n(x) \exp\{ i n \vartheta \} $ becomes
\bml{
	\label{Q form}
	\Q_n[\Xi] : = \int_{\R^2} \diff \xv \lf\{ \half \lf| \lf( \nabla - i \aavoo \ri) \Xi \ri|^2 + \gamma \Omega^2 W |\Xi|^2+ 4 \eps^{-2} f_n^2 |\Xi|^2 - \mu_n |\Xi|^2  \ri\} +	\\
	2 \eps^{-2} \Re \int_{\R^2} \diff \xv \: f_n^2 \Xi^2 \exp\lf\{-2 i n \vartheta \ri\},
}
where $ \Xi \in H^1(\R^2) $ and $ \mu_n $ is the chemical potential associated with $ f_n $, i.e., $ \mu_n = E_n + \eps^{-2} \| f_n \|_4^4 $. 

\begin{proof}[Proof of Theorem \ref{symmetry break: teo}]
	We assume that the symmetric vortex with degree $ \bar{n} $ is a global minimizer of the GP functional $ \gpfoo $ and then show that this yields a contradiction since one can find a function $ \Xi $ which makes the quadratic form $ \Q_{\bar{n}} $ negative.
	
	However we start by certain useful properties of $ f_{\bar{n}} $. Since we have assume that $ \gpm = f_{\bar{n}}(x) \exp\{i \bar{n} \vartheta\} $, the whole analysis contained in Section \ref{sec: GV GP} applies to $ f_{\bar{n}} $ too and a simple inspection shows that it can be extended to any $ \Omega \gtrsim \eps^{-4} $. In particular one can prove the estimates
	\beq
		\gpe \leq \OO(\Omega),		\hspace{1cm}	\lf\| f_{\bar{n}} \ri\|_{\infty}^{2} \leq \OO(\eps^2 \Omega),
		\hspace{1cm}
		\lf| f_{\bar{n}}(x) \ri|^2 \leq C \eps^{-2} \exp \lf\{- \Omega^{1/2} \lf|1 - x \ri| \ri\},
	\eeq
	for any $ \xv \in \R^2 $, which imply that $ f_{\bar{n}} $ is concentrated in the annulus 
	\beq
		\A_{\bar{n}} : = \lf\{ \xv \in \R^2 : \: 1 - \Omega^{-1/2}|\log\eps|^2 \leq x \leq 1 + \Omega^{1/2}|\log\eps|^2] \ri\},
	\eeq
	i.e., $ \| f_{\bar{n}} \|_{L^2(\A_{\bar{n}})} = 1 - o(1) $. The energy estimates \eqref{gvee asympt} and \eqref{gpe upbound} can as well be extended to angular velocities much larger than $ \eps^{-4} $, implying
	\bdm
 		\OO(\eps^{\infty}) \geq \half \int_{\A_{\bar{n}}} \diff \xv \lf( \bar{n} x^{-1} - \Omega x \ri)^2  f^2_{\bar{n}} \geq (1 - o(1)) \lf( \bar{n} - \Omega + \OO(\Omega^{1/2} |\log\eps|^2) \ri)^2,
	\edm
	which yields  $ \bar n = \Omega (1 + o(1)) $.

	Now we can introduce the test function which coincides with the one used in \cite[Proof of Theorem 1.6]{CPRY1}, i.e.,
	\beq
		\Xi(\xv) : = (A(x) + B(x)) e^{i(\bar n+d)\vartheta} + (A(x) - B(x)) e^{i(\bar n-d)\vartheta},
	\eeq
	\beq
		A(x) : = 
		\begin{cases}
			x^{d+1} f^{\prime}_{\bar{n}},	&	\mbox{for} \:\: 0 \leq x \leq \gvmaxe,	\\
			0,						&	\mbox{otherwise},
		\end{cases}
		\hspace{1cm}
		B(x) : = 
		\begin{cases}
			\bar n x^{d} f_{\bar{n}},		&	\mbox{for} \:\: 0 \leq x \leq \gvmaxe,	\\
			\bar n \gvmaxe^{d} f_{\bar{n}},	&	\mbox{otherwise}.
		\end{cases}
	\eeq
	A 
somewhat lengthy 
	computation yields
	\bml{
 		\Q_{\bar{n}}[\Xi] = 4 \pi \int_0^{\gvmaxe} \diff x \: x^{2d+2} f_{\bar{n}} f^{\prime}_{\bar{n}} \lf\{ 2(d+1) \mu_{\bar n} + 2 \Omega \bar n - (d+2) \Omega^2 x^2 - 2 \gamma (d+1) \Omega^2 W(x) +	\ri.	\\
		\lf. \gamma \Omega^2 x W^{\prime}(x) - 4 (d+1) \eps^{-2} f_{\bar n}^2 \ri\} + 2 \pi {\bar n}^2 d^2 \gvmaxe^{2d} \int_{\gvmaxe}^{\infty} \diff x \: x^{-1} f_{\bar n}^2.
	}
	Now using the facts
	\beq
		\bar n = \Omega (1 + o(1)), 	\hspace{1cm}	\gvmaxe = 1 + o(1),	\hspace{1cm} W(x) = o(1),	\hspace{1cm}	W^{\prime}(x) = o(1),
	\eeq
	for any $ \xv \in \A_{\bar{n}} $, as well as the monotonicity of $ f_{\bar n} $ for $ x \leq \gvmaxe $, we obtain
	\bml{
 		\Q_{\bar{n}}[\Xi] \leq  - 4 \pi \Omega^2 d (1 - o(1)) \int_0^{\gvmaxe} \diff x \: x^{2d+2} f_{\bar{n}} f^{\prime}_{\bar{n}} + C \Omega^2 d^2 \leq	\\
		- 2 \pi \Omega^2 d \lf[ (1 - o(1)) f^2_{\bar n}(\gvmaxe)  - C d \ri] < 0,
	}
	for $ \eps $ sufficiently small and any finite $ d > 0 $ since $  f^2_{\bar n}(\gvmaxe) = \OO(\eps^2\Omega) \gg 1 $.
\end{proof}

\bigskip

{\small {\bf Acknowledgements.} MC and NR acknowledge the hospitality of the Erwin Schr\"odinger Institute for Mathematical Physics in Vienna where part of this work was carried out. The work of NR was supported by the European Research Council under the European Community Seventh Framework Programme (FP7/2007-2013 Grant Agreement MNIQS no. 258023), the work of MC by the same programme under the Grant Agreement CoMBoS no. 239694.}

\end{document}